\pdfoutput=1 

\documentclass[11pt]{article}
\usepackage[left=1in, right=1in, top=1in, bottom=1in]{geometry}

\usepackage[utf8]{inputenc}

\usepackage{t1enc}
\usepackage{amsmath}
\usepackage{amsthm}
\usepackage{latexsym}
\usepackage{natbib}
\usepackage{epsfig}
\usepackage{textcomp}
\usepackage{alltt}
\usepackage{graphicx}
\usepackage{rotating}
\usepackage{dcolumn}
\usepackage{enumitem}
\usepackage{amsfonts}
\usepackage{algorithm}
\usepackage{amssymb}
\usepackage{bbm, dsfont}
\usepackage[svgnames]{xcolor}
\usepackage{mathtools}
\usepackage{listings}
\usepackage[mathscr]{euscript}

\usepackage{booktabs}
\usepackage{siunitx}
\usepackage{adjustbox}
\usepackage{multirow}

\usepackage{tabularx}
\usepackage{booktabs}

\definecolor{Bleu}{RGB}{0,0,175} 

\usepackage[hypertexnames=false]{hyperref}
\hypersetup{
colorlinks,
  citecolor=Bleu,
  linkcolor=Bleu,
  urlcolor=Bleu}

\definecolor{CBmagenta}{HTML}{EE3377}
\definecolor{CBred}{HTML}{CC3311}
\definecolor{CBorange}{HTML}{EE7733}
\definecolor{CBteal}{HTML}{009988}
\definecolor{CBcyan}{HTML}{33BBEE}
\definecolor{CBblue}{HTML}{0077BB}
\definecolor{CBgrey}{HTML}{BBBBBB}

\definecolor{CB5red}{HTML}{DA0000} 
\definecolor{CB5blue}{HTML}{0066FF} 

\usepackage[protrusion=false]{microtype} 

\usepackage[medium]{titlesec}

\usepackage{titletoc}
\newcommand\DoToC{%
  \startcontents
  \printcontents{}{1}{\vskip 1.5em\hrule\vskip .75em}
  \vskip .75em\hrule\vskip 2em
}

\titlespacing*{\section} {0pt}{3.5ex plus 1ex minus .2ex}{*1} 

\makeatletter
\toks@\expandafter{\@endtheorem\@endpetrue}
\edef\@endtheorem{\the\toks@}
\makeatother

\titlespacing*{\paragraph}{0pt}{1ex}{0.75ex}

\usepackage{setspace}
\newcommand{\slightspacing}{\setstretch{1.15}}
\slightspacing

\usepackage{authblk}

\usepackage[labelfont=bf,font=small,singlelinecheck=false]{caption}
\usepackage[singlelinecheck=true]{subcaption}

\usepackage{pdflscape}

\usepackage{tikz}
\tikzstyle{ov}=[shape=rectangle,
                draw=black!50,
                thick,
                minimum width=0.7cm,
                minimum height=0.7cm]

\tikzstyle{av}=[shape=rectangle,
                draw=black!50,
                fill=black!10,
                thick,
                minimum width=0.7cm,
                minimum height=0.7cm]

\tikzstyle{lv}=[shape=circle,draw=black!50,thick]

\usetikzlibrary{shapes,calendar,matrix,backgrounds,folding,snakes}

\newcommand{\nothere}[1]{}

\newtheorem{theorem}{Theorem}
\newtheorem{lemma}{Lemma}

\title{\LARGE\bfseries Inference on Variable Importance for \\Treatment Effect Heterogeneity: \\
Shapley Values and Beyond
\vspace{.5em}}

\author{\Large Pawe{\l} Morzywo{\l}ek$^{1,2}$, Peter B. Gilbert$^{2,3}$, Alex Luedtke$^{2,4}$\vspace{-.25em}}
\affil{\normalsize $^1$University of Copenhagen,\hspace{1.5ex}$^2$University of Washington,\hspace{1.5ex}$^3$Fred Hutchinson Cancer Center,\hspace{1.5ex}$^4$Harvard University}

\date{\today}

\begin{document}
\allowdisplaybreaks
\maketitle

\begin{abstract}
We provide an inferential framework to assess variable importance for heterogeneous treatment effects. This assessment is especially useful in high-risk domains such as medicine, where decision makers hesitate to rely on black-box treatment recommendation algorithms. The variable importance measures we consider are local in that they may differ across individuals, while the inference is global in that it tests whether a given variable is important for any individual. Our approach builds on recent developments in semiparametric theory for function-valued parameters, and is valid even when statistical machine learning algorithms are employed to quantify treatment effect heterogeneity. We demonstrate the applicability of our method to infectious disease prevention strategies.
\end{abstract}

\section{Introduction} \label{Section1}

We present an inferential framework for evaluating variable importance for heterogeneous treatment effects. Our approach focuses on local variable importance measures that vary among individuals, coupled with global inference that assesses the overall significance of a variable across all individuals. Our global tests evaluate null hypotheses of the form
\begin{align*}
 H_0 : \textnormal{the variable $X_j$ is of no importance for predicting an individual's treatment effect}
\end{align*}
against the complementary alternatives. Each such null is formalized by assigning a scalar measure of importance to a variable $X_j$, where zero indicates no importance. Accompanying function-valued local measures of variable importance are also provided, with corresponding confidence bands.

This work is motivated by the growing interest in quantifying treatment effect heterogeneity, which aids individualized decision making. Recent years have seen substantial advances in this line of literature \citep{Nie2021, Kennedy2023, Morzywolek2023}. However, decision makers, particularly in high-risk domains like medicine, may be hesitant to trust black-box decision support systems unless they understand the rationale behind recommendations. Therefore, statisticians and data analysts must not only provide accurate risk predictions and treatment effect estimates but also offer insights into the key variables influencing these predictions. 

In the literature on variable importance, it is common to distinguish between measures tied to a particular fitted model \citep{Breiman2001, Ishwaran2007, Strobl2007, Gromping2009, Benard2023} and `model-agnostic' measures, which are defined independently of any specific model \citep{Williamson2021a, Williamson2021, Williamson2020, Hines2025, Verdinelli2023, Li2023, Ziersen2024, Dai2024, Boileau2025}. In this work, we are interested in gaining insight into the data-generating process, rather than explaining predictions coming from a particular model. Hence, we will focus on model-agnostic measures. Some of the most popular such measures are leave-out-covariates (LOCO) \citep{Lei2018, Rinaldo2019, Williamson2021a, Williamson2021, Verdinelli2023} and Shapley values \citep{Shapley1953, Lundberg2017, Covert2020, Williamson2020, Verdinelli2023}.

Model-agnostic variable importance measures can be further classified into two main types. Global measures quantify the relevance of a given feature across the whole population \citep{Covert2020, Williamson2021a, Hines2025, Verdinelli2023, Li2023, Ziersen2024}. Local measures---also called individual measures---assess the importance of a variable for specific individuals with particular covariate profiles \citep{Dai2024}. Interestingly, local variable importance measures can still be used to compute a global variable importance measure after being averaged out across the population.

In recent years, several papers have explored variable importance measures for heterogeneous treatment effects. For example, \cite{Hines2025}, \cite{Li2023}, \cite{Ziersen2024}, and \cite{Boileau2025} focus on global, model-agnostic variable importance measures, while \cite{Benard2023} addresses variable importance in causal random forests. To the best of our knowledge, our paper is the first to study statistical inference for model-agnostic local variable importance measures in the context of heterogeneous treatment effects.

Statistical inference for local variable importance measures is challenging because these measures are function-valued. They take an individual's covariates as input and return the importance of a variable. Testing the null hypothesis of no importance is particularly difficult since the null lies on the boundary of the parameter space. This may lead to vanishing efficient influence function, requiring a higher-order expansion to develop valid, non-conservative tests \citep{Luedtke2019, Williamson2021a, Williamson2021, Verdinelli2023}. A common alternative has been to construct conservative tests using sample splitting \citep{Williamson2021, Hines2025, Ziersen2024}. Another approach leverages recent advances in semiparametric theory for pathwise differentiable Hilbert-valued parameters \citep{Luedtke2024}. This approach overcomes the boundary issue faced by other methods and facilitates statistical inference, even when flexible, data-adaptive algorithms are used to estimate nuisance parameters. We build on this framework, but differ from \cite{Luedtke2024} in how regularization is achieved. Specifically, rather than using an orthogonal expansion of the parameter with decaying scaling coefficients, we regularize through an embedding into a reproducing kernel Hilbert space (RKHS). This embedding is injective when the kernel is universal \citep{sriperumbudur2011universality}, ensuring no information is lost through regularization. We show that the resulting regularized target is amenable to semiparametric analysis.

Our contributions are:
\begin{enumerate}
    \item local, model-agnostic importance measures for a variable or set of variables,
    \item regularization of a Hilbert-valued parameter through an RKHS embedding to obtain its efficient influence function,
    \item Wald-type tests of no importance for a variable or set of variables, and
    \item confidence bands for the RKHS embeddings of local variable importance measures.
\end{enumerate}
Our work contributes to the literature on the inference for function-valued parameters \citep{Hudson2021, Williamson2021, Hudson2023, Hudson2023a,Luedtke2024}. We also demonstrate the applicability of our methodology in the context of infectious disease prevention strategies.

\section{Problem Setup} \label{Section2}

\subsection{Notation}

We consider independent observations $Z_1, Z_2, \dots, Z_n$ drawn from a distribution $P_0$, which is assumed to lie in a non-parametric statistical model $\mathcal{M}$ of mutually absolutely continuous distributions. Each $Z$ takes the form $(X,A,Y)$, with $X \coloneqq ( X_1, X_2, \dots, X_d ) \in \mathcal{X} \subset \mathbb{R}^d$ a vector of baseline covariates, $A$ a binary treatment, and $Y$ a real-valued outcome. We use $z=(x,a,y)$ to denote generic realizations of the random variable $Z=(X,A,Y)$. For a distribution $P\in\mathcal{M}$, we let $P_X$ denote the marginal distribution of $X$. Furthermore, let $\mu_P ( a, x ) \coloneqq \mathbb{E}_P ( Y \vert A=a, X=x )$ be the outcome regression and $\nu_P ( \mathcal{V} )$ be the function
\begin{align*}
     \nu_P ( \mathcal{V} ) ( x_\mathcal{V} ) \coloneqq \mathbb{E}_P \left\lbrace \mu_P ( 1, X ) - \mu_P ( 0, X ) \mid X_\mathcal{V} = x_\mathcal{V} \right\rbrace,
\end{align*}
where $\mathcal{V} \subseteq [d] \coloneqq \lbrace 1, \dots, d \rbrace$ is a subset of indices and $X_\mathcal{V} \coloneqq \lbrace X_j \colon j \in \mathcal{V} \rbrace$ is a subset of baseline covariates indexed by $\mathcal{V}$ \citep{Coston2020, Morzywolek2023}. Furthermore, define $X_{-\mathcal{V}} \coloneqq \lbrace X_j \colon j \in [d] \backslash \mathcal{V} \rbrace$. For notational convenience, whenever a function is indexed by `$\hat{P}_n$', we write $f_n$ in place of $f_{\hat{P}_n}$; similarly, we write $f_0$ in place of $f_{P_0}$. Under the assumptions of conditional exchangeability, positivity and consistency \citep{Neyman1923, Rubin1974,Robins1986, Hernan2020}, the function $\nu_P ( \mathcal{V} )$ corresponds to the conditional average treatment effect (CATE). 

Often it is of interest to identify and quantify how different variables contribute to treatment effect heterogeneity. Below, we discuss a few commonly used variable importance measures. Although these measures were originally developed for general prediction tasks, we adapt them to the setting of treatment effect heterogeneity. They can be used to test the null hypothesis of no importance for a given variable. Such tests can indicate whether a variable contributes to treatment effect heterogeneity under the chosen importance measure. They may also help assess whether measuring that variable is useful when the goal is to estimate treatment effect heterogeneity in future patients.

\subsection{Leave-Out-Covariates (LOCO) variable importance measures}

LOCO measures seek to assess whether the CATEs are identical when conditioned on the nested covariate subsets $X_{\mathcal{V}_1}$ and $X_{\mathcal{V}_2}$, where $\mathcal{V}_2 \subseteq \mathcal{V}_1$. If
\begin{align*}
    H_0\ \colon \ \gamma_{LOCO} ( P; \mathcal{V}_1, \mathcal{V}_2 ) ( x_{\mathcal{V}_1} ) \coloneqq \nu_P ( \mathcal{V}_1 ) ( x_{\mathcal{V}_1} ) - \nu_P ( \mathcal{V}_2 ) ( x_{\mathcal{V}_2} ) = 0 \hspace{1em} \text{ $P$-almost surely ($P$-a.s.),}
\end{align*}
then variables in $\mathcal{V}_1\backslash \mathcal{V}_2$ are considered globally unimportant, and otherwise they are important for at least some subset of the population. Hence, a global variable importance test evaluates $H_0$ against the complementary alternative. There are several variable importance measures and the corresponding tests of no importance based on $\gamma_{LOCO}$ for particular choices of index subsets $\mathcal{V}_1$ and $\mathcal{V}_2$. These measures can be used to assess:
\begin{enumerate}
    \item \textbf{unconditional effect modification:} whether the treatment effect varies across different levels of a set of variables $X_{\mathcal{V}_1}$, as compared to the trivial set $X_{\mathcal{V}_2}$ with $\mathcal{V}_2=\emptyset$. When $\mathcal{V}_1=\{i\}$, this variable importance measure is referred to as keep-one-in (KOI) \citep{Hines2025}.
    \item \textbf{conditional effect modification:} whether an additional set of variables $X_{\mathcal{V}_1\backslash\mathcal{V}_2}$ modifies the treatment effect, conditional on a nontrivial set of baseline covariates $X_{\mathcal{V}_2}$. When $\mathcal{V}_1 = [d]$ and $\mathcal{V}_2 = [d] \backslash \lbrace i \rbrace$ the resulting variable importance measure is referred to as leave-one-out (LOO) \citep{Hines2025}.
\end{enumerate}
Many variable importance measures in the literature are LOCO measures \citep{Lei2018, Rinaldo2019, Williamson2021a, Williamson2021, Hines2025, Verdinelli2023}. Their primary advantage lies in their ease of interpretation and computation. A disadvantage is that they are affected by multicollinearity among variables, which can lead to counterintuitive results. For example, if two variables $X_i$ and $X_j$ are perfectly correlated, then both will have zero LOO variable importance, even if they are highly predictive of the treatment effect \citep{Verdinelli2023}.

\subsection{Shapley values}

Shapley values are another popular variable importance measure \citep{Lundberg2017, Covert2020, Williamson2020, Verdinelli2023}. They stem from cooperative game theory, where they have been developed as means to distribute the payout among a group of cooperative players in a fair manner \citep{Shapley1953}. When using the CATE as the payoff function, the Shapley value of $i$-th feature is given by
\begin{align*}
     \gamma_{SHAP,i} ( P ) ( x ) &\coloneqq \sum_{\mathcal{V} \subseteq [d] \backslash \lbrace i \rbrace} \beta_\mathcal{V}\, \gamma_{LOCO} ( P; \mathcal{V} \cup \lbrace i \rbrace, \mathcal{V} ) ( x_{\mathcal{V} \cup \left\lbrace i \right\rbrace} ),
\end{align*}
where $\beta_\mathcal{V} \coloneqq \frac{1}{d} \binom{d-1}{\lvert \mathcal{V} \rvert}^{-1}$. The above display illustrates that the Shapley value is a $\beta_\mathcal{V}$-weighted LOCO variable importance measure, as discussed in the context of regression in \cite{Verdinelli2023}. Alternatively, one can view the Shapley values as the marginal contribution of the variable $X_i$ to the subset of variables $X_\mathcal{V}$, where $\mathcal{V} \subseteq [d] \backslash \lbrace i \rbrace$, appropriately averaged over the set of all permutations of the features \citep{Mitchell2022, Verdinelli2023}---see Appendix~\ref{app:permutation} for further discussion.

One reason Shapley values are appealing as a variable importance measure is that they uniquely satisfy a specific set of axioms \citep{Shapley1953}. These axioms include efficiency (which differs from the statistical meaning of the term), symmetry, null player, and linearity. Together, these properties are desirable for the payout function in the context of cooperative game theory. We discuss their formulation in the context of variable importance for heterogeneous treatment effects in Appendix~\ref{SectionA2}. 

Table \ref{tab:var_importance_disagreement} provides examples of scenarios where the KOI, LOO, and Shapley variable importance measures will be approximately null. Beyond those scenarios, it is theoretically possible for the Shapley value to be near zero while KOI and LOO are not. This would require a precise cancellation of variable importance across different feature subsets when weighted by $\beta_{\mathcal{V}}$---a situation that, while possible, seems unlikely in practice.

\begin{table}[tb]\small
\centering
\caption{Comparison of variable importance measures under scenarios that cause them to disagree. Measures marked with `$\approx 0$' will necessarily be small, while others may not be.}\label{tab:var_importance_disagreement}
\renewcommand{\arraystretch}{1.2} 
\begin{tabular}{l c c c}
\textbf{Scenario} & \textbf{KOI} & \textbf{LOO} & \textbf{Shapley} \\
\midrule
Variable only predictive through interactions with others 
& $\approx 0$
&  
&  \\
Variable is highly correlated with another variable  
&  
& $\approx 0$
&  \\
Variable is highly correlated with many other variables  
&  
& $\approx 0$
& $\approx 0$ \\
Variable only predictive in combination with strict subset of others 
& $\approx 0$
& $\approx 0$ \\
Variable only predictive in combination with all others 
& $\approx 0$
& 
& $\approx 0$ 
\end{tabular}
\end{table}

\section{Estimation and inference} \label{Section3}

\subsection{Proposed estimator}

We introduce a parameter that makes it possible to express different variable importance measures. For a vector of weights $\omega= \lbrace \omega_\mathcal{V} : \mathcal{V} \subseteq [d] \rbrace$, this parameter is defined as
\begin{align} \label{eqn1}
    \gamma_\omega (P) (x) \coloneqq \sum_{\mathcal{V} \subseteq [d]} \omega_\mathcal{V}\,\nu_P ( \mathcal{V}) (x_\mathcal{V}).
\end{align}
In particular, the parameter $\gamma_{LOCO} ( P; \mathcal{V}_1, \mathcal{V}_2 )$ is a special case of $\gamma_\omega ( P )$ with $\omega_{\mathcal{V}_1} = 1$, $\omega_{\mathcal{V}_2} = -1$, and $\omega_{\mathcal{V}}=0$ for all $\mathcal{V} \not\in\{\mathcal{V}_1,\mathcal{V}_2\}$, and $\gamma_{SHAP,i} ( P )$ is a special case of $\gamma_\omega ( P )$ with Shapley weights $\omega_{
\mathcal{V} \cup \lbrace i \rbrace} = \frac{1}{d} \binom{d-1}{\lvert \mathcal{V} \rvert}^{-1}$ and $\omega_{\mathcal{V}} = - \frac{1}{d} \binom{d-1}{\lvert \mathcal{V} \rvert}^{-1}$, for $\mathcal{V} \subseteq [d] \backslash \lbrace i \rbrace$.

To assess variable importance, we consider testing
\begin{align*}
    H_0 \colon \gamma_{\omega} ( P_0 ) = 0\ \ (P_0\textnormal{-a.s.}) \quad \textnormal{ vs. } \quad H_1\colon \textnormal{not $H_0$.}
\end{align*}
To this end, we wish to construct a Wald-type test based on a root-$n$ consistent estimator of the parameter $\gamma_\omega (P)$. Unfortunately, the root-$n$ consistent estimation of the parameter $\gamma_\omega (P)$ may often be infeasible due to lack of sufficient smoothness of the parameter---see Appendix \ref{SectionA1} for a review of semiparametric theory for Hilbert-valued parameters and a more extensive discussion. As a remedy, we embed the parameter in an RKHS, which enforces additional smoothness and facilitates statistical inference. Let $\mathcal{K} \colon \mathcal{X} \times \mathcal{X} \rightarrow \mathbb{R}$ be a symmetric, continuous positive semi-definite kernel function associated with the RKHS $\mathcal{H}$. The embedding of the parameter of interest into the RKHS $\mathcal{H}$ is given by \citep[][Theorem 106]{berlinet2003reproducing}
\begin{align} \label{eqn3}
    \gamma_\omega^\mathcal{K} (P) (x) \coloneqq \int_\mathcal{X} \mathcal{K} (x, x^\prime)\, \gamma_\omega (P) (x^\prime)\, P_{X} (dx^\prime).
\end{align}
Intuitively, one may view this embedding as follows. The function $\gamma_\omega (P)$ may not be smooth, which can result in statistical inference being infeasible. However, the function resulting from the RKHS embedding is smooth, which will allow us to draw statistical inferences about it. Moreover, provided the kernel is universal \citep{sriperumbudur2011universality}---as, for example, is a Gaussian or Laplace kernel on a compact domain---the embedding is one-to-one, allowing inferences about $\gamma_\omega^\mathcal{K}(P)$ to be mapped back to the original target of inference, $\gamma_\omega (P)$. Figure \ref{figure1} illustrates the effect of the RKHS embedding on two sample functions. In practice we recommend selecting the bandwidth parameter the kernel relies on using the median heuristic \citep{Garreau2018}. Furthermore, for $f \in \mathcal{H}$ and $r\ge 1$, we define the Bochner norm $\left\| f \right\|_{L^r ( P; \mathcal{H} )} \coloneqq \left\lbrace \int \left\| f (z) \right\|_{\mathcal{H}}^r P (dz) \right\rbrace^{1/r}$.

\begin{figure} 
\begin{center}
 \includegraphics[width=\textwidth]{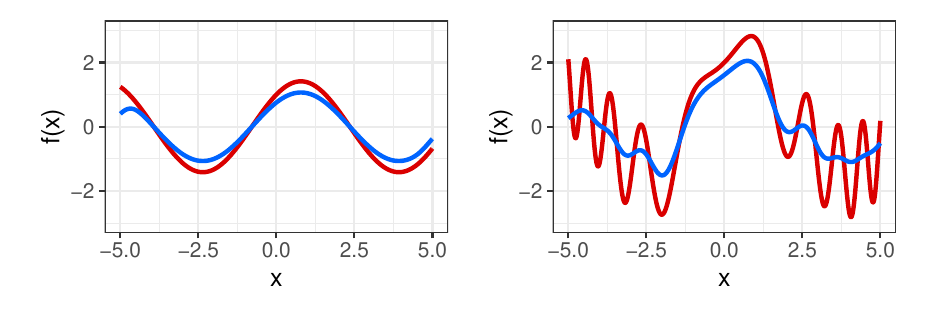}    
\end{center}
\caption{Illustration of {\color{CB5red}\textbf{functions}} and {\color{CB5blue}\textbf{their RKHS embeddings}}. Smooth functions remain almost unchanged (left), whereas rough ones become noticeably smoother (right). For a universal kernel, the RKHS embedding is injective, allowing the original function to be losslessly recovered from its embedding.} \label{figure1}
\end{figure}

Below, we establish the pathwise differentiability of $\gamma_\omega^\mathcal{K}$. To do so, we impose the following:
\begin{enumerate}[label={C\arabic*)},ref=C\arabic*, series=conds]
    \item \label{cond:bounded} \textit{Kernel boundedness:} The kernel satisfies $\sup_{x} \mathcal{K}(x, x) < \infty.$
    \item \label{cond:positivity} \textit{Strong positivity:} there exists a constant \(c > 0\) such that $g_P(a \mid x) \coloneqq \mathbb{P}(A = a \mid X = x) \geq c$ $P_X$-a.s. for each $a \in \lbrace 0,1 \rbrace$.
    \item \label{cond:finitesecondmoment} \textit{Finite conditional second moment of outcome:} $\sup_{P\in\mathcal{M}}\|\mathbb{E}_P[Y^2\mid (A,X)=\cdot\,]\|_{L^\infty(P_{A,X})} < \infty$.
\end{enumerate}
Above, $P_{A,X}$ denotes the marginal distribution of $(A,X)$ under $P$.

\begin{theorem}[Pathwise differentiability] \label{Theorem1}
    Suppose \ref{cond:bounded}-\ref{cond:finitesecondmoment}. Then the function-valued parameter $\gamma_\omega^\mathcal{K}$ is pathwise differentiable at each $P \in \mathcal{M}$ relative to $\mathcal{M}$ with efficient influence function
\begin{align*}
    \phi_{P}^{\omega}(z)(x^\prime)&= \sum_{\mathcal{V} \subseteq [d]} \omega_\mathcal{V} \left[ \kappa_{P,\mathcal{V}}(x_\mathcal{V})(x^\prime) \psi_P(z) + \left\lbrace \mathcal{K} (x^\prime, x) - \kappa_{P,\mathcal{V}}(x_\mathcal{V})(x^\prime) \right\rbrace \nu_P( \mathcal{V})( x_\mathcal{V}) \right] - \gamma_\omega^\mathcal{K}( P)( x^\prime),
\end{align*}
where $\psi_P ( z ) \coloneqq \frac{2a-1}{g_P ( a \mid x)} \lbrace y - \mu_{P} (a, x) \rbrace + \mu_{P} (1, x) - \mu_{P} (0, x)$.
\end{theorem}
Here, $\kappa_{P,\mathcal{V}} (x_\mathcal{V}) (\cdot) \coloneqq \mathbb{E}_P \lbrace \mathcal{K} (\cdot, X) \mid X_{\mathcal{V} } = x_\mathcal{V}\rbrace$ denotes the kernel conditional mean embedding \citep{Song2009, Muandet2017, Park2020, Park2021}. Heuristically, the pathwise differentiability of the parameter is important as it informs the von Mises expansion $\gamma_\omega^\mathcal{K} (P_0) (\cdot) \approx \gamma_\omega^\mathcal{K}(\hat{P}_n) (\cdot) + P_0 \phi_{n}^{\omega} (\cdot)$, where $\hat{P}_n$ is an estimate of $P_0$ and $Pf$ denotes the integral $\int f (z) dP(z)$ for any probability measure $P$ and $P$-integrable real-valued function $f$.

Replacing the mean under the unknown $P_0$ by an empirical mean, this approximation suggests the form of the one-step estimator $\hat{\gamma}^\mathcal{K}_\omega (\cdot) \coloneqq \gamma^\mathcal{K}_\omega (\hat{P}_n) (\cdot) + P_n \phi_{n}^\omega (\cdot)$, which simplifies as
\begin{align*}
    \hat{\gamma}^\mathcal{K}_\omega (\cdot) = \frac{1}{n} \sum_{i=1}^n \sum_{\mathcal{V} \subseteq [d]} \omega_\mathcal{V} \left[ \nu_n (\mathcal{V}) ( X_{\mathcal{V},i} ) \mathcal{K} ( \cdot, X_i ) + \left\lbrace \psi_n ( Z_i ) - \nu_n (\mathcal{V}) ( X_{\mathcal{V},i} ) \right\rbrace \kappa_{n,\mathcal{V}} (X_{\mathcal{V},i})(\cdot) \right].
\end{align*}
Computing $\hat{\gamma}^\mathcal{K}_\omega$ can be computationally expensive for certain choices of weights, such as the Shapley weights $\beta_{\mathcal{V}}$. In Appendix \ref{app:computationalconsiderations}, we discuss computational considerations and provide recommendations for performing these computations efficiently.

\subsection{Weak convergence}

We now study the large-sample properties of the one-step estimator $\hat{\gamma}_\omega^\mathcal{K}$. For simplicity, our study focuses on the sample splitting case where $\hat{P}_n$ is obtained using an independent set of iid draws from $P_0$. Similar results apply if cross-fitting is employed---see \cite{Luedtke2024}.

The result below states that $\hat{\gamma}_\omega^\mathcal{K}$ is asymptotically linear---and therefore asymptotically Gaussian---which enables statistical inference. We discuss its conditions at the end of this subsection.

\begin{theorem}[Asymptotic linearity and weak convergence]
\label{Theorem2} 
    Suppose \ref{cond:bounded}-\ref{cond:finitesecondmoment} and
    \begin{enumerate}[label={C\arabic*)},ref=C\arabic*,resume*=conds]
    \item \label{cond:rategmu} $\max_{a\in\{0,1\}} \left\lbrace \left\| 1 - \frac{g_0 (a \mid x)}{g_n (a \mid x)} \right\|_{L^4 ( P_{0,X} )} \| \mu_{0} (a,x) - \mu_{n} (a,x) \|_{L^4 ( P_{0,X} )} \right\rbrace = o_{P} ( n^{-1/2} )$,
    \item \label{cond:ratekappanu} $\| \kappa_{n,\mathcal{V}} - \kappa_{0,\mathcal{V}} \|_{L^2(P_{0,X};\mathcal{H})} \left\| \nu_n ( \mathcal{V} ) - \nu_0 ( \mathcal{V} ) \right\|_{L^2 ( P_{0,X} )} = o_{P} ( n^{-1/2} )$,
    \item \label{cond:consistentkappa} $\left\| \kappa_{n,\mathcal{V}} - \kappa_{0,\mathcal{V}} \right\|_{L^4(P_{0,X};\mathcal{H})} = o_P(1)$ and $\| \kappa_{n,\mathcal{V}} \|_{L^2(P_{0,X};\mathcal{H})}=O_P(1)$,
    \item \label{cond:consistentpsi} $\left\| \psi_n - \psi_0 \right\|_{L^2(P_0)} = o_P(1)$ and $\| \psi_n \|_{L^4(P_{0})}  = O_P(1)$, and
    \item \label{cond:consistentnu} $\left\| \nu_n (\mathcal{V}) - \nu_0 (\mathcal{V}) \right\|_{L^2(P_{0,X})} = o_P(1)$ and $\| \nu_n (\mathcal{V})\|_{L^4 (P_{0,X})} = O_P(1)$.
    \end{enumerate}
    Then, the following hold:
    \begin{enumerate}[label=(\roman*)]
        \item Asymptotic linearity: $\hat{\gamma}_\omega^\mathcal{K} - \gamma_\omega^\mathcal{K} ( P_0 ) = \frac{1}{n} \sum_{i=1}^n \phi_{0}^\omega ( Z_i ) + o_P ( n^{-1/2} )$, and
        \item Weak convergence: $n^{1/2} \left[ \hat{\gamma}_\omega^\mathcal{K} - \gamma_\omega^\mathcal{K} ( P_0 ) \right] \rightsquigarrow \mathbb{H}_\omega$, where $\mathbb{H}_\omega$ is a tight $\mathcal{H}$-valued Gaussian random variable such that, for each $h \in \mathcal{H}$, $\langle \mathbb{H}_\omega, h \rangle_\mathcal{H} \sim \mathcal{N} \left( 0, P_0 [ \langle \phi_{0}^\omega, h \rangle_\mathcal{H}^2 ] \right)$.
    \end{enumerate}
\end{theorem}
Using the above weak convergence result for the one-step estimator $\hat{\gamma}_\omega^\mathcal{K}$, we can construct an asymptotically valid confidence set. Consider the following set:
\begin{equation} \label{eqn4}
    \mathcal{C}_n^\omega ( \xi_\omega ) \coloneqq \left\lbrace h \in \mathcal{H} \colon \| \hat{\gamma}_\omega^\mathcal{K} - h \|^2_\mathcal{H} \leq \xi_\omega / n \right\rbrace,
\end{equation}
where $\xi_\omega \geq 0$ is a specified threshold. Let $\xi_{\omega, 1 - \alpha}$ be the $( 1 - \alpha )$-quantile of $\| \mathbb{H}_\omega \|_\mathcal{H}^2$.

\begin{theorem}[Asymptotically valid confidence set] \label{Theorem3} Suppose the conditions of Theorem \ref{Theorem2} hold. If $\hat{\xi}_{\omega} \rightarrow \xi_{\omega, 1-\alpha}$ in probability, then for any weights $\omega$,
\begin{align*}
    \lim_{n \rightarrow \infty} P_0^n \left\lbrace \gamma^\mathcal{K}_\omega ( P_0 ) \in \mathcal{C}_n^\omega ( \hat{\xi}_{\omega} ) \right\rbrace = 1 - \alpha.
\end{align*}
\end{theorem}

The bootstrap can be used to obtain the consistent estimator of $\xi_{\omega, 1-\alpha}$ required by the theorem. Let $Z_1^\sharp, Z_2^\sharp, \dots, Z_n^\sharp$ be an iid sample from the empirical measure $P_n$ and let $P_n^\sharp$ be the empirical distribution of $Z_1^\sharp, Z_2^\sharp, \dots, Z_n^\sharp$. Furthermore, define $\mathbb{H}_{\omega,n}^\sharp \coloneqq n^{1/2} ( P_n^{\sharp} - P_n ) \phi_n^\omega$. Under conditions, the bootstrap estimator $\hat{\xi}_{\omega}$ obtained as $( 1 - \alpha )$-quantile of $\|\mathbb{H}_{\omega,n}^\sharp\|_\mathcal{H}^2$ is a consistent estimator of the threshold $\xi_{\omega, 1-\alpha}$---see the Appendix \ref{SectionA4} for further discussion.

Theorem \ref{Theorem2} relies on conditions \ref{cond:rategmu}-\ref{cond:consistentnu}, which impose requirements on the magnitudes and convergence rates of the nuisance estimators. Specifically, these conditions involve the $L^2 (P_{0,X};\mathcal{H})$-norm of $\kappa_{n,\mathcal{V}}$, $L^4 (P_{0,X})$-norms of $\psi_n$ and $\nu_n$, the $L^2 (P_{0,X};\mathcal{H})$ and $L^4 (P_{0,X};\mathcal{H})$-distances between $\kappa_{n,\mathcal{V}}$ and the kernel conditional mean embedding $\kappa_{0,\mathcal{V}}$, $L^4 (P_{0,X})$-distances between $g_n (a \mid x)$ and the propensity score $g_0 (a \mid x)$ and $\mu_n (a, x)$ and the outcome regression $\mu_0 (a, x)$, the $L^2 (P_{0,X})$-distance between $\nu_n$ and $\nu_0$, and $L^2 (P_{0})$-distance between $\psi_n$ and $\psi_0$. Such conditions have been studied extensively across a wide range of semiparametric and nonparametric problems \citep[e.g.,][]{vanderLaan2006,chernozhukov2018double,Kennedy2023,Luedtke2024}; they often hold when the relevant nuisance functions are sufficiently smooth and are estimated using approaches that exploit this smoothness. For example, suppose $X$ is continuous and $\mathbb{R}^{d}$-valued, and that $g_0$, $\mu_0$, $\nu_0$ and $\kappa_{0,\mathcal{V}}$ are H\"older smooth with exponents $b_g$, $b_\mu$, $b_\nu$ and $b_\kappa$, respectively. If these quantities are estimated using kernel regression with kernels of sufficiently high orders, then \ref{cond:rategmu} holds if $\frac{b_g}{2b_g+d} + \frac{b_\mu}{2b_\mu+d} > \frac{1}{2}$, while \ref{cond:ratekappanu} holds provided $\frac{b_\kappa}{2b_\kappa+d_\mathcal{V}} + \frac{b_\nu}{2b_\nu+d_\mathcal{V}} > \frac{1}{2}$, where $d_\mathcal{V}$ is the dimension of $X_\mathcal{V}$. Condition \ref{cond:consistentkappa} requires consistency and boundedness of $\kappa_{n,\mathcal{V}}$. In case of kernel-based estimators, this is ensured by choosing a bandwidth $h$ such that $h \rightarrow 0$ and $n h^{d_\mathcal{V}} \rightarrow \infty$, which yields consistency, while boundedness follows from standard stability and moment conditions on the kernel. Condition \ref{cond:consistentpsi} requires consistency of $\psi_n$, which follows from the consistency of either $g_n$ or $\mu_n$ via the double-robustness property of $\psi_n$, together with boundedness implied by \ref{cond:positivity}. Finally, \ref{cond:consistentnu} requires consistency of $\nu_n$. For kernel regression, this again follows from choosing bandwidth $h$ such that $h \rightarrow 0$ and $n h^{d_\mathcal{V}} \rightarrow \infty$, and boundedness following from standard stability and moment conditions on the kernel.

\subsection{Statistical inference: confidence band and hypothesis test} \label{Section3.3}

A confidence band for the RKHS-embedded parameter evaluation $\gamma_\omega^\mathcal{K}(P_0)$ can be defined as
\begin{equation} \label{eqn5}
    \tilde{\mathcal{C}}_n^\omega ( \hat{\xi}_{\omega} ) \coloneqq \left\lbrace h \in \mathcal{H} \colon \| \hat{\gamma}_\omega^\mathcal{K} - h \|^2_\infty \leq \hat{\xi}_{\omega} \sup_x \mathcal{K} (x,x) / n \right\rbrace,
\end{equation}
where $\| \cdot \|_\infty$ is the supremum norm. It follows from Theorem \ref{Theorem3} and the fact that $\| \cdot \|_\infty \leq \sqrt{\sup_x \mathcal{K} (x,x)} \| \cdot \|_\mathcal{H}$, by Cauchy-Schwarz inequality, that the confidence band $\tilde{\mathcal{C}}_n^\omega ( \hat{\xi}_{\omega} )$ is asymptotically valid under the conditions of Theorems \ref{Theorem2} and \ref{Theorem3}.

To test for variable importance, we seek to evaluate
\begin{align*}
    H_0 \colon \gamma_{\omega}^\mathcal{K} ( P_0 ) = 0\ \ \textnormal{ vs. }\ \ H_1\colon \textnormal{not $H_0$.}
\end{align*}
We can construct a test based on the confidence set given in (\ref{eqn4}), rejecting $H_0$ if and only if $\gamma_{\omega}^\mathcal{K} ( P_0 ) \notin \mathcal{C}_n^\omega ( \hat{\xi}_{\omega} )$. By Theorem~\ref{Theorem3}, this test asymptotically controls Type I error at level $\alpha$, when $\mathcal{C}_n^\omega ( \hat{\xi}_{\omega} )$ has asymptotically valid coverage. Furthermore, this test is consistent against fixed alternatives and has asymptotically non-trivial power against local alternatives that decay at an $n^{-1/2}$ rate---see Theorem \ref{Theorem5} and Theorem \ref{Theorem4} in Appendix \ref{SectionA5}. 

\subsection{Confidence interval for the norm of Hilbert-valued parameter}

Beyond null hypothesis testing, the magnitude of a variable's importance may be of interest. This can be quantified by computing the norm of the RKHS embedding of the local variable importance measure, $\|\gamma_\omega^\mathcal{K} ( P_0 )\|_\mathcal{H}$. A $( 1-\alpha )$-confidence interval for this quantity is given by
\begin{align*}
   \left\| \hat{\gamma}_{\omega}^\mathcal{K} \right\|_\mathcal{H} \pm \sqrt{\hat{\xi}_{\omega}/n}.
\end{align*}
This is justified by the reverse triangle inequality and the fact that the confidence set studied in Theorem \ref{Theorem3} is asymptotically valid, yielding $\lvert \|\hat{\gamma}_\omega^\mathcal{K}\|_\mathcal{H} - \|\gamma_\omega^\mathcal{K} ( P_0 )\|_\mathcal{H} \rvert \leq \| \hat{\gamma}_{\omega}^\mathcal{K} - \gamma_\omega^\mathcal{K} ( P_0 )\|_\mathcal{H} \leq (\hat{\xi}_{\omega}/n)^{1/2}$ with probability tending to $1-\alpha$. This interval is asymptotically non-conservative in the special case where $\gamma_\omega^\mathcal{K} ( P_0 ) = 0$, but can otherwise be conservative.

An alternative approach uses the functional delta method to construct confidence intervals for the norm of a Hilbert-valued parameter.
\begin{lemma}[Delta method] 
\label{lemma2}
    Let $\rho \colon \mathcal{H} \rightarrow \mathbb{R}$ be Hadamard-differentiable at $\gamma_\omega^\mathcal{K} ( P_0 )$. Suppose $n^{1/2} \left[ \hat{\gamma}_\omega^\mathcal{K} - \gamma_\omega^\mathcal{K} ( P_0 ) \right] \rightsquigarrow \mathbb{H}_\omega$. Then, $n^{1/2} \left[ \rho ( \hat{\gamma}_\omega^\mathcal{K} ) - \rho ( \gamma_\omega^\mathcal{K} ( P_0 ) ) \right] \rightsquigarrow \rho^\prime_{\gamma_\omega^\mathcal{K} ( P_0 )} ( \mathbb{H}_\omega )$.
\end{lemma}
In particular, for $\rho = \| \cdot \|_\mathcal{H}^2$, it holds that $n^{1/2} ( \| \hat{\gamma}_{\omega}^\mathcal{K} \|_\mathcal{H}^2 - \| \gamma_\omega^\mathcal{K} ( P_0 ) \|_\mathcal{H}^2 ) \rightsquigarrow 2 \langle \gamma_\omega^\mathcal{K} ( P_0 ), \mathbb{H}_\omega \rangle_\mathcal{H}$. By the continuous mapping theorem, this also implies $n^{1/2} \left\lvert \| \hat{\gamma}_{\omega}^\mathcal{K} \|_\mathcal{H}^2 - \| \gamma^\mathcal{K}_\omega ( P_0 ) \|_\mathcal{H}^2 \right\rvert \rightsquigarrow 2 \lvert \langle \gamma_\omega^\mathcal{K} ( P_0 ), \mathbb{H}_\omega \rangle_\mathcal{H} \rvert=: 2 \mathbb{W}_\omega$. If $\sigma_\omega^2\coloneqq E_0[\langle \phi_{0}^\omega ( Z ), \gamma_\omega^\mathcal{K} ( P_0 ) \rangle_\mathcal{H}^2]>0$, then $\mathbb{W}_\omega$ follows a half-normal distribution with parameter $\sigma_\omega^2$. That result is the key to establishing the following consequence of Theorem~\ref{Theorem2}.

\begin{theorem}[Confidence interval for the variable importance measure] \label{Theorem:confint}

Suppose \ref{cond:bounded}-\ref{cond:finitesecondmoment} and:
\begin{enumerate}[label={C\arabic*)},ref=C\arabic*,resume*=conds]
    \item \label{cond:consistentest} $\hat{\varsigma}_{\omega}$ is a consistent estimator of $\varsigma_{\omega, 1-\alpha}$, the $(1-\alpha)$-quantile of $\mathbb{W}_\omega$
    \item \label{cond:spdcovariance} $\Sigma \coloneqq P_0 \left[ \phi_0^\omega \otimes \phi_0^\omega \right]$---the covariance operator of $\phi_0^\omega (Z)$---is strictly positive definite.
\end{enumerate}
Define
\begin{align*}
    \mathcal{C}_{\omega,n}^{>0}&\coloneqq \left\lbrace x \in \mathbb{R}_{>0}\, \colon \left\lvert \|\hat{\gamma}_{\omega}^\mathcal{K} \|_\mathcal{H}^2 - x^2 \right\rvert \leq \frac{2 \hat{\varsigma}_{\omega}}{\sqrt{n}} \right\rbrace,\textnormal{ and }\hspace{2em}  \mathcal{C}_{\omega,n}^{0} \coloneqq\begin{cases}
        \{0\}, &\text{ if } \| \hat{\gamma}_{\omega}^\mathcal{K} \|_\mathcal{H} \leq \sqrt{\hat{\xi}_{\omega}/n} \quad \\
        \emptyset, &\mbox{ otherwise.}
    \end{cases}
\end{align*}
An asymptotically valid $(1-\alpha)$ level confidence set for the variable importance measure $\| \gamma_\omega^\mathcal{K} ( P_0 ) \|_\mathcal{H}$ is given by $\mathcal{C}_{\omega,n}^{>0}\cup \mathcal{C}_{\omega,n}^{0}$, and its convex hull is an asymptotically valid confidence interval.
\end{theorem}

The consistent estimator $\hat{\varsigma}_{\omega}$ can be obtained via the bootstrap, as the empirical quantile of $\lvert \langle \hat{\gamma}_{\omega}^\mathcal{K}, \mathbb{H}_{\omega,n}^\sharp \rangle_\mathcal{H} \rvert$. Alternatively, it can be taken to be the $1-\alpha$ quantile of a half-normal distribution with parameter $\widehat{\sigma}_\omega^2\coloneqq P_n\langle \phi_n^\omega,\hat{\gamma}^\mathcal{K}_\omega\rangle_{\mathcal{H}}^2$.

To prove the above theorem, we first show that---under the stated conditions---$\sigma_\omega^2=0$ if and only if the variable importance measure is equal to zero. We then separately consider two cases to establish coverage of the union confidence set, $\mathcal{C}_{\omega,n}^{>0}\cup \mathcal{C}_{\omega,n}^{0}$. In the first, the variable importance measure is nonzero, so that $\sigma_\omega^2>0$. In this case, an application of the delta method readily shows $\mathcal{C}_{\omega,n}^{>0}$ covers the variable importance measure at the desired asymptotic level. When the variable importance measure is $0$, the latter set in the union, $\mathcal{C}_{\omega,n}^{0}$, will cover the variable importance measure at the desired level; this follows since $\mathcal{C}_{\omega,n}^{0}$ contains $0$ if and only if the test of no importance from Section~\ref{Section3.3} fails to reject the null.

We have described two strategies for constructing a confidence interval for the variable importance measure $\|\hat{\gamma}_{\omega}^\mathcal{K}\|_\mathcal{H}$: one based on the reverse triangle inequality, and the other based on the delta method. The main advantage of the first is its ease of implementation. A disadvantage is that it tends to yield conservative confidence intervals. In contrast, the delta method approach yields asymptotically exact confidence intervals, but requires separate handling of the null case. A limitation common to both methods is that the resulting intervals are not centered on the original parameter of interest---the norm of the variable importance measure---but rather on a smoothed approximation given by its RKHS embedding.

\section{Numerical experiments} \label{Section4}

We conducted a simulation study to evaluate the performance of the proposed test of no variable importance for a variable $X_1$. We considered three measures mentioned earlier: LOO, KOI, and Shapley values. 

\subsection{Experiment 1: 5-dimensional covariate}

This simulation set-up is partially inspired by the data generating process `DGP 1' from \cite{Hines2025}. The observed data consist of $(X, A, Y)$, where $X \coloneqq (X_1, X_2, X_3, X_4, X_5)$ is a covariate vector sampled from a Gaussian copula with covariance matrix $\Sigma$. The off-diagonal elements of $\Sigma$ are set to $\sigma \in \lbrace 0, 0.4, 0.8 \rbrace$. The treatment satisfies $A| X\sim \textnormal{Bernoulli}(\mathrm{expit}[-0.4 X_1 + 0.1 X_1 X_2])$, where $\mathrm{expit}(x) = 1 / \lbrace 1 + \exp(-x) \rbrace$. The outcome satisfies $Y|A,X\sim \mathcal{N}(X_1 X_2 + 2 X_2^2 - X_1 + A\tau, 1)$, where the CATE is defined as $\tau = \beta g(X_1) + f \left( X_2, X_3, X_4, X_5 \right)$, with $\beta \in \lbrace 0, 1, 5 \rbrace$ and an auxiliary function $f \left( X_2, X_3, X_4, X_5 \right) = 0.2 \left( X_2^2 + X_3 - 2 X_3 X_4 + 4 X_5 \right)$. We consider two kinds of alternatives: smooth, where $g(X_1) = X_1$, and rough, where $g(X_1) = \sin(5 \pi X_1)$.

In all experiments, the proposed test was implemented using twofold cross-fitting to generate pseudo-outcomes, which were then used to estimate the CATE on the full dataset. We conducted $1000$ Monte Carlo simulations for sample sizes $n \in \lbrace 250, 500, 1000, 2000 \rbrace$ and $5000$ bootstrap repetitions. The nuisance functions $g_P$ and $\mu_P$ were estimated using random forests (\texttt{ranger} with default parameters) \citep{Wright2017}, while $\nu_P$ was estimated using kernel ridge regression \citep{Wainwright2019}. Although the regularization parameter $\lambda$ in KRR is typically selected by cross-validation, in our simulations it was set to $\sqrt{\log(n)/n}$ for computational efficiency. The kernel bandwidth was selected using the median heuristic \citep{Garreau2018}.

Figure \ref{figure2a} presents the results of the simulation study. Each row corresponds to a different variable importance measure, with `Shapley approx.' referring to an approximation of the Shapley value computed by sampling 40 permutations of the covariates. The left column displays the empirical Type I error rates under the null hypothesis of no importance for $X_1$, showing that our method effectively controls Type I error across a range of correlation parameter values $\sigma$. The middle column illustrates the empirical rejection probabilities under the smooth alternative $g(X_1) = X_1$ for $\beta = 1$ and $\beta = 5$, while the right column shows the corresponding results under the rough alternative $g(X_1) = \sin(5 \pi X_1)$. Although the method exhibits low power for small sample sizes, power improves substantially as sample size increases. There is no meaningful inflation of Type I error or loss of power from using the Monte Carlo approximation of the Shapley values rather than the exact Shapley values.

Notably, for the KOI and Shapley variable importance measures, the simulation setup with $\beta = 0$ and $\sigma > 0$ does not represent a true null scenario but rather an alternative, due to correlations between variables. In Appendix \ref{app:experiment3}, we discuss a simulation setting where setups with $\beta = 0$ represent the null scenario for all variable importance measures and for all values of $\sigma$. Just as in the simulation presented in the main text, our method effectively controls the Type I error across all variable importance measures and simulation configurations---see Appendix \ref{app:experiment3} for further discussion.

To compare our method with existing approaches, we report results obtained using the treatment effect variable importance measures (TE-VIMs) proposed by \cite{Hines2025} under the simulation setup of Experiment 1. To compute the TE-VIM results, we used the publicly available code provided by the authors for the leave-one-out (LOO) measure, which we adapted for importance testing as described in Section 2.3.2 of \cite{Hines2025}. Results are presented for two nuisance estimation strategies: generalized additive models (GAMs), as used in the original paper, and random forests (RFs), which aligns with the strategy we used for our methods. Figure \ref{figureOH} presents the comparison. TE-VIM appears to perform substantially better when GAMs are used for nuisance parameter estimation. However, even in this setting, the results suggest that TE-VIM does not adequately control the Type I error at the considered sample sizes. Moreover, it exhibits lower power under smooth alternatives but higher power under rough alternatives.

\begin{figure}[tb]
\begin{center}
 \includegraphics[width=\textwidth]{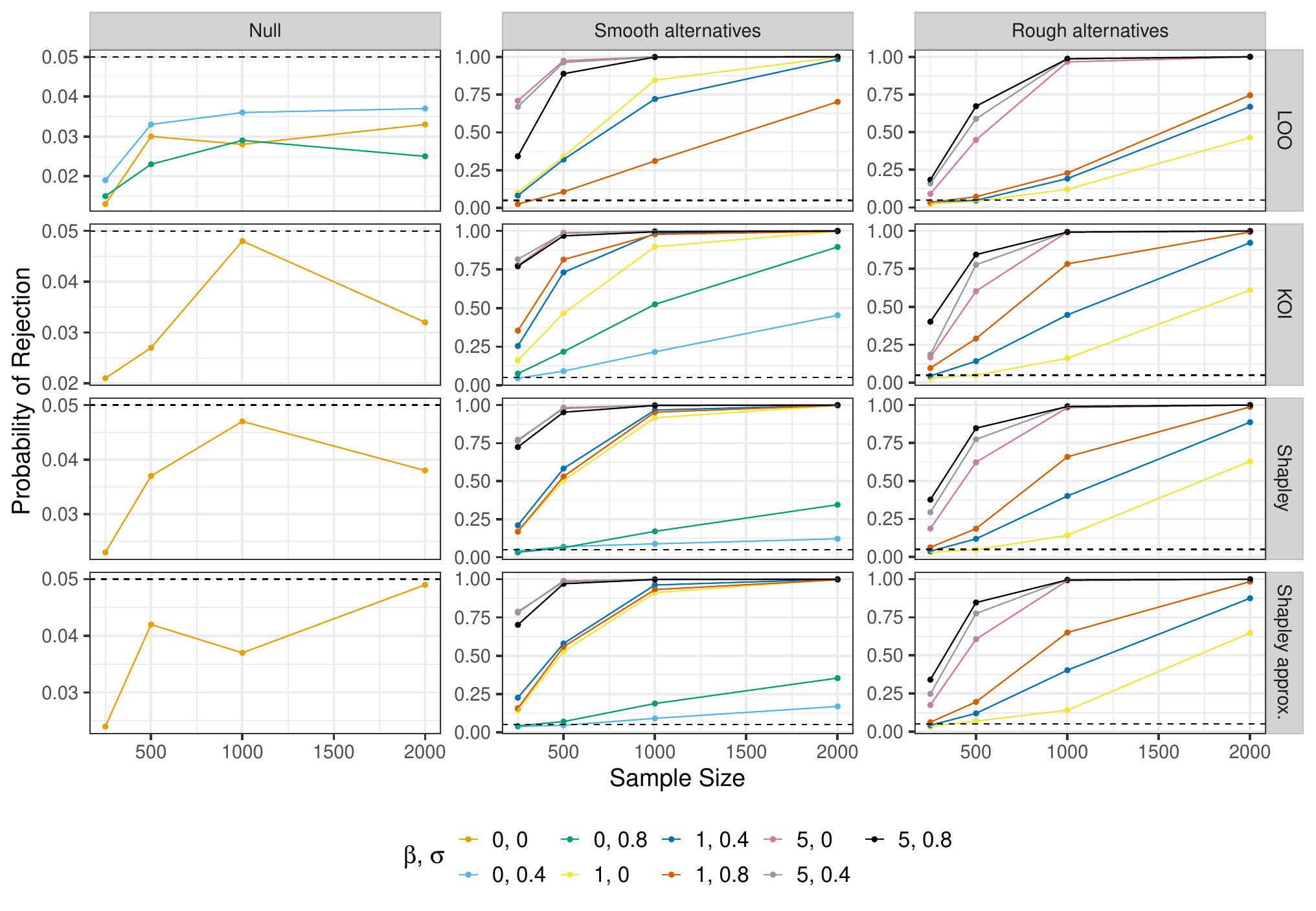}  
\end{center}
\caption{Simulation results from Experiment 1 comparing different variable importance measures (rows) across different data-generating processes (columns).} \label{figure2a}
\end{figure}

\begin{figure}[tb]
\begin{center}
 \includegraphics[width=\textwidth]{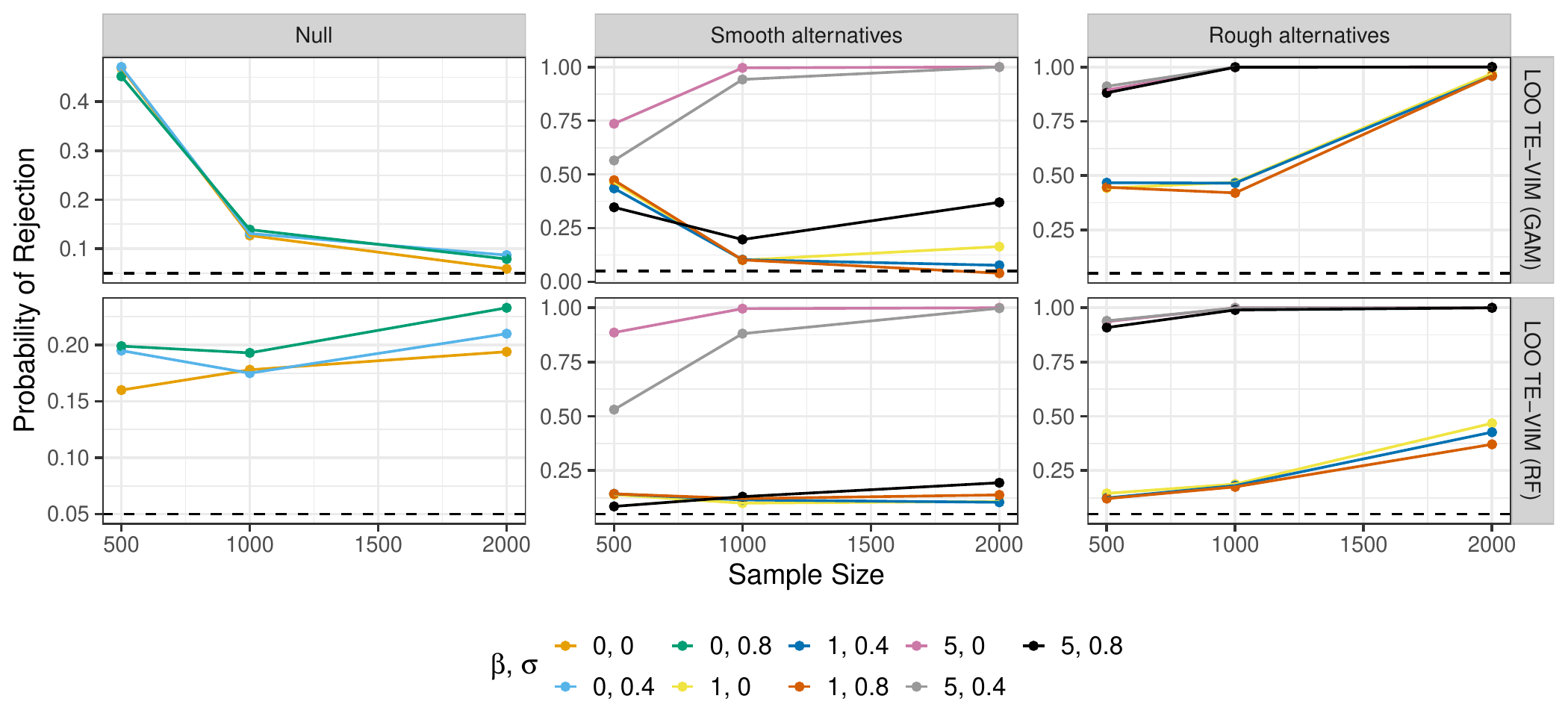}  
\end{center}
\caption{Simulation results from Experiment 1 comparing the leave-one-out (LOO) variable importance measure computed using the TE-VIM approach of \cite{Hines2025}, across different choices of nuisance parameter estimation (rows) and data-generating processes (columns).} \label{figureOH}
\end{figure}

\subsection{Experiment 2: 10-dimensional covariate}

This simulation scenario mirrors that of Experiment 1, but incorporates a higher-dimensional covariate. The observations consist of $(X, A, Y)$, where $X \coloneqq (X_1, X_2, \ldots, X_{10})$ is a 10-dimensional covariate vector drawn from a Gaussian copula with covariance matrix $\Sigma$. The off-diagonal elements of $\Sigma$ are set to $\sigma \in \lbrace 0, 0.4, 0.8 \rbrace$. The data-generating process for $(A,Y)|(X_1,\ldots,X_5)$ is the same as in Experiment 1 and, conditionally on $(X_1,\ldots,X_5)$, $(A,Y)$ is independent of $(X_6,\ldots,X_{10})$. All nuisance estimation settings were the same, except that, to account for the increased dimensionality, the regularization parameter $\lambda$ in KRR was set to $\log^2(n)/n^{1/2}$.

Figure~\ref{figure2b} displays the results. Similarly as in Experiment 1, although the method exhibits low power for small sample sizes, power improves substantially as sample size increases. Exact Shapley values were not computed due to computational constraints. However, the Shapley approximation---which again used 40 permutations of the covariates---showed reasonable Type I error control and power. 

\begin{figure}[tb]
\begin{center}
 \includegraphics[width=\textwidth]{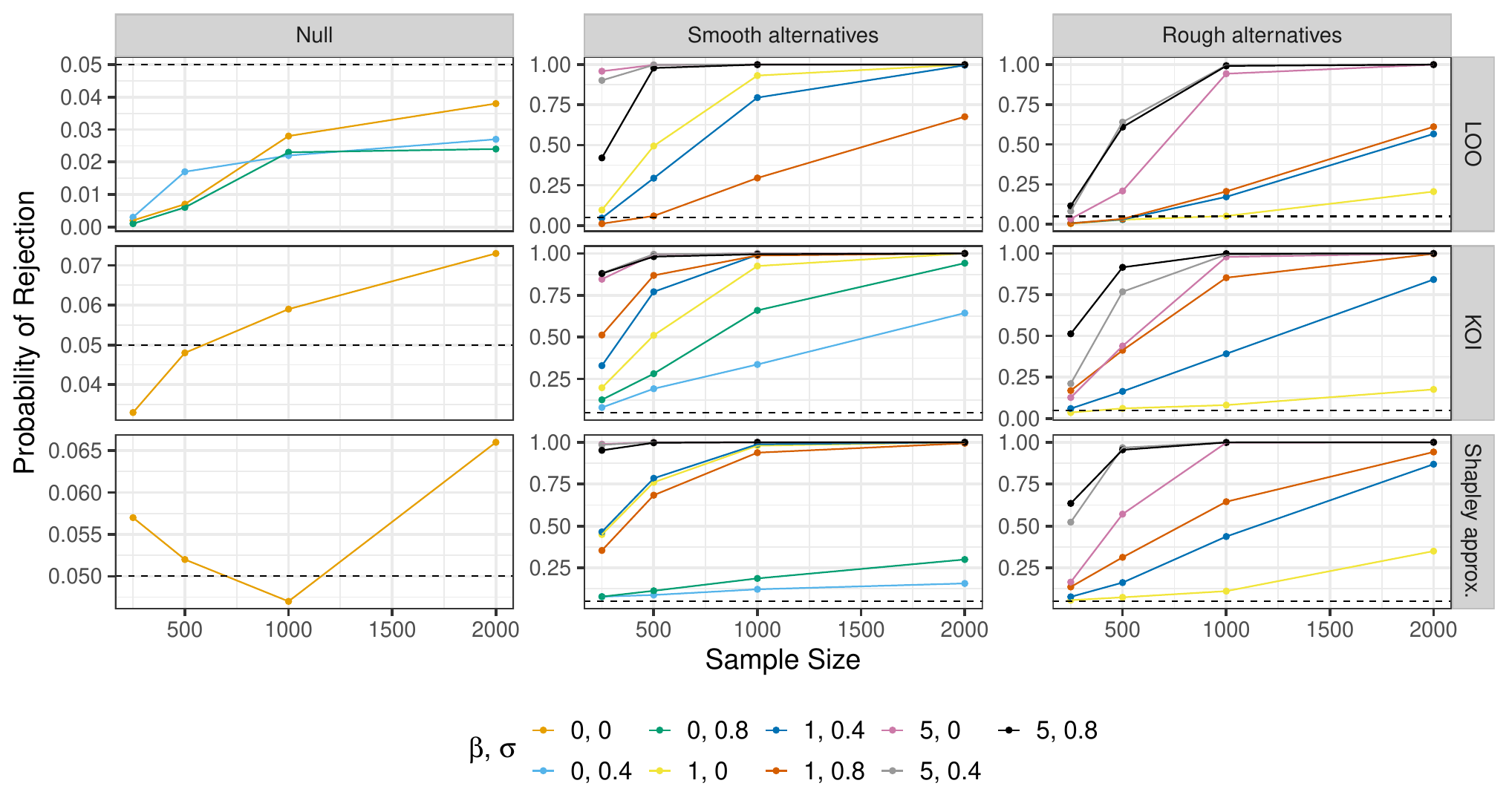}  
\end{center}
\caption{Simulation results from Experiment 2 comparing different variable importance measures (rows) across different data-generating processes (columns).} \label{figure2b}
\end{figure}

\section{Application to infectious disease prevention strategies} \label{Section5}

A global Phase 3 randomized, placebo-controlled efficacy trial VAT00008 (NCT04904549) was conducted among adults aged 18 and older to compare two COVID-19 vaccine regimens to placebo. Vaccinations or placebo were administered on Day 1 (D01) and Day 22 (D22). The study was conducted in two stages: Stage 1 assessed a monovalent vaccine containing the original D614 SARS-CoV-2 Spike protein, while Stage 2 evaluated a bivalent formulation including both D614 and the Beta SARS-CoV-2 Spike protein \citep{Dayan2023a,Dayan2023b}. Stage 2 of the trial enrolled 13,002 participants between October 19, 2021, and February 15, 2022. 

Our analysis focuses on the per-protocol Stage 2 cohort that had evidence of previous SARS-CoV-2 infection at enrollment, where per-protocol means both assigned doses were received without major protocol violations. The objective of our analysis is to examine how various baseline characteristics influence the effect of COVID-19 vaccination on a key immunological outcome: the log-transformed neutralizing antibody titer ($\log_{10}$ nAb ID50) against the BA.1 (B.1.1.529.1) strain of SARS-CoV-2, measured at Day 43. The treatment variable $A$ reflects membership in the per-protocol vaccine ($n = 305$) or per-protocol placebo ($n = 292$) group. The baseline covariates $X$ include demographic factors (age, sex, body mass index, and country) and immunological markers ($\log_{10}$-transformed baseline neutralizing antibody titer against the BA.1 strain and $\log_{10}$-transformed concentration of IgG binding antibodies against the Omicron (B.1.1.529) strain). The dataset includes 597 individuals from Ghana, Kenya, Honduras, India, Japan, Mexico, and the USA, all of whom have complete data for the variables listed above. 

Figure \ref{figure3} displays the estimated variable importance---defined as the norm of the RKHS embedding of the local variable importance measures: KOI, LOO, and Shapley values---along with nominal $95\%$ confidence intervals. Based on the KOI and Shapley values, there is evidence that two variables, `Country' and `BA.1 NAb', have non-zero importance. Supplementary results, including p-values and descriptive graphics, are provided in Table \ref{table1} and Figure \ref{figure4} in Appendix \ref{app:resultsempirical}. When testing the null hypothesis of no importance for multiple variables, it is advisable to apply a multiple testing adjustment. Accordingly, we report Benjamini–Hochberg–adjusted p-values in Table \ref{table1}. The result for BA.1 NAb shows a dose-response relationship wherein individuals with lower baseline titers have greater vaccine effect on Day 43 titer. This occurs in part because the immune system has a cap on the number of SARS-CoV-2 specific antibodies that it can make, as it must balance its caloric-cost investment into antibody production against myriad microbial pathogen threats (not only SARS-CoV-2), such that individuals starting with high levels of pre-vaccination neutralizing antibodies have less room for vaccination to further increase antibody levels. 

The vaccine effect was estimated to be highest in Ghana and the US, and lowest in Japan, Kenya, and Mexico, with the placebo arm distribution of Omicron BA.1 nAb titer notably higher in the latter three countries (see Figure \ref{figure4} in Appendix \ref{app:resultsempirical} for more details). One potential explanation is a greater frequency of prior infection with Omicron SARS-CoV-2 (compared to with Delta SARS-CoV-2) in these countries, given that the studied titer is against an Omicron strain.

\begin{figure}[tb]
\begin{center}
 \includegraphics[width=\textwidth]{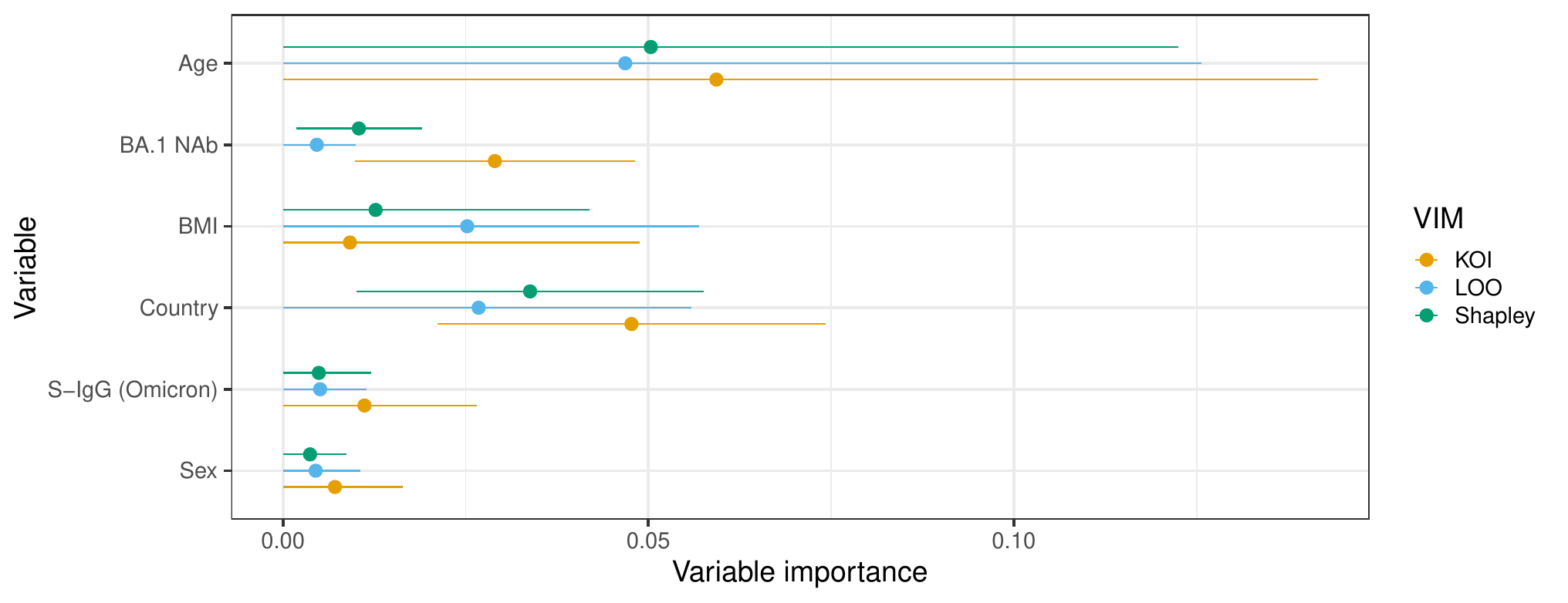}  
\end{center}
\caption{Variable importance, defined as the norm of the RKHS embedding of the local variable importance measures: KOI, LOO, and Shapley values, along with nominal $95\%$ confidence intervals. Intervals are truncated at zero because variable importance measures are non-negative.} 
\label{figure3}
\end{figure}

\section{Discussion} \label{Section6}

A key strength of our approach lies in its ability to perform inference on the parameter regularized through an RKHS embedding, facilitating robust hypothesis testing in the challenging context of zero-importance evaluation. While the inference is tailored to the regularized parameter, tests of the null hypothesis of no importance remain valid and equivalent when applied to the original parameter provided the kernel of the RKHS is universal.

Beyond evaluating global variable importance, our hypothesis testing strategy can be adapted to test the localized null
\begin{equation*}
    H_0 \colon \gamma_\omega (P)(X) = 0 \quad \text{for} \quad X \in \widetilde{\mathcal{X}} \subset \mathcal{X},
\end{equation*}
where $\widetilde{\mathcal{X}}$ denotes a specified region of the covariate space. To handle this case, one can use the kernel embedding $\widetilde{\gamma}_\omega^\mathcal{K} (P) (x) \coloneqq \int_\mathcal{X} \mathcal{K} (x, x^\prime)\, \gamma_\omega (P) (x^\prime)\, P_{X} (dx^\prime\,|\, X\in\widetilde{\mathcal{X}})$. This parameter can be shown to have EIF $\widetilde{\phi}_{P}^\omega(z)(x')=1(x\in\widetilde{\mathcal{X}})\phi_{P(\,\cdot\,
\mid X\in\widetilde{\mathcal{X}})}^\omega(z)(x')/P(X\in\widetilde{\mathcal{X}})$ \citep[][Ch. 25.5]{vanderVaart1998}, where $\phi_{P(\,\cdot\,
\mid X\in\widetilde{\mathcal{X}})}^\omega$ is the EIF of $\gamma_\omega^{\mathcal{K}}$ at the conditional distribution of $Z\mid X\in\widetilde{\mathcal{X}}$ under $P$. The resulting estimation and inference procedure then follow analogously to the one from this work, using this EIF to debias an initial plug-in estimator.

The proposed variable importance framework, developed here in the context of heterogeneous treatment effects, can be readily extended to traditional prediction settings. The parameter of interest---analogous to (\ref{eqn1})---is then defined as
\begin{align*}
\bar{\gamma}_\omega (P) (x) \coloneqq \sum_{\mathcal{V} \subseteq [d]} \omega_\mathcal{V}  \, \bar{\nu}_P ( \mathcal{V}) (x_\mathcal{V}),
\end{align*}
where $\bar{\nu}_P (\mathcal{V}) (x) \coloneqq \mathbb{E}_P ( Y \mid X_\mathcal{V} = x_\mathcal{V} )$ denotes the conditional expectation of the outcome given a subset of covariates. The estimation procedure and statistical inference follow from the same strategy in Section~\ref{Section3}, by simply replacing $\psi_P$ by $\bar{\psi}_P(z) \coloneqq y$, $\nu_P$ by $\bar{\nu}_P$, and $\gamma_\omega^{\mathcal{K}}$ by $\bar{\gamma}_\omega^{\mathcal{K}}$.

Given that our procedures rely on kernel-based methods, scalability with respect to sample size is an important consideration. For larger sample sizes, one practical approach is to employ Nystr\"{o}m approximation \citep{Williams2001, Rasmussen2006} or randomly pivoted Cholesky approximation \citep{Chen2025} of the kernel matrix, which significantly reduce the computational burden by avoiding operations on large kernel matrices.

There are several promising directions for future work. One limitation of the Shapley values discussed in this paper is their high computational cost. While we outline some strategies to mitigate these limitations in Appendix \ref{app:computationalconsiderations}, further research is needed to scale the proposed approach to settings with a larger number of covariates. In our analysis, we focused on the Gaussian kernel with bandwidth selected using the median heuristic. Future work should explore the impact of different kernel functions and bandwidth selection methods. Although the proposed framework is quite flexible, it does not currently accommodate some widely used variable importance methods, such as Local Interpretable Model-agnostic Explanations (LIME) \citep{Ribeiro2016}. Investigating these methods in the context of heterogeneous treatment effects would be a valuable avenue for future research.

\section*{Supporting Information}

Additional supporting material, including proofs of the main results, is provided in the Appendix. The \texttt{R} code for the numerical experiments is available at \texttt{https://github.com/pmorzywolek/VIMsimulations}.

\section*{Acknowledgments}

The authors thank Antonio Olivas-Martinez for helpful discussions.  Research reported in this work was funded through a Patient-Centered Outcomes Research Institute (PCORI) Award (ME-2024C2-39990) and two National Institutes of Health (NIH) Awards (5UM1AI068635-20, DP2LM013340-01). Pawe{\l} Morzywo{\l}ek was partially supported through the Pacific Institute for the Mathematical Sciences (PIMS) and eScience Institute at the University of Washington. The statements in this work are solely the responsibility of the authors and do not necessarily represent the views of PCORI, its Board of Governors or Methodology Committee, NIH, or PIMS. 
The VAT00008 data were generated from federal funds from the Department of Health and Human Services and from Sanofi.

\section*{Data Availability Statement}

Qualified researchers can request access to participant-level data and related study documents, including the statistical analysis report. Participant-level data will be anonymized to protect the privacy of trial participants. Further details on Sanofi’s data sharing criteria, eligible studies, and process for requesting access can be found at \texttt{https://vivli.org/}.

\bibliography{VIM}
\bibliographystyle{abbrvnat}

\appendix

\setcounter{equation}{0}
\renewcommand{\theequation}{S\arabic{equation}}
\setcounter{theorem}{0}
\setcounter{figure}{0}
\setcounter{table}{0}
\setcounter{lemma}{0}
\setcounter{corollary}{0}
\renewcommand{\thetheorem}{S\arabic{theorem}}
\renewcommand{\thecorollary}{S\arabic{corollary}}
\renewcommand{\thelemma}{S\arabic{lemma}}
\renewcommand{\thefigure}{S\arabic{figure}}
\renewcommand{\thetable}{S\arabic{table}}
\renewcommand{\thealgorithm}{S\arabic{algorithm}}

\section*{\LARGE Appendices}

\DoToC

\section{Review of semiparametric efficiency theory} \label{SectionA1}

Consider a statistical model $\mathcal{M}$ dominated by a $\sigma$-finite measure $\lambda$. A submodel $\lbrace P_\epsilon \colon \epsilon \in [ 0, \delta ) \rbrace \subset \mathcal{M}$ is called quadratic mean differentiable at $P$ if and only if there exists a score function $s\in L^2 ( P )$ such that 
\begin{align*}
    \| p_\epsilon^{1/2} - p^{1/2} - \epsilon s p^{1/2} / 2 \|_{L^2 ( \lambda )} = o ( \epsilon ), 
\end{align*}
where $p_\epsilon^{1/2} = \sqrt{d P_\epsilon/d \lambda}$ for $\epsilon \geq 0$, and $p^{1/2} = \sqrt{dP/d\lambda}$. Let $\mathscr{P} ( P_0, \mathcal{M}, s )$ denote the set of models that are quadratic mean differentiable at $P$ with score function $s$. The set of corresponding score functions, $\lbrace s \in L^2 ( P ) \colon \mathscr{P} ( P_0, \mathcal{M}, s ) \neq \emptyset \rbrace$, is called the tangent set, while the closure of its linear span is called the tangent space of $\mathcal{M}$ at $P$, denoted by $\dot{\mathcal{M}}_P$. Note that for all $s \in \dot{\mathcal{M}}_P$, the function $s$ is bounded and satisfies $Ps = \int s dP = 0$. Therefore, the space $L_0^2 ( P ) \coloneqq \lbrace h \in L^2 ( P ) \colon Ph = 0 \rbrace$ is the largest possible tangent space at $P$. A statistical model $\mathcal{M}$ is called a (locally) nonparametric if its tangent space equals $L_0^2 ( P )$ at all $P$. 

Let $\Psi \colon \mathcal{M} \rightarrow \mathcal{G}$ be a parameter of interest. The action space $\mathcal{G}$ may be either finite- or infinite-dimensional. In what follows, we assume that it is an infinite-dimensional separable Hilbert space. The parameter $\Psi$ is said to be pathwise differentiable at $P$ if and only if there exists a continuous linear operator $\dot{\Psi}_P \colon \dot{\mathcal{M}}_P \rightarrow \mathcal{G}$, called the local parameter, such that for all submodels in the tangent set, $\lbrace P_\epsilon \colon \epsilon \in [ 0,\delta ) \rbrace \in \mathscr{P} ( P_0, \mathcal{M}, s )$, we have 
\begin{align*}
    \| \Psi ( P_\epsilon ) - \Psi ( P ) - \epsilon \dot{\Psi}_P ( s ) \|_\mathcal{G} = o ( \epsilon ). 
\end{align*}
The Hermitian adjoint of $\dot{\Psi}_P$, denoted by $\dot{\Psi}_P^\ast \colon \mathcal{G} \rightarrow \dot{\mathcal{M}}_P$, is referred to as the efficient influence operator. Furthermore, the parameter $\Psi$ is said to admit an efficient influence function (EIF), also known as a canonical gradient, $\phi_P \colon \mathcal{Z} \rightarrow \mathcal{G}$ if there exists a $P$-probability-one set $\mathcal{Z}^\prime$ such that $\dot{\Psi}_P^\ast (h) (z) = \langle h, \phi_P (z) \rangle_\mathcal{G}$ for all $( h, z ) \in \mathcal{G} \times \mathcal{Z}^\prime$. If it exists, the EIF is unique. Figure \ref{figure2} illustrates the basic notions of the semiparametric theory. See \cite{Luedtke2024} for a more comprehensive review. 

The development of semiparametric theory has primarily focused on the estimation of finite-dimensional parameters \citep{Bickel1993, Pfanzagl1990, vanderLaan2003}. When such a parameter is pathwise differentiable, it also admits an EIF, which is the Riesz representer of the pathwise derivative. The EIF is used as a building block for constructing efficient estimators, such as those derived via one-step estimation \citep{Bickel1993, vanderLaan2003}, estimating equations \citep{vanderVaart1998}, double/debiased machine learning \citep{chernozhukov2018double} or targeted minimum loss-based estimation (TMLE) \citep{vanderLaan2006, vanderLaan2011}. In contrast, the semiparametric theory for infinite-dimensional parameters has been less explored. As with finite-dimensional cases, the primary focus remains on pathwise differentiable parameters. However, in infinite-dimensional settings, pathwise differentiability does not guarantee the existence of an EIF.

As noted by \cite{Luedtke2024}, certain parameters may not be pathwise differentiable when viewed as finite-dimensional objects (e.g., point evaluations of a function), yet they can be pathwise differentiable as infinite-dimensional parameters. An example of such a parameter is $\gamma_\omega ( P )$, as defined in (\ref{eqn1}). One can view $\gamma_\omega ( P )$ either as an element of the function space $L_2 (P)$ or as its pointwise evaluation $\gamma_\omega (P)(x)$ for some $x$. Unfortunately, this pointwise evaluation is not pathwise differentiable, rendering standard semiparametric tools inapplicable. A common workaround involves kernel regression methods \citep{Kennedy2017, Zimmert2019}. 

In contrast, we adopt an alternative approach by treating $\gamma_\omega (P)$ as an infinite-dimensional parameter, which turns out to be pathwise differentiable because it represents a weighted sum of heterogeneous treatment effects---a parameter shown by \cite{Luedtke2024} to be pathwise differentiable. However, as mentioned above, even in this case the existence of an EIF is not automatic. To address this, \cite{Luedtke2024} propose a regularized parameter for which the EIF does exist. By Mercer's Theorem \citep{berlinet2003reproducing}, the proposed regularization can be interpreted as an embedding into a reproducing kernel Hilbert space (RKHS) $\mathcal{H}$.

\begin{figure} 
\begin{center}
 \includegraphics[width=\textwidth]{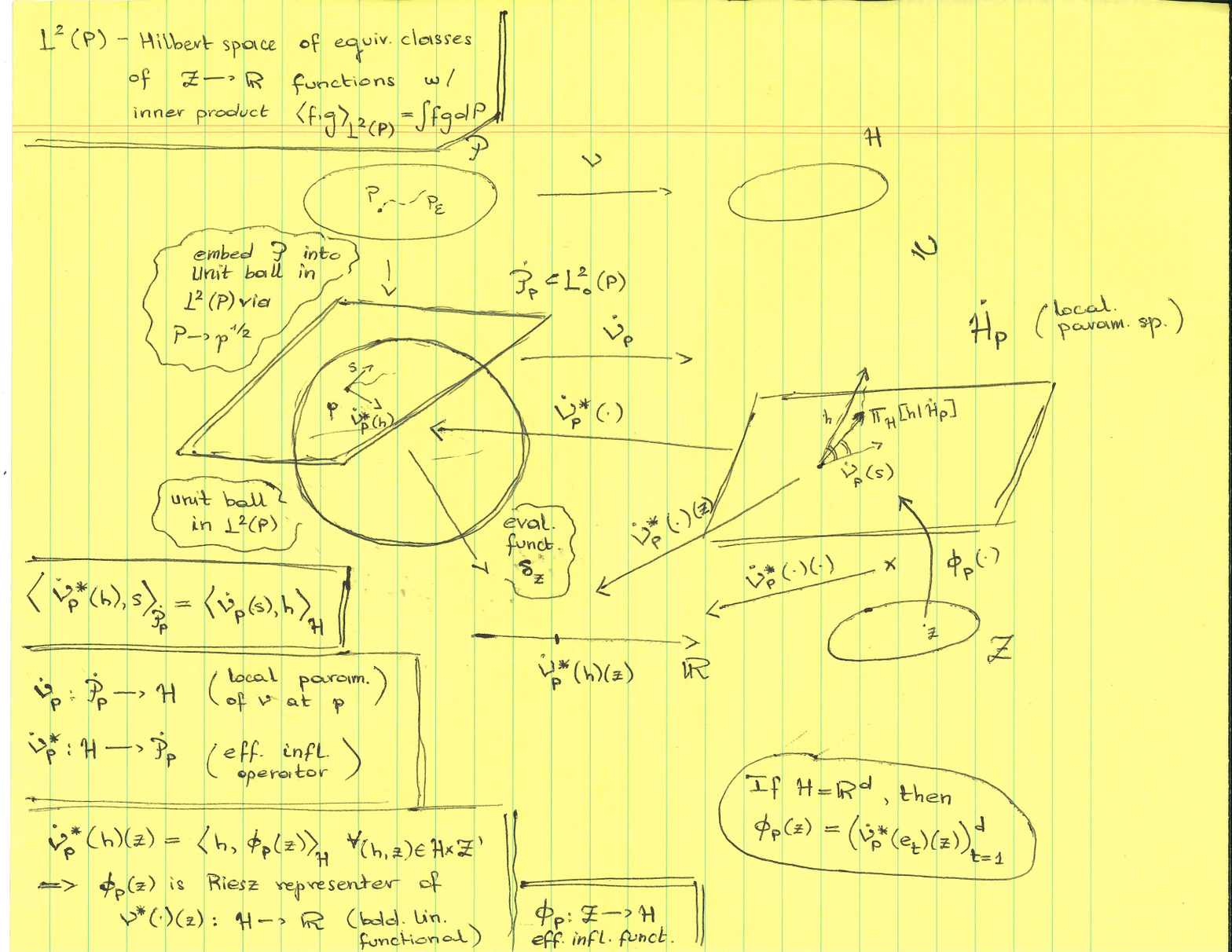}    
\end{center}
\caption{Overview of the semiparametric theory for Hilbert-valued parameters.} \label{figure2}
\end{figure}

\section{Shapley values axioms in the context of variable importance for heterogeneous treatment effects} \label{SectionA2}

Shapley values \citep{Shapley1953, Lundberg2017, Covert2020, Williamson2020, Verdinelli2023} are a widely used measure of variable importance, originally developed in cooperative game theory. They were initially proposed as a method for fairly distributing payouts among a group of cooperative players. More recently, Shapley values have been adopted in explainable machine learning to evaluate feature importance in predictive models. In the context of quantifying variable importance for heterogeneous treatment effects, the Shapley value for the $i$-th feature is defined for $\beta_\mathcal{V} \coloneqq \frac{1}{d} \binom{d-1}{\lvert \mathcal{V} \rvert}^{-1}$ as:
\begin{align*}
    \gamma_{SHAP,i} (P)(x) &\coloneqq \sum_{\mathcal{V} \subseteq [d] \backslash \lbrace i \rbrace} \beta_\mathcal{V} \lbrace \nu_P (  \mathcal{V} \cup \lbrace i \rbrace ) ( x_{\mathcal{V} \cup \lbrace i \rbrace} ) - \nu_P ( \mathcal{V} ) ( x_\mathcal{V} ) \rbrace.
\end{align*}
We are defining the Shapley values with respect to a fixed version of the CATE $\nu_P$, which is unique up to $P$-almost sure equivalence. The Shapley value is the only measure of variable importance that satisfies a set of axioms---namely, efficiency (distinct from its statistical meaning), symmetry, null player, and linearity---which are considered desirable properties of a payout function in cooperative game theory \citep{Shapley1953}. These axioms have also meaningful interpretations in the context of variable importance for heterogeneous treatment effects. In this setting, the efficiency axiom states that
    \begin{align*}
        \sum_{i=1}^d \gamma_{SHAP,i} (P) (x) = \mathbb{E} (Y^1 - Y^0 \mid X = x) - \mathbb{E} (Y^1 - Y^0),
    \end{align*}
which means that the sum of the individual contributions of variables $X_i$, $i=1,\dots,d$, to the heterogeneity of the treatment effect---quantified by their Shapley values---equals the total heterogeneity of the treatment effect attributable to the entire covariate set $X$. The symmetry axiom states that if $\nu_P ( \mathcal{V} \cup \lbrace i \rbrace )(x) = \nu_P ( \mathcal{V} \cup \lbrace j \rbrace )(x)$ for all subsets $\mathcal{V}$ that contain neither $i$ nor $j$, then the contribution of covariates $X_i$ and $X_j$ to treatment effect heterogeneity is indistinguishable, and thus both should be assigned the same Shapley value when $X=x$. The null axiom states that if the effect modification due to the covariate $X_i$ conditional on the covariates $X_\mathcal{V}$ is $0$ for all subsets of covariates $X_\mathcal{V}$, then the Shapley value assigned to $X_i$ should also be zero. Finally, the linearity axiom states that Shapley values are linear operators: if treatment effect heterogeneity is measured by a linear combination of outcomes---e.g., for outcome variables $Y$ and $W$, considering the CATE on $\alpha Y + \beta W$ for $\alpha, \beta \in\mathbb{R}$---then the resulting Shapley values are the corresponding linear combination of the Shapley values computed separately on $Y$ and $W$.

\section{Computation of the test statistic for the statistical test of no importance} \label{SectionA3}

\subsection{Test statistic for the test of no importance for a generic weighted parameter}

The parameter of interest is RKHS-valued function $\gamma^\mathcal{K}_\omega (P) \coloneqq \sum_{\mathcal{V} \subseteq [d]} \omega_\mathcal{V} \nu_P^\mathcal{K} ( \mathcal{V})$, where $\omega_\mathcal{V} \in \mathbb{R}$ are weights.
For $\hat{\nu}^\mathcal{K}_\mathcal{V}$ the one-step estimator of $\nu_0^\mathcal{K} ( \mathcal{V})$, the one-step estimator of $\gamma^\mathcal{K}_\omega (P_0)$ is of the form:
\begin{align*}
    \hat{\gamma}^\mathcal{K}_\omega (\cdot) &\coloneqq \gamma^\mathcal{K}_\omega ( \hat{P}_n ) (\cdot) + P_n \phi_{n}^\omega (\cdot) = \sum_{\mathcal{V} \subseteq [d]} \omega_\mathcal{V} \hat{\nu}^\mathcal{K}_\mathcal{V} (\cdot).
\end{align*}
\noindent
We wish to compute an explicit formula for the test statistic $\| \hat{\gamma}^\mathcal{K}_\omega\|_\mathcal{H}^2$, which is used to test for the importance of a given variable in Section \ref{Section3.3}. To this end, note that
\begin{align*}
    \nonumber
    \| \hat{\gamma}^\mathcal{K}_\omega \|_\mathcal{H}^2 &= \left\langle \hat{\gamma}^\mathcal{K}_\omega, \hat{\gamma}^\mathcal{K}_\omega \right\rangle_\mathcal{H} = \sum_{\mathcal{V} \subseteq [d]} \sum_{{\mathcal{V}^\prime} \subseteq [d]} \omega_\mathcal{V}  \omega_{\mathcal{V}^\prime} \left\langle \hat{\nu}^\mathcal{K}_\mathcal{V}, \hat{\nu}^\mathcal{K}_{\mathcal{V}^\prime} \right\rangle_\mathcal{H}.
\end{align*}
Let us consider first the one-step estimator $\hat{\nu}^\mathcal{K}_\mathcal{V}$:
\begin{align} \label{eqn7} 
    \nonumber
    \hat{\nu}_\mathcal{V}^\mathcal{K} (\cdot) &= \frac{1}{n} \sum_{i=1}^n \left[ \kappa_{n,\mathcal{V}} (X_{\mathcal{V},i}) (\cdot) \psi_n ( Z_i ) + \left\lbrace \mathcal{K} ( \cdot, X_i ) - \kappa_{n,\mathcal{V}} (X_{\mathcal{V},i}) (\cdot) \right\rbrace \nu_n (\mathcal{V}) ( X_{\mathcal{V},i} ) \right] \\
    &= \frac{1}{n} \sum_{i=1}^n \left[\alpha_{\mathcal{V},i} \mathcal{K} ( \cdot, X_i ) + \beta_{\mathcal{V},i} \kappa_{n,\mathcal{V}} (X_{\mathcal{V},i})(\cdot)\right],
\end{align}
where $\alpha_{\mathcal{V},i} \coloneqq \nu_n (\mathcal{V}) ( X_{\mathcal{V},i} )$ and $\beta_{\mathcal{V},i} \coloneqq \psi_n ( Z_i ) - \nu_n (\mathcal{V}) ( X_{\mathcal{V},i} )$. Furthermore, let us consider a single inner product in the expression above. Substituting the expression (\ref{eqn7}) we obtain
\begin{align} \label{eqn9}
    \nonumber
    \left\langle \hat{\nu}^\mathcal{K}_{\mathcal{V}}, \hat{\nu}^\mathcal{K}_{\mathcal{V}^\prime} \right\rangle_\mathcal{H} &= \frac{1}{n^2} \sum_{i=1}^n \sum_{j=1}^n \left\langle \alpha_{\mathcal{V},i} \mathcal{K} (\cdot, X_i ) + \beta_{\mathcal{V},i} \kappa_{n,\mathcal{V}} (X_{\mathcal{V},i}) (\cdot), \alpha_{\mathcal{V}^\prime,j} \mathcal{K} ( \cdot, X_j ) + \beta_{\mathcal{V}^\prime,j} \kappa_{n,\mathcal{V}^\prime} (X_{\mathcal{V}^\prime,j}) (\cdot) \right\rangle_\mathcal{H} \\
    \nonumber
    &= \frac{1}{n^2} \sum_{i=1}^n \sum_{j=1}^n \Big[\alpha_{\mathcal{V},i} \alpha_{\mathcal{V}^\prime,j} \left\langle \mathcal{K} ( \cdot, X_i ), \mathcal{K} ( \cdot, X_j ) \right\rangle_\mathcal{H} + \beta_{\mathcal{V},i} \beta_{\mathcal{V}^\prime,j} \left\langle  \kappa_{n,\mathcal{V}} (X_{\mathcal{V},i}) (\cdot), \kappa_{n,\mathcal{V}^\prime} (X_{\mathcal{V}^\prime,j}) (\cdot) \right\rangle_\mathcal{H} \\
    &\quad + \alpha_{\mathcal{V},i} \beta_{\mathcal{V}^\prime,j} \left\langle \mathcal{K} ( \cdot, X_i ), \kappa_{n,\mathcal{V}^\prime} (X_{\mathcal{V}^\prime,j}) (\cdot) \right\rangle_\mathcal{H} + \alpha_{\mathcal{V}^\prime,j} \beta_{\mathcal{V},i} \left\langle \mathcal{K} ( \cdot, X_j ), \kappa_{n,\mathcal{V}} (X_{\mathcal{V},i}) (\cdot) \right\rangle_\mathcal{H}\Big].
\end{align}
\noindent
A possible choice for the estimator $\kappa_{n,\mathcal{V}} (X_{\mathcal{V},i}) (\cdot)$ can be obtained from kernel ridge regression \citep{Park2020,Park2021}, so that
\begin{align} \label{eqn10}
    \kappa_{n,\mathcal{V}} (X_{\mathcal{V},i}) (\cdot) = \mathbf{k}_\mathcal{V}^\intercal ( X_{\mathcal{V},i} ) \mathbf{W}_\mathcal{V} \mathbf{k}_X (\cdot)
\end{align}
where $\mathbf{W}_\mathcal{V} = ( \mathbf{K}_\mathcal{V} + \lambda I_n )^{-1}$ with $[ \mathbf{K}_\mathcal{V} ]_{ij} = \mathcal{K}_\mathcal{V} ( X_{\mathcal{V},i}, X_{\mathcal{V},j} )$, $\mathbf{k}_\mathcal{V} ( x_\mathcal{V} ) = ( \mathcal{K}_\mathcal{V} ( X_{\mathcal{V},1}, x_\mathcal{V} ), \dots, \mathcal{K}_\mathcal{V} ( X_{\mathcal{V},n}, x_\mathcal{V} ) )^\intercal$ and $\mathbf{k}_X (x) = ( \mathcal{K} ( X_1, x ), \dots, \mathcal{K} ( X_n, x ) )^\intercal$. 
\noindent
Substituting the expression (\ref{eqn10}) for the estimator, one may compute the inner products from (\ref{eqn9}):
\begin{align*}
    \left\langle  \kappa_{n,\mathcal{V}} (X_{\mathcal{V},i}) (\cdot), \kappa_{n,\mathcal{V}^\prime} (X_{\mathcal{V}^\prime,j}) (\cdot) \right\rangle_\mathcal{H} &= \left\langle \mathbf{k}_\mathcal{V}^\intercal ( X_{\mathcal{V},i} ) \mathbf{W}_\mathcal{V} \mathbf{k}_X (\cdot), \mathbf{k}_{\mathcal{V}^\prime}^\intercal ( X_{\mathcal{V}^\prime,j} ) \mathbf{W}_{\mathcal{V}^\prime} \mathbf{k}_X (\cdot) \right\rangle_\mathcal{H} \\
    &= \mathbf{k}_\mathcal{V}^\intercal ( X_{\mathcal{V},i} ) \mathbf{W}_\mathcal{V} \mathbf{K} \mathbf{W}_{\mathcal{V}^\prime}^\intercal \mathbf{k}_{\mathcal{V}^\prime} ( X_{\mathcal{V}^\prime,j} ),
\end{align*}
where $\mathbf{K}_{ij} = \mathcal{K} ( X_i, X_j )$. Furthermore,
\begin{align*}
    \left\langle \kappa_{n,\mathcal{V}} (X_{\mathcal{V},j}) (\cdot), \mathcal{K} ( \cdot, X_i ) \right\rangle_\mathcal{H} &= \left\langle \mathbf{k}_\mathcal{V}^\intercal ( X_{\mathcal{V},j} ) \mathbf{W}_\mathcal{V} \mathbf{k}_X (\cdot), \mathcal{K} ( \cdot, X_i ) \right\rangle_\mathcal{H} = \mathbf{k}_\mathcal{V}^\intercal ( X_{\mathcal{V},j} ) \mathbf{W}_\mathcal{V} \mathbf{k}_X ( X_i ).
\end{align*}
Combining the above computations, we obtain the following matrix expression:
\begin{align*}
    \left\langle \hat{\nu}^\mathcal{K}_{\mathcal{V}}, \hat{\nu}^\mathcal{K}_{\mathcal{V}^\prime} \right\rangle_\mathcal{H} &= \frac{1}{n^2} \sum_{i=1}^n \sum_{j=1}^n \alpha_{\mathcal{V},i} \alpha_{\mathcal{V}^\prime,j} \mathcal{K} ( X_i, X_j ) + \beta_{\mathcal{V},i} \beta_{\mathcal{V}^\prime,j} \mathbf{k}_\mathcal{V}^\intercal ( X_{\mathcal{V},i} ) \mathbf{W}_\mathcal{V} \mathbf{K} \mathbf{W}_{\mathcal{V}^\prime}^\intercal \mathbf{k}_{\mathcal{V}^\prime} ( X_{\mathcal{V}^\prime,j} ) \\
    &\quad + \alpha_{\mathcal{V},i} \beta_{\mathcal{V}^\prime,j} \mathbf{k}_{\mathcal{V}^\prime}^\intercal ( X_{\mathcal{V}^\prime,j} ) \mathbf{W}_{\mathcal{V}^\prime} \mathbf{k}_X ( X_i ) + \alpha_{\mathcal{V}^\prime,j} \beta_{\mathcal{V},i} \mathbf{k}_{\mathcal{V}}^\intercal ( X_{\mathcal{V},i} ) \mathbf{W}_{\mathcal{V}} \mathbf{k}_X ( X_j ) \\
    &=\frac{1}{n^2} \lbrace \mathbf{\alpha}_{\mathcal{V}}^\intercal \mathbf{K} \mathbf{\alpha}_{\mathcal{V}^\prime} + \mathbf{\beta}_{\mathcal{V}}^\intercal \mathbf{K}_{\mathcal{V}}^\intercal \mathbf{W}_{\mathcal{V}} \mathbf{K} \mathbf{W}_{\mathcal{V}^\prime} \mathbf{K}_{\mathcal{V}^\prime} \mathbf{\beta}_{\mathcal{V}^\prime} + \mathbf{\beta}_{\mathcal{V}^\prime}^\intercal \mathbf{K}_{\mathcal{V}^\prime}^\intercal \mathbf{W}_{\mathcal{V}^\prime} \mathbf{K} \mathbf{\alpha}_{\mathcal{V}} + \mathbf{\beta}_{\mathcal{V}}^\intercal \mathbf{K}_{\mathcal{V}}^\intercal \mathbf{W}_{\mathcal{V}} \mathbf{K} \mathbf{\alpha}_{\mathcal{V}^\prime} \rbrace \\
    &= \frac{1}{n^2} ( \mathbf{\alpha}_{\mathcal{V}} +  \mathbf{W}_{\mathcal{V}} \mathbf{K}_{\mathcal{V}} \mathbf{\beta}_{\mathcal{V}} )^\intercal \mathbf{K}  ( \mathbf{\alpha}_{\mathcal{V}^\prime} +  \mathbf{W}_{\mathcal{V}^\prime} \mathbf{K}_{\mathcal{V}^\prime} \mathbf{\beta}_{\mathcal{V}^\prime} ).
\end{align*}
\noindent
Finally, we can express the test statistic as
\begin{align} \label{eqn8}
    \nonumber
    &\| \hat{\gamma}^\mathcal{K}_\omega \|_\mathcal{H}^2 \\
    \nonumber
    &= \frac{1}{n^2} \sum_{\mathcal{V} \subseteq [d]} \sum_{{\mathcal{V}^\prime} \subseteq [d]} \omega_\mathcal{V}  \omega_{\mathcal{V}^\prime} \lbrace \mathbf{\alpha}_{\mathcal{V}}^\intercal \mathbf{K} \mathbf{\alpha}_{\mathcal{V}^\prime} + \mathbf{\beta}_{\mathcal{V}}^\intercal \mathbf{K}_{\mathcal{V}}^\intercal \mathbf{W}_{\mathcal{V}} \mathbf{K} \mathbf{W}_{\mathcal{V}^\prime} \mathbf{K}_{\mathcal{V}^\prime} \mathbf{\beta}_{\mathcal{V}^\prime} + \mathbf{\beta}_{\mathcal{V}^\prime}^\intercal \mathbf{K}_{\mathcal{V}^\prime}^\intercal \mathbf{W}_{\mathcal{V}^\prime} \mathbf{K} \mathbf{\alpha}_{\mathcal{V}} + \mathbf{\beta}_{\mathcal{V}}^\intercal \mathbf{K}_{\mathcal{V}}^\intercal \mathbf{W}_{\mathcal{V}} \mathbf{K} \mathbf{\alpha}_{\mathcal{V}^\prime} \rbrace \\
    &= \frac{1}{n^2} \left\lbrace \sum_{\mathcal{V} \subseteq [d]} \omega_\mathcal{V} ( \mathbf{\alpha}_{\mathcal{V}} +  \mathbf{W}_{\mathcal{V}} \mathbf{K}_{\mathcal{V}} \mathbf{\beta}_{\mathcal{V}} ) \right\rbrace^\intercal \mathbf{K} \left\lbrace \sum_{\mathcal{V}^\prime \subseteq [d]} \omega_{\mathcal{V}^\prime} ( \mathbf{\alpha}_{\mathcal{V}^\prime} +  \mathbf{W}_{\mathcal{V}^\prime} \mathbf{K}_{\mathcal{V}^\prime} \mathbf{\beta}_{\mathcal{V}^\prime} ) \right\rbrace.
\end{align}

\subsection{Test statistic for the test of no importance for the LOCO variable importance measure}

We specialize the computation of the general weighted estimand to the specific case of the LOCO variable importance measure. For $\mathcal{P} ( [d] )$ the powerset of $[d]$, let $\omega \in \mathbb{R}^{\lvert \mathcal{P} ( [d] ) \rvert}$, such that $\omega_{\mathcal{V}} = 1$ and $\omega_{\mathcal{V}^\prime} = -1$, for $\mathcal{V}^\prime \subsetneq \mathcal{V}$, then $\gamma_\omega^\mathcal{K} = \gamma_{LOCO}^\mathcal{K}$. The parameter of interest is
\begin{align*}
    \gamma_{LOCO}^\mathcal{K} (P; \mathcal{V}, \mathcal{V}^\prime) (x) = \nu_P^\mathcal{K} (\mathcal{V}) (x) - \nu_P^\mathcal{K} (\mathcal{V}^\prime ) ( x ).
\end{align*}
The corresponding one-step estimator is given by $\hat{\gamma}_{LOCO}^\mathcal{K} (\cdot) = \hat{\nu}^\mathcal{K}_{\mathcal{V}} (\cdot) - \hat{\nu}^\mathcal{K}_{\mathcal{V}^\prime} (\cdot)$. Furthermore, the test statistic for the test of no importance for the LOCO variable importance measure is of the form
\begin{align*}
    \| \hat{\gamma}_{LOCO}^\mathcal{K} \|_\mathcal{H}^2 &= \left\langle \hat{\nu}^\mathcal{K}_{\mathcal{V}}, \hat{\nu}^\mathcal{K}_{\mathcal{V}} \right\rangle_\mathcal{H} + \left\langle \hat{\nu}^\mathcal{K}_{\mathcal{V}^\prime}, \hat{\nu}^\mathcal{K}_{\mathcal{V}^\prime} \right\rangle_\mathcal{H} - 2 \left\langle \hat{\nu}^\mathcal{K}_{\mathcal{V}}, \hat{\nu}^\mathcal{K}_{\mathcal{V}^\prime} \right\rangle_\mathcal{H}.
\end{align*}
Therefore, analogously to the expression (\ref{eqn8}), we can express the test statistic as
\begin{align*}
&\| \hat{\gamma}_{LOCO}^\mathcal{K} \|_\mathcal{H}^2 \\
&= \frac{1}{n^2} \lbrace ( \mathbf{\alpha}_{\mathcal{V}} +  \mathbf{W}_{\mathcal{V}} \mathbf{K}_{\mathcal{V}} \mathbf{\beta}_{\mathcal{V}} ) - ( \mathbf{\alpha}_{\mathcal{V}^\prime} +  \mathbf{W}_{\mathcal{V}^\prime} \mathbf{K}_{\mathcal{V}^\prime} \mathbf{\beta}_{\mathcal{V}^\prime} ) \rbrace^\intercal \mathbf{K} \lbrace ( \mathbf{\alpha}_{\mathcal{V}} +  \mathbf{W}_{\mathcal{V}} \mathbf{K}_{\mathcal{V}} \mathbf{\beta}_{\mathcal{V}} ) - ( \mathbf{\alpha}_{\mathcal{V}^\prime} +  \mathbf{W}_{\mathcal{V}^\prime} \mathbf{K}_{\mathcal{V}^\prime} \mathbf{\beta}_{\mathcal{V}^\prime} ) \rbrace.
\end{align*}

\subsection{Test statistic for the statistical test of no importance for the Shapley values}

We specialize the computation of the general weighted estimand to the particular case of Shapley values. Let $\omega \in \mathbb{R}^{\lvert \mathcal{P} ( [d] ) \rvert}$, such that $\omega_{\mathcal{V}} = -\frac{1}{d} \binom{d-1}{\lvert \mathcal{V} \rvert}^{-1}$ and $\omega_{\mathcal{V} \cup \lbrace i \rbrace} = \frac{1}{d} \binom{d-1}{\lvert \mathcal{V} \rvert}^{-1}$, for all $\mathcal{V} \subseteq [d] \backslash \lbrace i \rbrace$, then $\gamma_\omega^\mathcal{K} = \gamma_{SHAP,i}^\mathcal{K}$, which is given by
\begin{align}
    \gamma_{SHAP,i}^\mathcal{K} (P) (x) = \frac{1}{d} \sum_{\mathcal{V} \subseteq [d] \backslash \lbrace i \rbrace} \binom{d-1}{\lvert \mathcal{V} \rvert}^{-1} \lbrace \nu_P^\mathcal{K} ( {\mathcal{V} \cup \lbrace i \rbrace} ) (x_{\mathcal{V} \cup \lbrace i \rbrace}) - \nu_P^\mathcal{K} (\mathcal{V}) (x_\mathcal{V}) \rbrace. \label{eq:ShapleyNonPermRep}
\end{align}
The corresponding one-step estimator is
\begin{align*}
    \hat{\gamma}_{SHAP,i}^\mathcal{K} (\cdot) &= \frac{1}{d} \sum_{\mathcal{V} \subseteq [d] \backslash \lbrace i \rbrace} \binom{d-1}{\lvert \mathcal{V} \rvert}^{-1} \lbrace \hat{\nu}^\mathcal{K}_{\mathcal{V} \cup \lbrace i \rbrace} (\cdot) - \hat{\nu}^\mathcal{K}_{\mathcal{V}} (\cdot) \rbrace.
\end{align*}
Furthermore, the test statistic for the test of no importance for the Shapley values is of the form
\begin{align*}
    \| \hat{\gamma}_{SHAP,i}^\mathcal{K} \|_\mathcal{H}^2 &= \left\langle \frac{1}{d} \sum_{\mathcal{V} \subseteq [d] \backslash \lbrace i \rbrace} \binom{d-1}{\lvert \mathcal{V} \rvert}^{-1} ( \hat{\nu}^\mathcal{K}_{\mathcal{V} \cup \lbrace i \rbrace} - \hat{\nu}^\mathcal{K}_{\mathcal{V}} ), \frac{1}{d} \sum_{\mathcal{V}^\prime \subseteq [d] \backslash \lbrace i \rbrace} \binom{d-1}{\lvert \mathcal{V}^\prime \rvert}^{-1} ( \hat{\nu}^\mathcal{K}_{\mathcal{V}^\prime \cup \lbrace i \rbrace} - \hat{\nu}^\mathcal{K}_{\mathcal{V}^\prime} ) \right\rangle_\mathcal{H} \\
    &=\frac{1}{d^2} \sum_{\mathcal{V} \subseteq [d] \backslash \lbrace i \rbrace} \sum_{\mathcal{V}^\prime \subseteq [d] \backslash \lbrace i \rbrace} \binom{d-1}{\lvert \mathcal{V} \rvert}^{-1} \binom{d-1}{\lvert \mathcal{V}^\prime \rvert}^{-1} \\
    &\quad \times \left\lbrace \left\langle \hat{\nu}^\mathcal{K}_{\mathcal{V} \cup \lbrace i \rbrace}, \hat{\nu}^\mathcal{K}_{\mathcal{V}^\prime \cup \lbrace i \rbrace} \right\rangle_\mathcal{H} + \left\langle \hat{\nu}^\mathcal{K}_{\mathcal{V}}, \hat{\nu}^\mathcal{K}_{\mathcal{V}^\prime} \right\rangle_\mathcal{H} - \left\langle \hat{\nu}^\mathcal{K}_{\mathcal{V} \cup \lbrace i \rbrace}, \hat{\nu}^\mathcal{K}_{\mathcal{V}^\prime} \right\rangle_\mathcal{H} - \left\langle \hat{\nu}^\mathcal{K}_{\mathcal{V}^\prime \cup \lbrace i \rbrace}, \hat{\nu}^\mathcal{K}_{\mathcal{V}} \right\rangle_\mathcal{H} \right\rbrace.
\end{align*}
Therefore, analogously to (\ref{eqn8}), we can express the test statistic as
\begin{align*}
    &\| \hat{\gamma}_{SHAP,i}^\mathcal{K} \|_\mathcal{H}^2 \\
    &= \frac{1}{n^2 d^2} \left[ \sum_{\mathcal{V} \subseteq [d] \backslash \lbrace i \rbrace} \binom{d-1}{\lvert \mathcal{V} \rvert}^{-1} \left\lbrace \mathbf{\alpha}_{\mathcal{V} \cup \lbrace i \rbrace} +  \mathbf{W}_{\mathcal{V} \cup \lbrace i \rbrace} \mathbf{K}_{\mathcal{V} \cup \lbrace i \rbrace} \mathbf{\beta}_{\mathcal{V} \cup \lbrace i \rbrace} - ( \mathbf{\alpha}_{\mathcal{V}} +  \mathbf{W}_{\mathcal{V}} \mathbf{K}_{\mathcal{V}} \mathbf{\beta}_{\mathcal{V}} ) \right\rbrace \right]^\intercal \\
    &\quad \times \mathbf{K} \left[ \sum_{\mathcal{V}^\prime \subseteq [d] \backslash \lbrace i \rbrace} \binom{d-1}{\lvert \mathcal{V}^\prime \rvert}^{-1} \lbrace \mathbf{\alpha}_{\mathcal{V}^\prime \cup \lbrace i \rbrace} +  \mathbf{W}_{\mathcal{V}^\prime \cup \lbrace i \rbrace} \mathbf{K}_{\mathcal{V}^\prime \cup \lbrace i \rbrace} \mathbf{\beta}_{\mathcal{V}^\prime \cup \lbrace i \rbrace} - ( \mathbf{\alpha}_{\mathcal{V}^\prime} +  \mathbf{W}_{\mathcal{V}^\prime} \mathbf{K}_{\mathcal{V}^\prime} \mathbf{\beta}_{\mathcal{V}^\prime} ) \rbrace \right].
\end{align*}

\subsection{Test statistic for the statistical test of no importance for the Shapley values: permutation formulation}\label{app:permutation}

The Shapley value of a variable $X_i$ can be viewed as its marginal contribution to the subset of variables $X_\mathcal{V}$, where $\mathcal{V} \subseteq [d] \backslash \lbrace i \rbrace$, averaged over the set of all permutations of the features. Specifically, let $\mathcal{S}_d$ denote the symmetric group on $d$ elements and for $\sigma \in \mathcal{S}_d$, let $\sigma(i)=j$ indicate that feature $i$ is assigned rank $j$ by the permutation $\sigma$. Alternatively to \eqref{eq:ShapleyNonPermRep}, the Shapley value of $i$-th feature for CATE can be represented as
\begin{align*}
    \gamma_{SHAP,i}^\mathcal{K} ( P ) ( x ) &= \frac{1}{d!} \sum_{\sigma \in \mathcal{S}_d} \lbrace \nu_P^\mathcal{K} ( [ \sigma ]_{i} \cup \lbrace i \rbrace ) ( x_{[ \sigma ]_{i} \cup \lbrace i \rbrace} ) - \nu_P^\mathcal{K} ( [ \sigma ]_{i} ) ( x_{[ \sigma ]_{i}} ) \rbrace, 
\end{align*}
where $[ \sigma ]_{i}$ denotes the subset of features ranked lower than feature $X_i$ in the permutation $\sigma$ \citep{Mitchell2022, Verdinelli2023}. Therefore, one may express the one-step estimator for Shapley values using the equivalent permutation formulation:
\begin{align*}
    \hat{\gamma}_{SHAP,i}^\mathcal{K} (\cdot) = \frac{1}{d!} \sum_{\sigma \in \mathcal{S}_d} \lbrace \hat{\nu}^\mathcal{K}_{[ \sigma ]_{i} \cup \lbrace i \rbrace} (\cdot) - \hat{\nu}^\mathcal{K}_{[ \sigma ]_{i}} (\cdot) \rbrace.
\end{align*} 
Furthermore, the test statistic for the corresponding test of no importance may be expressed as:
\begin{align*}
    \| \hat{\gamma}_{SHAP,i}^\mathcal{K} \|_\mathcal{H}^2 &= \left\langle \frac{1}{d!} \sum_{\sigma \in \mathcal{S}_d} ( \hat{\nu}^\mathcal{K}_{[ \sigma ]_{i} \cup \lbrace i \rbrace} - \hat{\nu}^\mathcal{K}_{[ \sigma ]_{i}} ), \frac{1}{d!} \sum_{\sigma^\prime \in \mathcal{S}_d} ( \hat{\nu}^\mathcal{K}_{[ \sigma^\prime ]_{i} \cup \lbrace i \rbrace} - \hat{\nu}^\mathcal{K}_{[ \sigma^\prime ]_{i}} ) \right\rangle_\mathcal{H} \\
    &= \frac{1}{( d! )^2} \sum_{\sigma \in \mathcal{S}_d} \sum_{\sigma^\prime \in \mathcal{S}_d} \left\lbrace \left\langle \hat{\nu}^\mathcal{K}_{[ \sigma ]_{i} \cup \lbrace i \rbrace}, \hat{\nu}^\mathcal{K}_{[ \sigma^\prime ]_{i} \cup \lbrace i \rbrace} \right\rangle_\mathcal{H} + \left\langle \hat{\nu}^\mathcal{K}_{[ \sigma ]_{i}}, \hat{\nu}^\mathcal{K}_{[ \sigma^\prime ]_{i}} \right\rangle_\mathcal{H} \right.\\
    &\quad \left.- \left\langle \hat{\nu}^\mathcal{K}_{[ \sigma ]_{i} \cup \lbrace i \rbrace}, \hat{\nu}^\mathcal{K}_{[ \sigma^\prime ]_{i}} \right\rangle_\mathcal{H} - \left\langle \hat{\nu}^\mathcal{K}_{[ \sigma^\prime ]_{i} \cup \lbrace i \rbrace}, \hat{\nu}^\mathcal{K}_{[ \sigma ]_{i}} \right\rangle_\mathcal{H} \right\rbrace.
\end{align*}
Therefore, we can express the test statistic as the following matrix expression:
\begin{align} \label{eqn15}
    \nonumber
    \| \hat{\gamma}_{SHAP,i}^\mathcal{K} \|_\mathcal{H}^2
    &= \frac{1}{(n d! )^2} \left[ \sum_{\sigma \in \mathcal{S}_d} \lbrace \mathbf{\alpha}_{[ \sigma ]_{i} \cup \lbrace i \rbrace} +  \mathbf{W}_{[ \sigma ]_{i} \cup \lbrace i \rbrace} \mathbf{K}_{[ \sigma ]_{i} \cup \lbrace i \rbrace} \mathbf{\beta}_{[ \sigma ]_{i} \cup \lbrace i \rbrace} - ( \mathbf{\alpha}_{[ \sigma ]_{i}} +  \mathbf{W}_{[ \sigma ]_{i}} \mathbf{K}_{[ \sigma ]_{i}} \mathbf{\beta}_{[ \sigma ]_{i}} ) \rbrace \right]^\intercal \\
    &\quad \times \mathbf{K} \left[ \sum_{\sigma^\prime \in \mathcal{S}_d} \lbrace \mathbf{\alpha}_{[ \sigma^\prime ]_{i} \cup \lbrace i \rbrace} +  \mathbf{W}_{[ \sigma^\prime ]_{i} \cup \lbrace i \rbrace} \mathbf{K}_{[ \sigma^\prime ]_{i} \cup \lbrace i \rbrace} \mathbf{\beta}_{[ \sigma^\prime ]_{i} \cup \lbrace i \rbrace} - ( \mathbf{\alpha}_{[ \sigma^\prime ]_{i}} +  \mathbf{W}_{[ \sigma^\prime ]_{i}} \mathbf{K}_{[ \sigma^\prime ]_{i}} \mathbf{\beta}_{ [ \sigma^\prime ]_{i}} ) \rbrace \right].
\end{align}

\section{Computation of an asymptotically valid confidence set for the weighted estimand using the bootstrap} \label{SectionA4}

\subsection{Computation of an asymptotically valid confidence set for a generic weighted parameter using the bootstrap}

Theorem \ref{Theorem3} states that $\mathcal{C}_n^\omega ( \hat{\xi}_{\omega} )$ is an asymptotically valid confidence set, if we construct a consistent estimator of the $( 1 - \alpha )$-quantile of $\| \mathbb{H}_\omega \|_\mathcal{H}^2$. A consistent estimator of $\xi_{\omega, 1-\alpha}$ can be obtained via the bootstrap. 
Let $Z_1^\sharp, Z_2^\sharp, \dots, Z_n^\sharp$ be an iid sample from the empirical measure $P_n$ and let $P_n^\sharp$ be the empirical distribution of $Z_1^\sharp, Z_2^\sharp, \dots, Z_n^\sharp$. Furthermore, define $\mathbb{H}_{\omega,n}^\sharp \coloneqq n^{1/2} ( P_n^{\sharp} - P_n ) \phi_n^\omega$. By Theorem 4 in \cite{Luedtke2024}, the bootstrap estimator $\hat{\xi}_{\omega}$ obtained as $( 1 - \alpha )$-quantile of $\| \mathbb{H}_{\omega,n}^\sharp \|_\mathcal{H}^2$ is a consistent estimator of the threshold $\xi_{\omega, 1-\alpha}$. To approximate this quantile empirically, the following closed-form expression is useful:
\begin{align*}
    \| \mathbb{H}_{\omega,n}^\sharp \|_\mathcal{H}^2 &= \langle n^{1/2} ( P_n^{\sharp} - P_n ) \phi_n^\omega, n^{1/2} ( P_n^{\sharp} - P_n ) \phi_n^\omega \rangle_\mathcal{H} \\
    &= n \sum_{\mathcal{V} \subseteq [d]} \sum_{\mathcal{V}^\prime \subseteq [d]} \omega_\mathcal{V} \omega_{\mathcal{V}^\prime} \langle ( P_n^{\sharp} - P_n ) \phi_{n}^\mathcal{V}, ( P_n^{\sharp} - P_n ) \phi_{n}^{\mathcal{V}^\prime} \rangle_\mathcal{H}.
\end{align*}
\noindent
The inner product in the above expression rewrites as
\begin{align*}
    &\left\langle ( P_n^{\sharp} - P_n ) \phi_{n}^{\mathcal{V}}, ( P_n^{\sharp} - P_n ) \phi_{n}^{\mathcal{V}^\prime}\right\rangle_\mathcal{H} \\
    &=\left\langle ( P_n^{\sharp} - P_n ) \left\lbrace \beta_{\mathcal{V}} \kappa_{n,\mathcal{V}} (X_{\mathcal{V}}) (\cdot) + \alpha_{\mathcal{V}} \mathcal{K} ( \cdot, X ) - \frac{1}{n} \sum_{l=1}^n \alpha_{\mathcal{V},l} \mathcal{K} ( \cdot, X_l ) \right\rbrace, \right. \\
    &\quad \left. ( P_n^{\sharp} - P_n ) \left\lbrace \beta_{\mathcal{V}^\prime} \kappa_{n,\mathcal{V}^\prime} (X_{\mathcal{V}^\prime}) (\cdot) + \alpha_{\mathcal{V}^\prime} \mathcal{K} ( \cdot, X ) - \frac{1}{n} \sum_{k=1}^n \alpha_{\mathcal{V}^\prime,k} \mathcal{K} ( \cdot, X_k ) \right\rbrace \right\rangle_\mathcal{H} \\
    &= \langle ( P_n^{\sharp} - P_n ) \lbrace \beta_{\mathcal{V}} \kappa_{n,\mathcal{V}} (X_{\mathcal{V}}) (\cdot) + \alpha_{\mathcal{V}} \mathcal{K} ( \cdot, X ) \rbrace, ( P_n^{\sharp} - P_n ) \lbrace \beta_{\mathcal{V}^\prime} \kappa_{n,\mathcal{V}^\prime} (X_{\mathcal{V}^\prime}) (\cdot) + \alpha_{\mathcal{V}^\prime} \mathcal{K} ( \cdot, X ) \rbrace \rangle_\mathcal{H} \\
    &=\langle ( P_n^{\sharp} - P_n ) \beta_{\mathcal{V}} \kappa_{n,\mathcal{V}} (X_{\mathcal{V}}) (\cdot), ( P_n^{\sharp} - P_n ) \beta_{\mathcal{V}^\prime} \kappa_{n,\mathcal{V}^\prime} (X_{\mathcal{V}^\prime}) (\cdot) \rangle_\mathcal{H} \\
    &\quad + \langle ( P_n^{\sharp} - P_n ) \alpha_{\mathcal{V}} \mathcal{K} ( \cdot, X ), ( P_n^{\sharp} - P_n ) \beta_{\mathcal{V}^\prime} \kappa_{n, \mathcal{V}^\prime} (X_{\mathcal{V}^\prime}) (\cdot) \rangle_\mathcal{H} \\
    &\quad + \langle ( P_n^{\sharp} - P_n ) \beta_{\mathcal{V}} \kappa_{n,\mathcal{V}} (X_{\mathcal{V}}) (\cdot), ( P_n^{\sharp} - P_n ) \alpha_{\mathcal{V}^\prime} \mathcal{K} ( \cdot, X ) \rangle_\mathcal{H} \\
    &\quad + \langle ( P_n^{\sharp} - P_n ) \alpha_{\mathcal{V}} \mathcal{K} ( \cdot, X ), ( P_n^{\sharp} - P_n ) \alpha_{\mathcal{V}^\prime} \mathcal{K} ( \cdot, X ) \rangle_\mathcal{H} \\
    &= \mathbf{\tilde{\beta}}_{\mathcal{V}}^\intercal \mathbf{K}_{\mathcal{V}}^\intercal \mathbf{W}_{\mathcal{V}} \mathbf{K} \mathbf{W}_{\mathcal{V}^\prime} \mathbf{K}_{\mathcal{V}^\prime} \mathbf{\tilde{\beta}}_{\mathcal{V}^\prime} + \mathbf{\tilde{\beta}}_{\mathcal{V}^\prime}^\intercal \mathbf{K}_{\mathcal{V}^\prime}^\intercal \mathbf{W}_{\mathcal{V}^\prime} \mathbf{K} \mathbf{\tilde{\alpha}}_{\mathcal{V}} + \mathbf{\tilde{\beta}}_{\mathcal{V}}^\intercal \mathbf{K}_{\mathcal{V}}^\intercal \mathbf{W}_{\mathcal{V}} \mathbf{K} \mathbf{\tilde{\alpha}}_{\mathcal{V}^\prime} + \mathbf{\tilde{\alpha}}_{\mathcal{V}}^\intercal \mathbf{K} \mathbf{\tilde{\alpha}}_{\mathcal{V}^\prime}, 
\end{align*}
where $\mathbf{\tilde{\alpha}}_{\mathcal{V}} \coloneqq ( P_n^{\sharp} - P_n ) ( \alpha_{\mathcal{V},1}, \dots, \alpha_{\mathcal{V},n} )^\intercal$ and $\mathbf{\tilde{\beta}}_{\mathcal{V}} \coloneqq ( P_n^{\sharp} - P_n ) ( \beta_{\mathcal{V},1}, \dots, \beta_{\mathcal{V},n} )^\intercal$. Therefore, we obtain the following expression for $\| \mathbb{H}_{\omega,n}^\sharp \|_\mathcal{H}^2$:
\begin{align*}
    \| \mathbb{H}_{\omega,n}^\sharp \|_\mathcal{H}^2
    = n \left\lbrace \sum_{\mathcal{V} \subseteq [d]} \omega_\mathcal{V} ( \mathbf{\tilde{\alpha}}_{\mathcal{V}} +  \mathbf{W}_{\mathcal{V}} \mathbf{K}_{\mathcal{V}} \mathbf{\tilde{\beta}}_{\mathcal{V}} ) \right\rbrace^\intercal \mathbf{K} \left\lbrace \sum_{\mathcal{V}^\prime \subseteq [d]} \omega_{\mathcal{V}^\prime} ( \mathbf{\tilde{\alpha}}_{\mathcal{V}^\prime} +  \mathbf{W}_{\mathcal{V}^\prime} \mathbf{K}_{\mathcal{V}^\prime} \mathbf{\tilde{\beta}}_{\mathcal{V}^\prime} ) \right\rbrace.
\end{align*}

\subsection{Computation of an asymptotically valid confidence set for LOCO variable importance measure using the bootstrap}

For the LOCO variable importance measure, $\mathbb{H}_{\omega,n}^\sharp$ is of the form:
\begin{align*}
    \mathbb{H}_{LOCO,n}^\sharp \coloneqq n^{1/2} ( P_n^{\sharp} - P_n ) ( \phi_n^{\mathcal{V}} - \phi_n^{\mathcal{V}^\prime} ),
\end{align*}
where $\phi_n^\mathcal{V}$ is $\phi_n^\omega$ with $\omega_{\mathcal{V}} = 1$ and $\omega_{\mathcal{V}^\prime} = 0$, for $\mathcal{V}^\prime \neq \mathcal{V}$. Therefore,
\begin{align*}
    &\| \mathbb{H}_{LOCO,n}^\sharp \|_\mathcal{H}^2 \\
    &= \langle n^{1/2} ( P_n^{\sharp} - P_n ) ( \phi_n^{\mathcal{V}} - \phi_n^{\mathcal{V}^\prime} ), n^{1/2} ( P_n^{\sharp} - P_n ) ( \phi_n^{\mathcal{V}} - \phi_n^{\mathcal{V}^\prime} ) \rangle_\mathcal{H} \\
    &= n \lbrace ( \mathbf{\tilde{\alpha}}_{\mathcal{V}} +  \mathbf{W}_{\mathcal{V}} \mathbf{K}_{\mathcal{V}} \mathbf{\tilde{\beta}}_{\mathcal{V}} ) - ( \mathbf{\tilde{\alpha}}_{\mathcal{V}^\prime} +  \mathbf{W}_{\mathcal{V}^\prime} \mathbf{K}_{\mathcal{V}^\prime} \mathbf{\tilde{\beta}}_{\mathcal{V}^\prime} ) \rbrace^\intercal \mathbf{K} \lbrace ( \mathbf{\tilde{\alpha}}_{\mathcal{V}} +  \mathbf{W}_{\mathcal{V}} \mathbf{K}_{\mathcal{V}} \mathbf{\tilde{\beta}}_{\mathcal{V}} ) - ( \mathbf{\tilde{\alpha}}_{\mathcal{V}^\prime} +  \mathbf{W}_{\mathcal{V}^\prime} \mathbf{K}_{\mathcal{V}^\prime} \mathbf{\tilde{\beta}}_{\mathcal{V}^\prime} ) \rbrace.
\end{align*}

\subsection{Computation of an asymptotically valid confidence set for Shapley values using the bootstrap}

For the Shapley value, $\mathbb{H}_{\omega,n}^\sharp$ is of the form:
\begin{align*}
    &\mathbb{H}_{SHAP,n}^\sharp \coloneqq  \frac{n^{1/2}}{d} \sum_{\mathcal{V} \subseteq [d] \backslash \lbrace i \rbrace} \binom{d-1}{\lvert \mathcal{V} \rvert}^{-1} ( P_n^{\sharp} - P_n ) ( \phi_n^{\mathcal{V} \cup \lbrace i \rbrace} - \phi_n^{\mathcal{V}} ).
\end{align*}
Therefore,
\begin{align*}
    \| \mathbb{H}_{SHAP,n}^\sharp \|_\mathcal{H}^2 &= \frac{n}{d^2} \sum_{\mathcal{V} \subseteq [d] \backslash \lbrace i \rbrace} \sum_{\mathcal{V}^\prime \subseteq [d] \backslash \lbrace i \rbrace} \binom{d-1}{\lvert \mathcal{V} \rvert}^{-1} \binom{d-1}{\lvert \mathcal{V}^\prime \rvert}^{-1} \\
    &\quad \times \langle ( P_n^{\sharp} - P_n ) ( \phi_n^{\mathcal{V} \cup \lbrace i \rbrace} - \phi_n^{\mathcal{V}} ), ( P_n^{\sharp} - P_n ) ( \phi_n^{\mathcal{V}^\prime \cup \lbrace i \rbrace} - \phi_n^{\mathcal{V}^\prime} ) \rangle_\mathcal{H} \\
    &= \frac{n}{d^2} \left[ \sum_{\mathcal{V} \subseteq [d] \backslash \lbrace i \rbrace} \binom{d-1}{\lvert \mathcal{V} \rvert}^{-1} \lbrace \mathbf{\tilde{\alpha}}_{\mathcal{V} \cup \lbrace i \rbrace} +  \mathbf{W}_{\mathcal{V} \cup \lbrace i \rbrace} \mathbf{K}_{\mathcal{V} \cup \lbrace i \rbrace} \mathbf{\tilde{\beta}}_{\mathcal{V} \cup \lbrace i \rbrace} - ( \mathbf{\tilde{\alpha}}_{\mathcal{V}} +  \mathbf{W}_{\mathcal{V}} \mathbf{K}_{\mathcal{V}} \mathbf{\tilde{\beta}}_{\mathcal{V}} ) \rbrace \right]^\intercal \\
    &\quad \times \mathbf{K} \left[ \sum_{\mathcal{V}^\prime \subseteq [d] \backslash \lbrace i \rbrace} \binom{d-1}{\lvert \mathcal{V}^\prime \rvert}^{-1} \lbrace \mathbf{\tilde{\alpha}}_{\mathcal{V}^\prime \cup \lbrace i \rbrace} +  \mathbf{W}_{\mathcal{V}^\prime \cup \lbrace i \rbrace} \mathbf{K}_{\mathcal{V}^\prime \cup \lbrace i \rbrace} \mathbf{\tilde{\beta}}_{\mathcal{V}^\prime \cup \lbrace i \rbrace} - ( \mathbf{\tilde{\alpha}}_{\mathcal{V}^\prime} +  \mathbf{W}_{\mathcal{V}^\prime} \mathbf{K}_{\mathcal{V}^\prime} \mathbf{\tilde{\beta}}_{\mathcal{V}^\prime} ) \rbrace \right].
\end{align*}

\subsection{Computation of an asymptotically valid confidence set for Shapley values using the bootstrap - permutation formulation}

Alternatively, one may express the Shapley value $\mathbb{H}_{\omega,n}^\sharp$ using the permutation formulation as follows:
\begin{align*}
    &\mathbb{H}_{SHAP,n}^\sharp \coloneqq  \frac{n^{1/2}}{d!} \sum_{\sigma \in \mathcal{S}_d} ( P_n^{\sharp} - P_n ) ( \phi_n^{[ \sigma ]_{i} \cup \lbrace i \rbrace} - \phi_n^{[ \sigma ]_{i}} ).
\end{align*}
Therefore,
\begin{align} \label{eqn16}
    \nonumber
    \| \mathbb{H}_{SHAP,n}^\sharp \|_\mathcal{H}^2 &= \frac{n}{( d! )^2} \sum_{\sigma \in \mathcal{S}_d} \sum_{\sigma^\prime \in \mathcal{S}_d} \left\langle ( P_n^{\sharp} - P_n ) ( \phi_n^{[ \sigma ]_{i} \cup \lbrace i \rbrace} - \phi_n^{[ \sigma ]_{i}} ), ( P_n^{\sharp} - P_n ) ( \phi_n^{[ \sigma^\prime ]_{i} \cup \lbrace i \rbrace} - \phi_n^{[ \sigma^\prime ]_{i}} ) \right\rangle_\mathcal{H} \\
    \nonumber
    &= \frac{n}{( d! )^2} \sum_{\sigma \in \mathcal{S}_d} \sum_{\sigma^\prime \in \mathcal{S}_d} \left\lbrace \left\langle ( P_n^{\sharp} - P_n ) \phi_n^{[ \sigma ]_{i} \cup \lbrace i \rbrace}, ( P_n^{\sharp} - P_n ) \phi_n^{[ \sigma^\prime ]_{i} \cup \lbrace i \rbrace} \right\rangle_\mathcal{H} \right. \\
    \nonumber
    &\quad + \left\langle ( P_n^{\sharp} - P_n ) \phi_n^{[ \sigma ]_{i}}, ( P_n^{\sharp} - P_n ) \phi_n^{[ \sigma^\prime ]_{i}} \right\rangle_\mathcal{H} - \left\langle ( P_n^{\sharp} - P_n ) \phi_n^{[ \sigma ]_{i} \cup \lbrace i \rbrace}, ( P_n^{\sharp} - P_n ) \phi_n^{[ \sigma^\prime ]_{i}} \right\rangle_\mathcal{H} \\
    \nonumber
    &\quad \left. - \left\langle ( P_n^{\sharp} - P_n ) \phi_n^{[ \sigma ]_{i}}, ( P_n^{\sharp} - P_n ) \phi_n^{[ \sigma^\prime ]_{i} \cup \lbrace i \rbrace} \right\rangle_\mathcal{H} \right\rbrace \\
    \nonumber
    &= \frac{n}{(d! )^2} \left[ \sum_{\sigma \in \mathcal{S}_d} \lbrace \mathbf{\tilde{\alpha}}_{[ \sigma ]_{i} \cup \lbrace i \rbrace} +  \mathbf{W}_{[ \sigma ]_{i} \cup \lbrace i \rbrace} \mathbf{K}_{[ \sigma ]_{i} \cup \lbrace i \rbrace} \mathbf{\tilde{\beta}}_{[ \sigma ]_{i} \cup \lbrace i \rbrace} - ( \mathbf{\tilde{\alpha}}_{[ \sigma ]_{i}} +  \mathbf{W}_{[ \sigma ]_{i}} \mathbf{K}_{[ \sigma ]_{i}} \mathbf{\tilde{\beta}}_{[ \sigma ]_{i}} ) \rbrace \right]^\intercal \\
    &\quad \times \mathbf{K} \left[ \sum_{\sigma^\prime \in \mathcal{S}_d} \lbrace \mathbf{\tilde{\alpha}}_{[ \sigma^\prime ]_{i} \cup \lbrace i \rbrace} +  \mathbf{W}_{[ \sigma^\prime ]_{i} \cup \lbrace i \rbrace} \mathbf{K}_{[ \sigma^\prime ]_{i} \cup \lbrace i \rbrace} \mathbf{\tilde{\beta}}_{[ \sigma^\prime ]_{i} \cup \lbrace i \rbrace} - ( \mathbf{\tilde{\alpha}}_{[ \sigma^\prime ]_{i}} +  \mathbf{W}_{[ \sigma^\prime ]_{i}} \mathbf{K}_{[ \sigma^\prime ]_{i}} \mathbf{\tilde{\beta}}_{[ \sigma^\prime ]_{i}} ) \rbrace \right].
\end{align}

\section{Additional results} \label{SectionA5}

For $\mathcal{V} \subseteq [d]$, define the parameter $\nu_\cdot^\mathcal{K} (\mathcal{V} ) : \mathcal{M}\rightarrow\mathcal{H}$ as follows:
\begin{align*}
    \nu_P^\mathcal{K} ( \mathcal{V} ) ( x ) 
    &= \int_\mathcal{X} \mathcal{K} ( x,x^\prime ) \mathbb{E}_P \lbrace \mathbb{E}_P ( Y \mid A=1, X ) \mid X_\mathcal{V} = x^\prime_{\mathcal{V}} \rbrace P_{X} ( dx^\prime ).
\end{align*}
\begin{lemma}\label{lem:pd}
    Fix $\mathcal{V} \subseteq [d]$. Under the conditions of Theorem~\ref{Theorem1}, $\nu_\cdot^{\mathcal{K}}(\mathcal{V})$ is pathwise differentiable at any $P\in\mathcal{M}$, and the adjoint of its local parameter is given by
    \begin{align}
        \dot{\nu}_{P}^{\mathcal{K},*} ( h; \mathcal{V} ) ( z ) &= \mathbb{E}_P \lbrace h ( X ) \mid X_\mathcal{V}=x_\mathcal{V} \rbrace \left[ \frac{a}{g_P ( 1 \mid x )} \lbrace y - \mu_P (1,x) \rbrace + \mu_P (1,x) \right] \nonumber \\
        &\quad + [ h ( x ) - \mathbb{E}_P \lbrace h ( X ) \mid X_\mathcal{V}=x_\mathcal{V} \rbrace ] \mathbb{E}_P \lbrace \mathbb{E}_P ( Y \mid A=1, X ) \mid X_\mathcal{V}=x_\mathcal{V} \rbrace \nonumber \\
    &\quad - \mathbb{E}_P [ h ( X )  \mathbb{E}_P \lbrace \mathbb{E}_P ( Y \mid A=1, X ) \mid X_\mathcal{V} \rbrace ]. \label{eq:noAdjoint}
    \end{align}
\end{lemma}

\begin{theorem}[Consistency against fixed alternatives] \label{Theorem5}

Assume that $\hat{\gamma}_{\omega}^\mathcal{K} \xrightarrow{P} \gamma_{\omega}^{\mathcal{K}}(P_0)$ as $n \to \infty$ under any fixed distribution $P_0$. Then the test which rejects $H_0$ if and only if
\[
0 \notin \mathcal{C}_n^\omega(\hat{\xi}_{\omega})
\]
is consistent against fixed alternatives, i.e., if $\gamma_{\omega}^{\mathcal{K}}(P_0) \ne 0$, then
\[
\lim_{n \to \infty} P_{0} \left( 0 \notin \mathcal{C}_n^\omega(\hat{\xi}_{\omega}) \right) = 1.
\]
\end{theorem}

\begin{theorem}[Local power of the hypothesis test of no importance] \label{Theorem4}
    Let $\mathcal{K}$ be a kernel and $h_0 \coloneqq \gamma_\omega^{\mathcal{K}}(P_0) \in \mathcal{H}$. Suppose $\gamma_\omega^\mathcal{K}$ is pathwise differentiable at $P_0$ with influence function $\phi_0^{\omega}$ and $\widehat{\xi}_{\omega,n}$ is a consistent estimator of $\xi_{\omega, 1-\alpha}$, the $\left( 1-\alpha \right)$-quantile of $\| \mathbb{H}_\omega \|_\mathcal{H}^2$. Fix $\left\lbrace P_\epsilon \colon \epsilon \in \left[ 0,\delta \right) \right\rbrace \in \mathscr{P} \left( P_0, \mathcal{M}, s \right)$ such that $\| \dot{\gamma}_{0,\omega}^\mathcal{K} \left( s \right) \|_\mathcal{H} > 0$. Then 
    \begin{equation*}
        P_{\epsilon = n^{-1/2}}^n \left\lbrace h_0 \notin \mathcal{C}_n^\omega \left( \hat{\xi}_{\omega}\right) \right\rbrace \xrightarrow{n \rightarrow \infty} \emph{Pr} \left\lbrace \| \mathbb{H}_\omega + \dot{\gamma}_{0,\omega}^\mathcal{K} \left( s \right) \|_\mathcal{H}^2 > \xi_{\omega, 1-\alpha} \right\rbrace > \alpha.
    \end{equation*}
    Furthermore, $h_n \coloneqq \gamma_\omega^\mathcal{K} \left( P_{\epsilon = n^{-1/2}}^n \right)$ is an $n^{1/2}$-rate local alternative, i.e. $\| h_n - h_0 \|_\mathcal{H} = \mathcal{O} \left( n^{-1/2} \right)$.
\end{theorem}
Theorem \ref{Theorem4} establishes that we can detect the class of local alternatives defined via smooth parametric submodels of $\mathcal{M}$. These models are characterized by a first-order direction, given by the value of the local parameter $\dot{\gamma}_{0,\omega}^\mathcal{K} \left( s \right)$, which remains fixed as the sample size increases. It is important to note that our test does not possess non-trivial asymptotic power against local alternatives whose direction varies with the sample size. For a more detailed discussion on this point, see \cite{Luedtke2024}.

\section{Proofs} \label{SectionA6}

\begin{proof}[Proof of Lemma~\ref{lem:pd}]
    We will write $\nu_\cdot^{\mathcal{K}}(\mathcal{V})$ as a composition of smooth maps and then apply a chain rule to establish its pathwise differentiability at $P$. Concretely, we will exhibit totally pathwise differentiable maps $\theta_j$, $j\in [6]$, such that
    \begin{align}
        \nu^{\mathcal{K}}(P')=\theta_6(P',\theta_5(P',\theta_4(P',\theta_3(P',\theta_2(P',\theta_1(P',0)))))) \ \textnormal{ for all } P'\in\mathcal{M}. \label{eq:nuComp}
    \end{align}
    This will allow us to apply Theorem~1 from \cite{luedtke2025simplifying} to establish the pathwise differentiability of the composition $\nu_\cdot^{\mathcal{K}}(\mathcal{V})$ and obtain an expression for the adjoint of its local parameter.
    
    For any probability measure $Q$, let $\mathcal{B}(Q)\coloneqq\{u\in L^\infty(P_{\mathcal{X}_\mathcal{V}}) : \|u\|_{L^\infty(P_{\mathcal{X}_\mathcal{V}})}\le m\}\subset L^2(P_{X_{\mathcal{V}}})$, with $m\coloneqq\sup_{P\in\mathcal{M}}\|\mathbb{E}_P[Y^2\mid (A,X)=\cdot\,]^{1/2}\|_{L^\infty(P_{A,X})}$ the quantity that \ref{cond:finitesecondmoment} shows is finite. Define the maps satisfying \eqref{eq:nuComp} as follows:
    \begin{align*}
        &\theta_1 : \mathcal{M}\times \{0\}\rightarrow \mathcal{B}(P_{A,X}),\; (P',0)\mapsto E_{P'}[Y\mid (A,X)=\cdot] \tag{conditional mean} \\
        &\theta_2 : \mathcal{M}\times \mathcal{B}(P_{A,X})\rightarrow \mathcal{B}(P_X),\; (P',u)\mapsto u(1,x) \tag{fix binary argument} \\
        &\theta_3 : \mathcal{M}\times \mathcal{B}(P_X)\rightarrow \mathcal{B}(P),\; (P',u)\mapsto [(x,a,y)\mapsto u(x)] \tag{lift domain} \\
        &\theta_4 : \mathcal{M}\times \mathcal{B}(P)\rightarrow \mathcal{B}(P_{X_{\mathcal{V}}}),\; (P',u)\mapsto E_{P'}\left[u(X,A,Y)\mid X_{\mathcal{V}}=\cdot\, \right] \tag{conditional mean} \\
        &\theta_5 : \mathcal{M}\times \mathcal{B}(P_{X_{\mathcal{V}}})\rightarrow \mathcal{B}(P_X),\; (P',u)\mapsto [x\mapsto u(x_{\mathcal{V}})] \tag{lift domain} \\
        &\theta_6 : \mathcal{M}\times \mathcal{B}(P_{X})\rightarrow \mathcal{H},\; (P',u)\mapsto {\textstyle\int} \mathcal{K}(\,\cdot\,,x)\,u(x)\, P_{X}'(dx). \tag{kernel embedding}
    \end{align*}
    Note that, by \ref{cond:finitesecondmoment} and Jensen's inequality, $\mathcal{B}(P_X)$ is indeed a codomain of $\theta_1$; similarly, since the domains for $u$ of all of the other maps consist of functions a.s. bounded by $m$, the outputs of these maps are also a.s. bounded by $m$.
    
    Under the conditions of Theorem~\ref{Theorem1}, each of the maps in \eqref{eq:nuComp} is totally pathwise differentiable at a generic $(P',u')$ in its domain with $P'=P$. Here we cite the relevant results. For $\theta_1$, we leverage \ref{cond:finitesecondmoment} and Example~5 in \cite{Luedtke2024} to show that $\theta_1(\,\cdot\,,0)$ is pathwise differentiable; Appendix~C.3.1 in \cite{luedtke2025simplifying} then shows $\theta_1$ is totally pathwise differentiable at $(P,0)$, and its differential operator has adjoint
    \begin{align*}
        \dot{\theta}_{1,P,0}^*(w)&= \left(z\mapsto w(x)[y-\mu_P(a,x)],0\right).
    \end{align*}
    For $\theta_2$, the results in Appendix~C.4.11 of \cite{luedtke2025simplifying} and the fact that \ref{cond:positivity} holds show $\theta_2$ is totally pathwise differentiable at $(P,u)$ with
    \begin{align*}
        \dot{\theta}_{2,P,0}^*(w)&= \left(0,(a,x)\mapsto \frac{a}{g_P(a|x)}\,w(x)\right).
    \end{align*}
    For $\theta_3$, the results in Appendix~C.4.10 of  \cite{luedtke2025simplifying} yield total pathwise differentiability with
    \begin{align*}
        \dot{\theta}_{3,P,u}^*(w)&= \left(0,E_P\left[w(Z)\mid X=\cdot\,\right]\right).
    \end{align*}
    For $\theta_4$, we again use Lemma~S8 from \cite{luedtke2025simplifying} with $r=1$ to show $\theta_4$ is totally pathwise differentiable at $(P,u)$ with
    \begin{align*}
        \dot{\theta}_{4,P,u}^*(w)&= \left(z\mapsto [u(z)-\theta_4(P,u)(x_{\mathcal{V}})]w(x_{\mathcal{V}}),z\mapsto w(x)\right).
    \end{align*}
    For $\theta_5$, we again apply the results in Appendix~C.4.10 of \cite{luedtke2025simplifying} to find total pathwise differentiability with
    \begin{align*}
        \dot{\theta}_{5,P,u}^*(w)&= \left(0,E_P\left[w(X)\mid X_{\mathcal{V}}=\cdot\,\right]\right).
    \end{align*}
    Finally, for $\theta_6$, \ref{cond:bounded} and Lemma~S11 yield total pathwise differentiability with
    \begin{align*}
         \dot{\theta}_{6,P,u}^*(w)&= \left(wu-E_P[u(X_{\mathcal{V}})w(X_{\mathcal{V}})],w\right).
    \end{align*}

    The proof is at hand. Eq.~\ref{eq:nuComp} shows $\nu_\cdot^{\mathcal{K}}(\mathcal{V})$ writes as in the composition in Algorithm~1 from \cite{luedtke2025simplifying}, with $\mathrm{pa}(1)=\emptyset$ and $\mathrm{pa}(j+1)=\{j\}$, $j\in [5]$. Since $\theta_j$, $j\in [6]$, are all totally pathwise differentiable, Theorem~1 from that work shows $\nu_\cdot^{\mathcal{K}}(\mathcal{V})$ is totally pathwise differentiable; that result further shows that Algorithm~2 from that work provides the form of the adjoint. Plugging the forms of the $\theta_j$ maps and the adjoints of their differential operators into that algorithm shows $\dot{\nu}_{P}^{\mathcal{K},*} ( h; \mathcal{V} )$ writes as in \eqref{eq:noAdjoint}.
\end{proof}

\begin{proof}[Proof of Theorem \ref{Theorem1}]

It follows from Lemma~\ref{lem:pd} that $\nu_\cdot^\mathcal{K} ( \mathcal{V} )$ is pathwise differentiable with the adjoint operator $\dot{\nu}^{\mathcal{K},*}_{P}$. We wish to apply Theorem~1 in \cite{Luedtke2024} to show the EIF takes the form 
\begin{align*}
    \phi_{P}^\mathcal{V} ( z ) ( x^\prime ) &= \dot{\nu}^{\mathcal{K},*}_{P} ( \mathcal{K} ( \cdot, x^\prime ); \mathcal{V} ) ( z ) \\
    &= \mathbb{E}_P \lbrace \mathcal{K} ( x^\prime, X ) \mid X_\mathcal{V}=x_\mathcal{V} \rbrace \left[ \frac{a}{g_P ( 1 \mid x )} \lbrace y - \mu_P (1,x) \rbrace + \mu_P (1,x) \right] \\
    &\quad + [ \mathcal{K} ( x^\prime, x ) - \mathbb{E}_P \lbrace \mathcal{K} ( x^\prime, X ) \mid X_\mathcal{V}=x_\mathcal{V} \rbrace ] \mathbb{E}_P \lbrace \mu_P (1,X) \mid X_\mathcal{V}=x_\mathcal{V} \rbrace \\
    &\quad - \underbrace{ \mathbb{E}_P [ \mathcal{K} (x^\prime, X )  \mathbb{E}_P \lbrace \mu_P (1, X) \mid X_\mathcal{V} \rbrace ] }_{\nu_P^\mathcal{K} ( \mathcal{V} ) ( x^\prime )}.
\end{align*}
The corresponding results for the parameter $\gamma_\omega^\mathcal{K} (P)$ follow directly from linearity.

To apply Theorem~1 in \cite{Luedtke2024}, we need to show that $\| \phi_{P}^\mathcal{V} \|_{L^2(P; \mathcal{H})}^2 \coloneqq \mathbb{E}_{P} \left[ \| \phi_{P}^\mathcal{V}(Z) \|_{\mathcal{H}}^2 \right]< \infty$. We have the following expansion for the squared RKHS norm:
\begin{align*}
\| \phi_{P}^\mathcal{V}(z) \|_{\mathcal{H}}^2 &= \lbrace \psi_P (z) - \nu_P (\mathcal{V})(x) \rbrace^2 \cdot \| \kappa_{P,\mathcal{V}} (X_{\mathcal{V}}) \|_{\mathcal{H}}^2 + \nu_P (\mathcal{V})(x)^2 \cdot \mathcal{K}(x, x) + \| \nu_P^\mathcal{K} ( \mathcal{V} ) \|_{\mathcal{H}}^2 \\
&\quad + 2 \lbrace \psi_P (z) - \nu_P ( \mathcal{V})(x) \rbrace \nu_P (\mathcal{V})(x) \langle \kappa_{P,\mathcal{V}} (X_{\mathcal{V}}) (\cdot), \mathcal{K}(\cdot, x) \rangle_{\mathcal{H}} \\
&\quad - 2 \lbrace \psi_P (z) - \nu_P ( \mathcal{V})(x) \rbrace \langle \kappa_{P,\mathcal{V}} (X_{\mathcal{V}}) (\cdot), \nu_P^\mathcal{K} ( \mathcal{V} ) (\cdot) \rangle_{\mathcal{H}} \\
&\quad - 2 \nu_P (\mathcal{V})(x) \langle \mathcal{K}(\cdot, x), \nu_P^\mathcal{K} ( \mathcal{V} ) (\cdot) \rangle_{\mathcal{H}}. \\
&= \lbrace \psi_P (z) - \nu_P ( \mathcal{V})(x) \rbrace^2 \cdot \| \kappa_{P,\mathcal{V}} (X_{\mathcal{V}}) \|_{\mathcal{H}}^2 + \nu_P ( \mathcal{V})(x)^2 \cdot \mathcal{K}(x, x) + \| \nu_P^\mathcal{K} ( \mathcal{V} ) \|_{\mathcal{H}}^2 \\
&\quad + 2 \lbrace \psi_P (z) - \nu_P ( \mathcal{V})(x) \rbrace \nu_P (\mathcal{V})(x) \kappa_{P,\mathcal{V}} (X_{\mathcal{V}}) (x) \\
&\quad - 2 \lbrace \psi_P(z) - \nu_P ( \mathcal{V})(x) \rbrace \cdot \mathbb{E}_{P}\left[ \nu_P (\mathcal{V})(X) \cdot \mathbb{E}_P \lbrace \mathcal{K}(X, X^\prime) \mid X_\mathcal{V} = x_\mathcal{V} \rbrace \right] \\
&\quad - 2 \nu_P (\mathcal{V})(x) \nu_P^\mathcal{K} ( \mathcal{V} ) (x).
\end{align*}
We wish to compute the expected squared RKHS norm:
\begin{align*}
\| \phi_{P}^\mathcal{V} \|_{L^2(P; \mathcal{H})}^2 &= 
\mathbb{E}_P \left[ \operatorname{Var}_P(\psi_P(Z) \mid X_{\mathcal{V}}) \cdot \left\| \kappa_{P,\mathcal{V}} (X_{\mathcal{V}}) \right\|_{\mathcal{H}}^2 \right] + \mathbb{E}_P \left[ \nu_P (\mathcal{V})(X)^2 \cdot \mathcal{K}(X,X) \right] \\
&\quad + \left\| \nu_P^\mathcal{K}( \mathcal{V}) \right\|_{\mathcal{H}}^2 + 2 \mathbb{E}_P \left[ \left( \psi_P(Z) - \nu_P (\mathcal{V})(X) \right) \cdot \nu_P (\mathcal{V})(X) \cdot \kappa_{P,\mathcal{V}} (X_{\mathcal{V}}) (X) \right] \\
&\quad - 2 \mathbb{E}_P \left[ \left( \psi_P(Z) - \nu_P (\mathcal{V})(X) \right) \cdot \mathbb{E}_{P}\left[ \nu_P (\mathcal{V})(X) \cdot \mathbb{E}_P \lbrace \mathcal{K}(X, X^\prime) \mid X_\mathcal{V} = x_\mathcal{V} \rbrace \right] \right] \\
&\quad - 2 \mathbb{E}_P \left[ \nu_P (\mathcal{V})(X) \cdot \nu_P^\mathcal{K}(\mathcal{V})(X) \right].
\end{align*}
All the terms are finite under conditions \ref{cond:bounded}-\ref{cond:finitesecondmoment}. Therefore, the norm $\| \phi_{P}^\mathcal{V} \|_{L^2(P; \mathcal{H})}^2 = \mathbb{E}_{P} \left[ \| \phi_{P}^\mathcal{V}(z) \|_{\mathcal{H}}^2 \right]$ is finite under the assumptions stated in Theorem \ref{Theorem1}.
\end{proof}

\begin{proof}[Proof of Theorem \ref{Theorem2}]
To establish asymptotic linearity, we show that $\| \hat{\gamma}^\mathcal{K}_\omega - \gamma_\omega^\mathcal{K} ( P_0 ) - P_n \phi_{0}^\omega \|_{\mathcal{H}}=o_P(n^{-1/2})$. To this end, consider
\begin{align*}
    &\| \hat{\gamma}^\mathcal{K}_\omega - \gamma_\omega^\mathcal{K} ( P_0 ) - P_n \phi_{0}^\omega \|_{\mathcal{H}} \\
    &= \| \gamma_\omega^\mathcal{K} ( \hat{P}_n ) + P_0 \phi_{n}^\omega - \gamma_\omega^\mathcal{K} ( P_0 ) + ( P_n - P_0 ) ( \phi_{n}^\omega - \phi_{0}^\omega ) \|_\mathcal{H} \\
    &= \| \sum_{\mathcal{V} \subseteq [d]} \omega_\mathcal{V} \nu_n^\mathcal{K} ( \mathcal{V} ) + P_0 \sum_{\mathcal{V} \subseteq [d]} \omega_\mathcal{V} \phi_{n}^\mathcal{V} - \sum_{\mathcal{V} \subseteq [d]} \omega_\mathcal{V} \nu_0^\mathcal{K} ( \mathcal{V} ) + ( P_n - P_0 ) \sum_{\mathcal{V} \subseteq [d]} \omega_\mathcal{V} ( \phi_{n}^\mathcal{V} - \phi_{0}^\mathcal{V} ) \|_\mathcal{H} \\
    &\leq \| \mathcal{R}_{n}^\omega \|_\mathcal{H} + \| \mathcal{D}_{n}^\omega \|_\mathcal{H},
\end{align*}
where
\begin{align*}
    \mathcal{R}_{n}^\omega &\coloneqq \sum_{\mathcal{V} \subseteq [d]} \omega_\mathcal{V} \nu_n^\mathcal{K} ( \mathcal{V} ) + P_0 \sum_{\mathcal{V} \subseteq [d]} \omega_\mathcal{V} \phi_{n}^\mathcal{V} - \sum_{\mathcal{V} \subseteq [d]} \omega_\mathcal{V} \nu_0^\mathcal{K} ( \mathcal{V} ) = \sum_{\mathcal{V} \subseteq [d]} \omega_\mathcal{V} \mathcal{R}_{n}^\mathcal{V} \\
    \mathcal{D}_{n}^\omega &\coloneqq ( P_n - P_0 ) \sum_{\mathcal{V} \subseteq [d]} \omega_\mathcal{V} ( \phi_{n}^\mathcal{V} - \phi_{0}^\mathcal{V} ) = \sum_{\mathcal{V} \subseteq [d]} \omega_\mathcal{V} \mathcal{D}_{n}^\mathcal{V}
\end{align*}
and
\begin{align*}
    \mathcal{R}_{n}^\mathcal{V} &\coloneqq \nu_n^\mathcal{K} ( \mathcal{V} ) + P_0 \phi_{n}^\mathcal{V} - \nu_0^\mathcal{K} ( \mathcal{V} ) \\
    \mathcal{D}_{n}^\mathcal{V} &\coloneqq ( P_n - P_0 ) ( \phi_{n}^\mathcal{V} - \phi_{0}^\mathcal{V} ).
\end{align*}
Therefore, it follows from $\| \mathcal{R}_{n}^\mathcal{V} \|_\mathcal{H} = o_{P} ( n^{-1/2} )$ and $\| \mathcal{D}_{n}^\mathcal{V} \|_\mathcal{H} = o_{P} ( n^{-1/2} )$, for all $\mathcal{V}$ with $\omega_{\mathcal{V}}\not=0$, that $\| \mathcal{R}_{n}^\omega \|_\mathcal{H} = o_{P} ( n^{-1/2} )$ and $\| \mathcal{D}_{n}^\omega \|_\mathcal{H} = o_{P} ( n^{-1/2} )$. Hence, we need to show that $\|\mathcal{R}_{n}^\mathcal{V} \|_\mathcal{H} = o_{P} ( n^{-1/2} )$ and $\| \mathcal{D}_{n}^\mathcal{V} \|_\mathcal{H} = o_{P} ( n^{-1/2} )$ for all $\mathcal{V}$ with $\omega_{\mathcal{V}}\not=0$. We have:
\begin{align*}
    \nu_0^\mathcal{K} ( \mathcal{V} ) &= \int_\mathcal{X} \mathcal{K} ( \cdot, x^\prime ) \nu_0 ( \mathcal{V} ) ( x^\prime ) P_{0,X} ( dx^\prime ), \\
    \nu_n^\mathcal{K} ( \mathcal{V} ) &= P_n \left[ \mathcal{K} ( \cdot, X ) P_n \left\lbrace \mu_{n} ( 1, X ) - \mu_{n} ( 0, X ) \mid X_{\mathcal{V}} \right\rbrace \right],
\end{align*}
and
\begin{align*}
    P_0 \phi_{n}^\mathcal{V} &= \int \kappa_{n,\mathcal{V}} (x_\mathcal{V})(\cdot) \left[ \frac{a}{g_n ( 1 \mid x )} \left\lbrace y - \mu_{n} ( 1,x ) \right\rbrace - \frac{1-a}{g_n ( 0 \mid x )} \left\lbrace y - \mu_{n} ( 0,x ) \right\rbrace \right] \\
    &\quad + \kappa_{n,\mathcal{V}} (x_\mathcal{V})(\cdot) \left\lbrace \mu_{n} ( 1,x ) - \mu_{n} ( 0,x ) \right\rbrace + \left[ \mathcal{K} ( \cdot, x ) - \kappa_{n,\mathcal{V}} (x_\mathcal{V})(\cdot) \right] \nu_n (\mathcal{V}) ( x_{\mathcal{V}} ) \\
    &\quad - P_n \left[ \mathcal{K} ( \cdot, X ) \nu_n (\mathcal{V}) ( X_{\mathcal{V}} ) \right] P_0 ( dz ) \\
    &= \int \kappa_{n,\mathcal{V}} (x_\mathcal{V})(\cdot) \left[ \frac{g_0 ( 1 \mid x )}{g_n ( 1 \mid x )} \left\lbrace \mu_{0} ( 1, x ) - \mu_{n} ( 1,x ) \right\rbrace - \frac{g_0 ( 0 \mid x )}{g_n ( 0 \mid x )} \left\lbrace \mu_{0} ( 0,x ) - \mu_{n} ( 0,x ) \right\rbrace \right] \\
    &\quad + \kappa_{n,\mathcal{V}} (x_\mathcal{V})(\cdot) \left\lbrace \mu_{n} ( 1,x ) - \mu_{n} ( 0,x ) \right\rbrace + \left[ \mathcal{K} ( \cdot, x ) - \kappa_{n,\mathcal{V}} (x_\mathcal{V})(\cdot) \right] \nu_n (\mathcal{V}) ( x_{\mathcal{V}} ) \\
    &\quad - P_n \left[ \mathcal{K} ( \cdot, X ) \nu_n (\mathcal{V}) ( X_{\mathcal{V}} ) \right] P_0 ( dz ).
\end{align*}
In what follows, we separately study the remainder term ($\mathcal{R}_{n}^\mathcal{V}$) and the drift term ($\mathcal{D}_{n}^\mathcal{V}$).

\subsubsection*{Remainder term, $\mathcal{R}_{n}^\mathcal{V}$}

We have the following:
\begin{align*}
    \mathcal{R}_{n}^\mathcal{V} &= \nu_n^\mathcal{K} ( \mathcal{V} ) + P_0 \phi_{n}^\mathcal{V} - \nu_0^\mathcal{K} ( \mathcal{V} ) \\
    &= P_n \left[ \mathcal{K} ( \cdot, X ) \nu_n (\mathcal{V})(X_{\mathcal{V}}) \right] + \sum_{a=0}^1 (2a-1) \int \kappa_{n,\mathcal{V}} (x_\mathcal{V})(\cdot) \frac{g_0 ( a \mid x )}{g_n ( a \mid x )} \left\lbrace \mu_{0} ( a, x ) - \mu_{n} ( a,x ) \right\rbrace P_0 ( dz ) \\
    &\quad + \int \kappa_{n,\mathcal{V}} (x_\mathcal{V})(\cdot) \left\lbrace \mu_{n} ( 1,x ) - \mu_{n} ( 0,x ) \right\rbrace P_0 ( dz ) + \int \left[ \mathcal{K} ( \cdot, x ) - \kappa_{n,\mathcal{V}} (x_\mathcal{V})(\cdot) \right] \nu_n (\mathcal{V}) ( x_{\mathcal{V}} ) P_0 ( dz ) \\
    &\quad - \int P_n \left[ \mathcal{K} ( \cdot, X ) \nu_n (\mathcal{V}) ( X_{\mathcal{V}} ) \right] P_0 ( dz ) - \int_\mathcal{X} \mathcal{K} ( \cdot, x^\prime ) \nu_0 ( \mathcal{V} ) ( x^\prime ) P_{0,X} ( dx^\prime )\\
    &= \sum_{a=0}^1 (2a-1) \int \kappa_{n,\mathcal{V}} (x_\mathcal{V})(\cdot) \frac{g_0 ( a \mid x )}{g_n ( a \mid x )} \left\lbrace \mu_{0} ( a, x ) - \mu_{n} ( a,x ) \right\rbrace P_0 ( dz ) \\
    &\quad + \int \kappa_{n,\mathcal{V}} (x_\mathcal{V})(\cdot) \left\lbrace \mu_{n} ( 1,x ) - \mu_{n} ( 0,x ) \right\rbrace P_0 ( dz ) - \int \kappa_{n,\mathcal{V}} (x_\mathcal{V})(\cdot) \left\lbrace \mu_{0} ( 1,x ) - \mu_{0} ( 0,x ) \right\rbrace P_0 ( dz ) \\
    &\quad + \int \kappa_{n,\mathcal{V}} (x_\mathcal{V})(\cdot) \left\lbrace \mu_{0} ( 1,x ) - \mu_{0} ( 0,x ) \right\rbrace P_0 ( dz ) - \int \kappa_{n,\mathcal{V}} (x_\mathcal{V})(\cdot) \nu_n (\mathcal{V}) ( x_{\mathcal{V}} ) P_0 ( dz ) \\
    &\quad - \int_\mathcal{X} \mathcal{K} ( \cdot, x^\prime ) \left\lbrace \nu_n ( \mathcal{V} ) ( x^\prime ) - \nu_0 ( \mathcal{V} ) ( x^\prime ) \right\rbrace P_{0,X} ( dx^\prime ) \\
    &= \sum_{a=0}^1 (2a-1) \int \kappa_{n,\mathcal{V}} (x_\mathcal{V})(\cdot) \frac{g_0 ( a \mid x )}{g_n ( a \mid x )} \left\lbrace \mu_{0} ( a, x ) - \mu_{n} ( a,x ) \right\rbrace P_0 ( dz ) \\
    &\quad + \int \kappa_{n,\mathcal{V}} (x_\mathcal{V})(\cdot) \left\lbrace \mu_{n} ( 1,x ) - \mu_{0} ( 1, x ) \right\rbrace P_0 ( dz ) - \int \kappa_{n,\mathcal{V}} (x_\mathcal{V})(\cdot) \left\lbrace \mu_{n} ( 0,x ) - \mu_{0} ( 0,x ) \right\rbrace P_0 ( dz ) \\
    &\quad + \int \kappa_{n,\mathcal{V}} (x_\mathcal{V})(\cdot) \left\lbrace \mu_{0} ( 1,x ) - \mu_{0} ( 0,x ) \right\rbrace P_0 ( dz ) - \int \kappa_{n,\mathcal{V}} (x_\mathcal{V})(\cdot) \nu_n (\mathcal{V}) ( x_{\mathcal{V}} ) P_0 ( dz ) \\
    &\quad - \int_\mathcal{X} \mathcal{K} ( \cdot, x^\prime ) \left\lbrace \nu_n ( \mathcal{V} ) ( x^\prime ) - \nu_0 ( \mathcal{V} ) ( x^\prime ) \right\rbrace P_{0,X} ( dx^\prime ) \\
    &= \sum_{a=0}^1 (2a-1) \int \kappa_{n,\mathcal{V}} (x_\mathcal{V})(\cdot) \left\lbrace 1 - \frac{g_0 ( a \mid x )}{g_n ( a \mid x )} \right\rbrace \left\lbrace \mu_{0} ( a, x ) - \mu_{n} ( a,x ) \right\rbrace P_0 ( dz ) \\
    &\quad + \int \kappa_{n,\mathcal{V}} (x_\mathcal{V})(\cdot) \left\lbrace \mu_{0} ( 1,x ) - \mu_{0} ( 0,x ) \right\rbrace P_0 ( dz ) - \int \kappa_{n,\mathcal{V}} (x_\mathcal{V})(\cdot) \nu_n (\mathcal{V}) ( x_{\mathcal{V}} ) P_0 ( dz ) \\
    &\quad - \int_\mathcal{X} \mathcal{K} ( \cdot, x^\prime ) \left\lbrace \nu_n ( \mathcal{V} ) ( x^\prime ) - \nu_0 ( \mathcal{V} ) ( x^\prime ) \right\rbrace P_{0,X} ( dx^\prime ) \\
    &= \sum_{a=0}^1 (2a-1) \int \kappa_{n,\mathcal{V}} (x_\mathcal{V})(\cdot) \left\lbrace 1 - \frac{g_0 ( a \mid x )}{g_n ( a \mid x )} \right\rbrace \left\lbrace \mu_{0} ( a, x ) - \mu_{n} ( a,x ) \right\rbrace P_0 ( dz ) \\
    &\quad - \int_\mathcal{X} \left[ \kappa_{n,\mathcal{V}} (x_\mathcal{V})(\cdot) - \mathcal{K} ( \cdot, x^\prime ) \right] \left\lbrace \nu_n ( \mathcal{V} ) ( x^\prime ) - \nu_0 ( \mathcal{V} ) ( x^\prime ) \right\rbrace P_{0,X} ( dx^\prime ) \\
    &= \sum_{a=0}^1 (2a-1) \int \kappa_{n,\mathcal{V}} (x_\mathcal{V})(\cdot) \left\lbrace 1 - \frac{g_0 ( a \mid x )}{g_n ( a \mid x )} \right\rbrace \left\lbrace \mu_{0} ( a, x ) - \mu_{n} ( a,x ) \right\rbrace P_0 ( dz ) \\
    &\quad - \int_\mathcal{X} \left[ \kappa_{n,\mathcal{V}} (x_\mathcal{V})(\cdot) - \kappa_{0,\mathcal{V}} (x_\mathcal{V})(\cdot) \right] \left\lbrace \nu_n ( \mathcal{V} ) ( x^\prime ) - \nu_0 ( \mathcal{V} ) ( x^\prime ) \right\rbrace P_{0,X} ( dx^\prime ).
\end{align*}
We can now bound the squared norm of the first and second term in the sum above, for $a=0,1$:
\begin{align*}
& \left\| \int \kappa_{n,\mathcal{V}} (x_\mathcal{V}) \left\lbrace 1 - \frac{g_0 ( a \mid x )}{g_n ( a \mid x )} \right\rbrace \left\lbrace \mu_{0} ( a, x ) - \mu_{n} ( a,x ) \right\rbrace P_{0,X} ( dx ) \right\|^2_\mathcal{H} \\
& \leq \left( \int \left\| \kappa_{n,\mathcal{V}} (x_\mathcal{V}) \right\|^2_\mathcal{H} P_{0,X} ( dx ) \right)  \left( \int \left\lbrace 1 - \frac{g_0 ( a \mid x )}{g_n ( a \mid x )} \right\rbrace^2 \left\lbrace \mu_{0} ( a, x ) - \mu_{n} ( a,x ) \right\rbrace^2 P_{0,X} ( dx ) \right) \\
&= \left\| \kappa_{n,\mathcal{V}} (X_\mathcal{V}) \right\|^2_{L^2(P_0;\mathcal{H})}
\left\| \left\lbrace 1 - \frac{g_0 ( a \mid x )}{g_n ( a \mid x )} \right\rbrace \left\lbrace \mu_{0} ( a, x ) - \mu_{n} ( a,x ) \right\rbrace \right\|^2_{L^2(P_0)} \\
&\leq \left\| \kappa_{n,\mathcal{V}} (X_\mathcal{V}) \right\|^2_{L^2(P_0;\mathcal{H})}
\left\| 1 - \frac{g_0 ( a \mid x )}{g_n ( a \mid x )} \right\|^2_{L^4(P_0)} \left\| \mu_{0} ( a, x ) - \mu_{n} ( a,x ) \right\|^2_{L^4(P_0)}
\end{align*}

We now bound the squared norm of the remaining term:
\begin{align*}
    &\left\|\int_\mathcal{X} \left[ \kappa_{n,\mathcal{V}} (x_\mathcal{V}) - \kappa_{0,\mathcal{V}} (x_\mathcal{V}) \right] \left\lbrace \nu_n ( \mathcal{V} ) ( x^\prime ) - \nu_0 ( \mathcal{V} ) ( x^\prime ) \right\rbrace P_{0,X} ( dx^\prime )\right\|_{\mathcal{H}}^2 \\
    &= \iint \left\langle \kappa_{n,\mathcal{V}} (x_{1,\mathcal{V}}) - \kappa_{0,\mathcal{V}} (x_{1,\mathcal{V}}), \kappa_{n,\mathcal{V}} (x_{2,\mathcal{V}}) - \kappa_{0,\mathcal{V}} (x_{2,\mathcal{V}}) \right\rangle_{\mathcal{H}} \\
    &\hspace{3em}\times\left[\prod_{j=1}^2 \left\lbrace \nu_n ( \mathcal{V} ) ( x_j ) - \nu_0 ( \mathcal{V} ) ( x_j ) \right\rbrace \right]P_{0,X} ( dx_1 )P_{0,X} ( dx_2 ) \\
    &\le \iint \left\| \kappa_{n,\mathcal{V}} (x_{1,\mathcal{V}}) - \kappa_{0,\mathcal{V}} (x_{1,\mathcal{V}}) \right\|_{\mathcal{H}}\left\| \kappa_{n,\mathcal{V}} (x_{2,\mathcal{V}}) - \kappa_{0,\mathcal{V}} (x_{2,\mathcal{V}}) \right\|_{\mathcal{H}} \\
    &\hspace{3em}\times\left|\prod_{j=1}^2 \left\lbrace \nu_n ( \mathcal{V} ) ( x_j ) - \nu_0 ( \mathcal{V} ) ( x_j ) \right\rbrace \right|P_{0,X} ( dx_1 )P_{0,X} ( dx_2 ) \\
    &= \left[\int \left\| \kappa_{n,\mathcal{V}} (x_{\mathcal{V}}) - \kappa_{0,\mathcal{V}} (x_{\mathcal{V}}) \right\|_{\mathcal{H}}\left|\nu_n ( \mathcal{V} ) ( x ) - \nu_0 ( \mathcal{V} ) ( x ) \right|P_{0,X} ( dx )\right]^2 \\
    &\le \left[\int \left\| \kappa_{n,\mathcal{V}} (x_{\mathcal{V}}) - \kappa_{0,\mathcal{V}} (x_{\mathcal{V}}) \right\|_{\mathcal{H}}^2 P_{0,X} ( dx )\right]\left[\int \left|\nu_n ( \mathcal{V} ) ( x ) - \nu_0 ( \mathcal{V} ) ( x ) \right|^2 P_{0,X} ( dx )\right] \\
    &= \left\| \kappa_{n,\mathcal{V}} (X_\mathcal{V}) - \kappa_{0,\mathcal{V}} (X_\mathcal{V}) \right\|_{L^2(P_{0,X};\mathcal{H})}^2 \left\| \nu_n ( \mathcal{V} ) - \nu_0 ( \mathcal{V} ) \right\|_{L^2 ( P_{0,X} )}^2 
\end{align*}
Therefore, we have the following upper bound:
\begin{align*}
    \left\| \mathcal{R}_{n}^\mathcal{V} \right\|_\mathcal{H} &\le 2 \left\| \kappa_{n,\mathcal{V}} (X_{\mathcal{V}}) \right\|_{L^2(P_{0,X};\mathcal{H})} \max_{a=0,1} \left\lbrace \left\| 1 - \frac{g_0 ( a \mid x )}{g_n ( a \mid x )} \right\|_{L^4 ( P_{0,X} )} \left\| \mu_{0} ( a, x ) - \mu_{n} ( a,x ) \right\|_{L^4 ( P_{0,X} )} \right\rbrace \\
    &\quad + \left\| \kappa_{n,\mathcal{V}} (X_{\mathcal{V}}) - \kappa_{0,\mathcal{V}} (X_{\mathcal{V}}) \right\|_{L^2(P_{0,X};\mathcal{H})} \left\| \nu_n ( \mathcal{V} ) - \nu_0 ( \mathcal{V} ) \right\|_{L^2 ( P_{0,X} )}.
\end{align*}
Therefore, $\| \mathcal{R}_{n}^\mathcal{V} \|_\mathcal{H} = o_{P} ( n^{-1/2} )$ whenever 
\begin{align*}
    \max_{a=0,1} \left\lbrace \left\| 1 - \frac{g_0 ( a \mid x )}{g_n ( a \mid x )} \right\|_{L^4 ( P_{0,X} )} \left\| \mu_{0} ( a, x ) - \mu_{n} ( a,x ) \right\|_{L^4 ( P_{0,X} )} \right\rbrace = o_{P} ( n^{-1/2} )
\end{align*}
and
\begin{align*}
    \left\| \kappa_{n,\mathcal{V}} (X_{\mathcal{V}}) - \kappa_{0,\mathcal{V}} (X_{\mathcal{V}}) \right\|_{L^2(P_{0,X};\mathcal{H})} \left\| \nu_n ( \mathcal{V} ) - \nu_0 ( \mathcal{V} ) \right\|_{L^2 ( P_{0,X} )} = o_{P} ( n^{-1/2} ).
\end{align*}

\subsubsection*{Drift term, $\mathcal{D}_{n}^\mathcal{V}$}

We want to show that $\left\| \mathcal{D}_n^\mathcal{V} \right\|_\mathcal{H} = o_P ( n^{-1/2})$. By Chebyshev's inequality for Hilbert-valued random elements \citep{grenander1963probabilities}, $\left\| \phi_{n}^\mathcal{V} - \phi_{0}^\mathcal{V} \right\|_{L^2 ( P_0; \mathcal{H} )} = o_P ( 1 )$ implies $\left\| \mathcal{D}_n^\mathcal{V} \right\|_\mathcal{H} = o_P ( n^{-1/2})$\footnote{If cross-fitting had been used instead of sample splitting, we could have instead employed Lemma 3 in \cite{Luedtke2024}.}. Therefore, it suffices to show that $\left\| \phi_{n}^\mathcal{V} - \phi_{0}^\mathcal{V} \right\|_{L^2 ( P_0; \mathcal{H} )} = o_P ( 1 )$. We have the following:
\begin{align*}
&\phi_{n}^\mathcal{V} - \phi_{0}^\mathcal{V} \\
&= \kappa_{n,\mathcal{V}} (X_\mathcal{V}) (\cdot) \psi_n ( z ) + \left\lbrace \mathcal{K} ( \cdot, x ) - \kappa_{n,\mathcal{V}} (X_\mathcal{V}) (\cdot) \right\rbrace \nu_n (\mathcal{V}) ( x_{\mathcal{V}} ) - P_n \left\lbrace \mathcal{K} ( \cdot, X ) \nu_n (\mathcal{V}) ( X_\mathcal{V} ) \right\rbrace \\
 &\quad - \left[ \kappa_{0,\mathcal{V}} (X_\mathcal{V}) (\cdot) \psi_0 ( z ) + \left\lbrace \mathcal{K} ( \cdot, x ) - \kappa_{0,\mathcal{V}} (X_\mathcal{V}) (\cdot) \right\rbrace \nu_0 ( \mathcal{V}) ( x_{\mathcal{V}} ) - P_0 \left\lbrace \mathcal{K} ( \cdot, X ) \nu_0 (\mathcal{V}) ( X_\mathcal{V} ) \right\rbrace \right] \\
 &= \kappa_{n,\mathcal{V}} (X_\mathcal{V}) (\cdot) \psi_n ( z ) - \kappa_{0,\mathcal{V}} (X_\mathcal{V}) (\cdot) \psi_n ( z ) + \kappa_{0,\mathcal{V}} (X_\mathcal{V}) (\cdot) \psi_n ( z ) - \kappa_{0,\mathcal{V}} (X_\mathcal{V}) (\cdot) \psi_0 ( z ) \\
 &\quad + \mathcal{K} ( \cdot, x ) \left\lbrace \nu_n (\mathcal{V}) ( x_{\mathcal{V}} ) - \nu_0 (\mathcal{V}) ( x_{\mathcal{V}} ) \right\rbrace \\
 &\quad - \left\lbrace \kappa_{n,\mathcal{V}} (X_\mathcal{V}) (\cdot) \nu_n (\mathcal{V}) ( x_{\mathcal{V}} ) - \kappa_{0,\mathcal{V}} (X_\mathcal{V}) (\cdot)  \nu_n (\mathcal{V}) ( x_{\mathcal{V}} ) + \kappa_{0,\mathcal{V}} (X_\mathcal{V}) (\cdot)  \nu_n (\mathcal{V}) ( x_{\mathcal{V}} ) - \kappa_{0,\mathcal{V}} (X_\mathcal{V}) (\cdot)  \nu_0 (\mathcal{V}) ( x_{\mathcal{V}} ) \right\rbrace \\
 &\quad - \left[ P_n \left\lbrace \mathcal{K} ( \cdot, X ) \nu_n (\mathcal{V}) ( X_\mathcal{V} ) \right\rbrace - P_n \left\lbrace \mathcal{K} ( \cdot, X ) \nu_0 (\mathcal{V}) ( X_\mathcal{V} ) \right\rbrace + P_n \left\lbrace \mathcal{K} ( \cdot, X ) \nu_0 (\mathcal{V}) ( X_\mathcal{V} ) \right\rbrace \right.\\
 &\quad \quad \left. - P_0 \left\lbrace \mathcal{K} ( \cdot, X ) \nu_0 (\mathcal{V}) ( X_\mathcal{V} ) \right\rbrace \right] \\
 &= \left( \kappa_{n,\mathcal{V}} (X_\mathcal{V}) (\cdot) - \kappa_{0,\mathcal{V}} (X_\mathcal{V}) (\cdot) \right) \psi_n ( z ) + \kappa_{0,\mathcal{V}} (X_\mathcal{V}) (\cdot) \left\lbrace \psi_n ( z ) - \psi_0 ( z ) \right\rbrace + \mathcal{K} ( \cdot, x ) \left\lbrace \nu_n (\mathcal{V}) ( x_{\mathcal{V}} ) - \nu_0 (\mathcal{V}) ( x_{\mathcal{V}} ) \right\rbrace \\
 &\quad - \nu_n (\mathcal{V}) ( x_{\mathcal{V}} ) \left( \kappa_{n,\mathcal{V}} (X_\mathcal{V}) (\cdot) - \kappa_{0,\mathcal{V}} (X_\mathcal{V}) (\cdot) \right) - \kappa_{0,\mathcal{V}} (X_\mathcal{V}) (\cdot) \left\lbrace \nu_n (\mathcal{V}) ( x_{\mathcal{V}} ) - \nu_0 (\mathcal{V}) ( x_{\mathcal{V}} ) \right\rbrace \\
 &\quad - P_n \left[ \mathcal{K} ( \cdot, X ) \left\lbrace \nu_n (\mathcal{V}) ( X_\mathcal{V} ) - \nu_0 (\mathcal{V}) ( X_\mathcal{V} ) \right\rbrace \right] - \left[ P_n \left\lbrace \mathcal{K} ( \cdot, X ) \nu_0 (\mathcal{V}) ( X_\mathcal{V} ) \right\rbrace - P_0 \left\lbrace \mathcal{K} ( \cdot, X ) \nu_0 (\mathcal{V}) ( X_\mathcal{V} ) \right\rbrace \right].
\end{align*}
Therefore, by the triangle inequality, we have:
\begin{align*}
 &\left\| \phi_{n}^\mathcal{V} - \phi_{0}^\mathcal{V} \right\|_{L^2(P_0;\mathcal{H})} \\
 &\leq \left\| \left( \kappa_{n,\mathcal{V}} (X_\mathcal{V}) (\cdot) - \kappa_{0,\mathcal{V}} (X_\mathcal{V}) (\cdot) \right) \psi_n \right\|_{L^2(P_0;\mathcal{H})} + \left\| \kappa_{0,\mathcal{V}} (X_\mathcal{V}) (\cdot) \left( \psi_n - \psi_0 \right) \right\|_{L^2(P_0;\mathcal{H})} \\
 &\quad + \left\| \mathcal{K} ( \cdot, x ) \left\lbrace \nu_n (\mathcal{V}) - \nu_0 (\mathcal{V}) \right\rbrace \right\|_{L^2(P_0;\mathcal{H})} + \left\| \nu_n (\mathcal{V}) \left( \kappa_{n,\mathcal{V}} (X_\mathcal{V}) (\cdot) - \kappa_{0,\mathcal{V}} (X_\mathcal{V}) (\cdot) \right) \right\|_{L^2(P_0;\mathcal{H})} \\
 &\quad + \left\| \kappa_{0,\mathcal{V}} (X_\mathcal{V}) (\cdot) \left\lbrace \nu_n (\mathcal{V}) - \nu_0 (\mathcal{V}) \right\rbrace \right\|_{L^2(P_0;\mathcal{H})} + \left\| P_n \left[ \mathcal{K} ( \cdot, X ) \left\lbrace \nu_n (\mathcal{V}) ( X_\mathcal{V} ) - \nu_0 (\mathcal{V}) ( X_\mathcal{V} ) \right\rbrace \right] \right\|_{\mathcal{H}} \\
 &\quad + \left\| ( P_n - P_0 ) \left\lbrace \mathcal{K} ( \cdot, X ) \nu_0 (\mathcal{V}) ( X_\mathcal{V} ) \right\rbrace \right\|_{\mathcal{H}} \\
 &\leq \left\| \kappa_{n,\mathcal{V}} (X_\mathcal{V}) (\cdot) - \kappa_{0,\mathcal{V}} (X_\mathcal{V}) (\cdot) \right\|_{L^4(P_0;\mathcal{H})} \left\| \psi_n \right\|_{L^4(P_0)} + \left\| \kappa_{0,\mathcal{V}} (X_\mathcal{V}) (\cdot) \right\|_{L^\infty(P_0;\mathcal{H})} \left\| \psi_n - \psi_0 \right\|_{L^2(P_0)} \\
  &\quad + \left\| \mathcal{K} ( \cdot, x ) \right\|_{L^\infty(P_0;\mathcal{H})} \left\| \nu_n (\mathcal{V}) - \nu_0 (\mathcal{V}) \right\|_{L^2(P_0)} + \left\| \nu_n (\mathcal{V}) \right\|_{L^4(P_0)}  \left\| \kappa_{n,\mathcal{V}} (X_\mathcal{V}) (\cdot) - \kappa_{0,\mathcal{V}} (X_\mathcal{V}) (\cdot) \right\|_{L^4(P_0;\mathcal{H})} \\
   &\quad + \left\| \kappa_{0,\mathcal{V}} (X_\mathcal{V}) (\cdot) \right\|_{L^\infty(P_0;\mathcal{H})} \left\| \nu_n (\mathcal{V}) - \nu_0 (\mathcal{V}) \right\|_{L^2(P_0)} + \left\| \mathcal{K} ( \cdot, x ) \right\|_{L^\infty(P_0;\mathcal{H})} \left\| \nu_n (\mathcal{V}) - \nu_0 (\mathcal{V}) \right\|_{L^2(P_0)} \\
   &\quad + \left\| (P_n - P_0)  \left[ \mathcal{K} ( \cdot, X ) \left\lbrace \nu_n (\mathcal{V}) ( X_\mathcal{V} ) - \nu_0 (\mathcal{V}) ( X_\mathcal{V} ) \right\rbrace \right] \right\|_{\mathcal{H}} + \left\| ( P_n - P_0 ) \left\lbrace \mathcal{K} ( \cdot, X ) \nu_0 (\mathcal{V}) ( X_\mathcal{V} ) \right\rbrace \right\|_{\mathcal{H}}.
\end{align*}
The final bound follows from H\"older’s, Jensen’s, and the triangle inequalities. The first term vanishes by \ref{cond:consistentkappa} and \ref{cond:consistentpsi}. The second term vanishes by \ref{cond:bounded} and \ref{cond:consistentpsi}. The third, fifth and sixth vanish by \ref{cond:bounded} and \ref{cond:consistentnu}. The fourth vanishes by \ref{cond:consistentkappa} and \ref{cond:consistentnu}. Since we are using sample-splitting, the seventh and eighth terms vanish by the law of large numbers for triangular arrays. Therefore, $\left\| \phi_{n}^\mathcal{V} - \phi_{0}^\mathcal{V} \right\|_{L^2(P_0;\mathcal{H})} = o_P (1)$.

\end{proof}

\begin{proof}[Proof of Theorem \ref{Theorem3}]
This result follows from Theorem 3 in \cite{Luedtke2024}.
\end{proof}

\begin{proof}[Proof of Lemma \ref{lemma2}]
This follows directly from Theorem \ref{Theorem2} by an application of the delta method. 
\end{proof}

\begin{proof}[Proof of Theorem~\ref{Theorem:confint}]

    We first show that---under \ref{cond:bounded}-\ref{cond:finitesecondmoment}---$\sigma_\omega^2 = P_0 \left[ \langle \phi_0^\omega, \gamma_\omega^\mathcal{K} (P_0) \rangle_\mathcal{H}^2 \right] =0$ if and only if  $\gamma_\omega^\mathcal{K} (P_0)=0$. If $\gamma_\omega^\mathcal{K} (P_0)=0$, then $\langle \phi_0^\omega (Z), \gamma_\omega^\mathcal{K} (P_0) \rangle_\mathcal{H} = 0$, and therefore $\sigma_\omega^2=0$. 

    Now we show that $\sigma_\omega^2=0$ implies $\gamma_\omega^\mathcal{K} (P_0)=0$:
    \begin{align*}
        &P_0 \left[ \langle \phi_0^\omega, \gamma_\omega^\mathcal{K} (P_0) \rangle_\mathcal{H}^2 \right] = P_0 \left[ \left\langle \left( \phi_0^\omega \otimes \phi_0^\omega \right) \gamma_\omega^\mathcal{K} (P_0), \gamma_\omega^\mathcal{K} (P_0) \right\rangle_\mathcal{H} \right] = \left\langle \Sigma \gamma_\omega^\mathcal{K} (P_0), \gamma_\omega^\mathcal{K} (P_0) \right\rangle_\mathcal{H},
    \end{align*}
    where $\Sigma \coloneqq P_0 \left[ \phi_0^\omega \otimes \phi_0^\omega \right]$ is the covariance operator of $\phi_0^\omega (Z)$. The first equality follows by the definition of a tensor product. Since $\Sigma$ it is strictly positive definite by condition \ref{cond:spdcovariance}, therefore $\sigma_\omega^2=0$ implies $\gamma_\omega^\mathcal{K} (P_0)$ is equal zero. Therefore, we have shown that---under stated conditions---$\sigma_\omega^2 = P_0 \left[ \langle \phi_0^\omega, \gamma_\omega^\mathcal{K} (P_0) \rangle_\mathcal{H}^2 \right] =0$ if and only if $\gamma_\omega^\mathcal{K} (P_0)=0$.
    
    We then separately consider two cases to establish coverage of the union confidence set, $\mathcal{C}_{\omega,n}^{>0}\cup \mathcal{C}_{\omega,n}^{0}$. In the first, the variable importance measure is nonzero, so that $\sigma_\omega^2>0$. In this case, Lemma~\ref{lemma2} together with the condition \ref{cond:consistentest}, shows $\mathcal{C}_{\omega,n}^{>0}$ covers the variable importance measure at the desired asymptotic level. When the variable importance measure is $0$, the latter set in the union, $\mathcal{C}_{\omega,n}^{0}$, will cover the variable importance measure at the desired level; this follows since $\mathcal{C}_{\omega,n}^{0}$ contains $0$ if and only if the test of no importance from Section~\ref{Section3.3} fails to reject the null. 
\end{proof}

\begin{proof}[Proof of Theorem \ref{Theorem5}]
Suppose $P_0$ is such that $\gamma_{\omega}^{\mathcal{K}}(P_0) = \gamma_1 \neq 0$. This corresponds to a fixed alternative. Since $\hat{\gamma}_{\omega}^\mathcal{K} \xrightarrow{P} \gamma_1$, and the confidence set $\mathcal{C}_n^\omega(\hat{\xi}_{\omega})$ is centered at $\hat{\gamma}_{\omega}^\mathcal{K}$ and shrinks at a rate depending on $n$, we have:
\[
\text{diam}\left( \mathcal{C}_n^\omega(\hat{\xi}_{\omega}) \right) = o_P(1).
\]
Hence, for any $\varepsilon > 0$, there exists $N$ such that for all $n > N$, with high probability, the confidence set $\mathcal{C}_n^\omega(\hat{\xi}_{\omega})$ lies entirely within a ball of radius $\varepsilon$ centered at $\gamma_1$. Since $\gamma_1 \ne 0$, we can choose $\varepsilon$ small enough that $0 \notin B(\gamma_1, \varepsilon)$. Therefore, for large $n$,
\begin{align*}
    P_{0} \left( 0 \in \mathcal{C}_n^\omega(\hat{\xi}_{\omega}) \right) \to 0,
\end{align*}
which implies
\begin{align*}
    P_{0} \left( 0 \notin \mathcal{C}_n^\omega(\hat{\xi}_{\omega}) \right) \to 1.
\end{align*}
Hence, the test rejects $H_0$ with probability tending to one under a fixed alternative.
\end{proof}

\begin{proof}[Proof of Theorem \ref{Theorem4}]

This result follows from Theorem 6 in \cite{Luedtke2024}.

\end{proof}

\section{Supplementary results from the empirical example}\label{app:resultsempirical}

Figure \ref{figure4} displays Day 43 $\log_{10}$ neutralizing antibody titers against BA.1 by country and treatment (Panel A), and by baseline $\log_{10}$ neutralizing antibody titer against BA.1 and treatment (Panel B). Panel B suggests that individuals with undetectable baseline titers have greater potential for a vaccine-induced increase in neutralizing antibody titers than those with high baseline titers prior to vaccination.

\begin{figure}[tb]
\begin{center}
 \includegraphics[width=\textwidth]{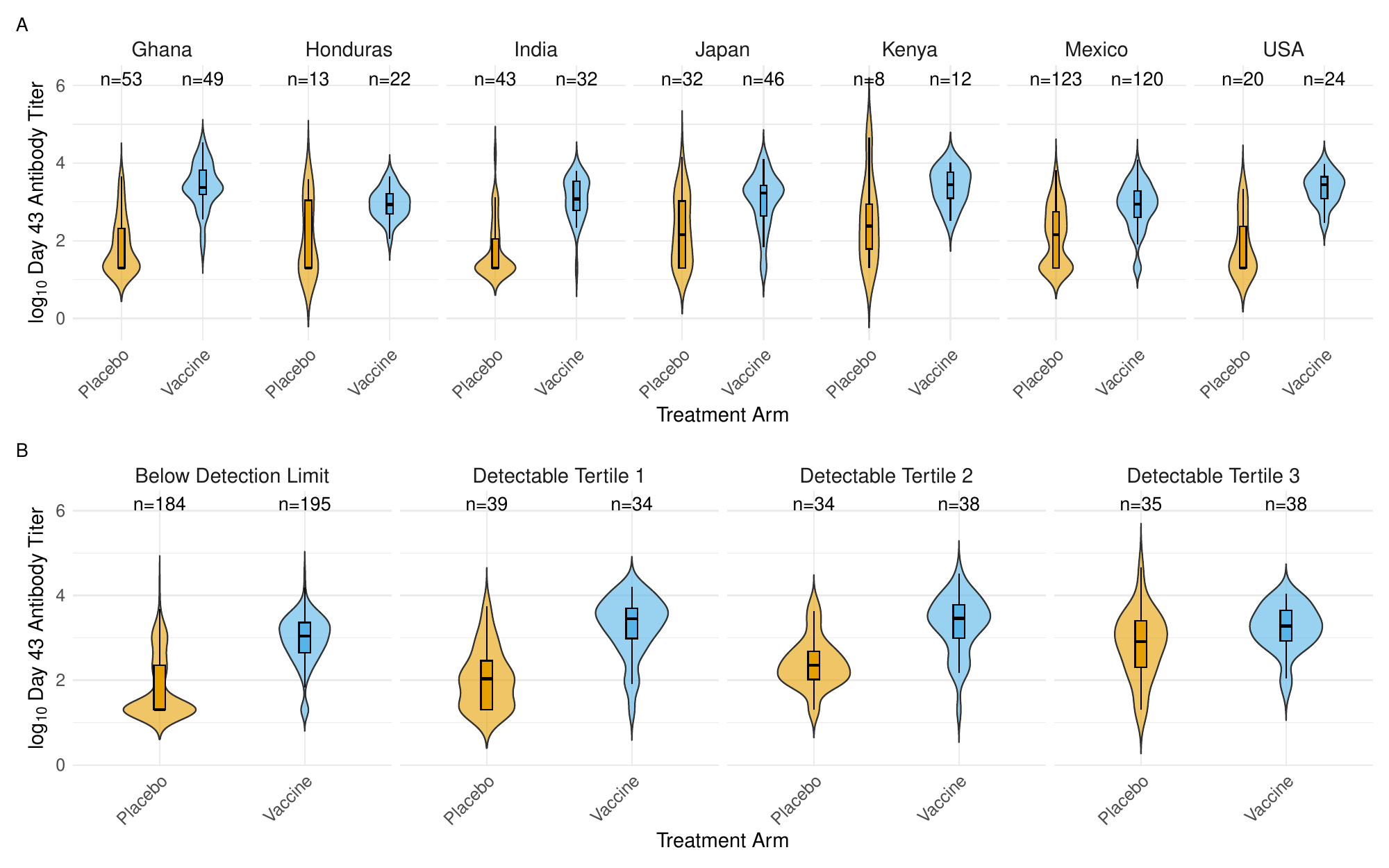}    
\end{center}
\caption{Panel A displays Day 43 $\log_{10}$ neutralizing antibody titer against BA.1 by country and treatment. Panel B displays Day 43 $\log_{10}$ neutralizing antibody titer against BA.1 by baseline titer and treatment.} \label{figure4}
\end{figure}

\begin{table}[tb]
\centering
\caption{Variable importance---defined as the norm of the RKHS embedding of the local variable importance measures---with standard errors and p-values, including Benjamini-Hochberg-adjusted p-values.} \label{table1}
\begin{adjustbox}{max width=0.95\textwidth}
\begin{tabular}{
  l
  l
  S[table-format=1.4]  
  S[table-format=1.4]  
  S[table-format=1.4]  
  S[table-format=1.4]  
}
\toprule
\textbf{Variable} & \textbf{Method} & \textbf{Var. Imp.} & \textbf{St. Err.} & \textbf{p-value} & \textbf{BH-adjusted p-value} \\
\midrule
\multirow{3}{*}{Age} 
  & LOO     & 0.0468 & 0.0788 & 0.5172 & 0.6206 \\
  & KOI     & 0.0593 & 0.0823 & 0.2532 & 0.3038 \\
  & Shapley & 0.0503 & 0.0722 & 0.3182 & 0.3818 \\
\midrule
\multirow{3}{*}{BA.1 NAb} 
  & LOO     & 0.0046 & 0.0053 & 0.2568 & 0.5136 \\
  & KOI     & 0.0290 & 0.0192 & \textbf{0.0028} & \textbf{0.0084} \\
  & Shapley & 0.0104 & 0.0086 & \textbf{0.0166} & \textbf{0.0498} \\
\midrule
\multirow{3}{*}{BMI} 
  & LOO     & 0.0252 & 0.0317 & 0.2538 & 0.5136 \\
  & KOI     & 0.0092 & 0.0396 & 0.9214 & 0.9214 \\
  & Shapley & 0.0127 & 0.0293 & 0.8046 & 0.8046 \\
\midrule
\multirow{3}{*}{Country} 
  & LOO     & 0.0268 & 0.0292 & 0.0920 & 0.5136 \\
  & KOI     & 0.0477 & 0.0266 & \textbf{0.0000} & \textbf{0.0000} \\
  & Shapley & 0.0338 & 0.0238 & \textbf{0.0028} & \textbf{0.0168} \\
\midrule
\multirow{3}{*}{S-IgG (Omicron)} 
  & LOO     & 0.0051 & 0.0064 & 0.4212 & 0.6206 \\
  & KOI     & 0.0111 & 0.0154 & 0.1682 & 0.2523 \\
  & Shapley & 0.0049 & 0.0071 & 0.1996 & 0.2994 \\
\midrule 
\multirow{3}{*}{Sex} 
  & LOO     & 0.0045 & 0.0061 & 0.6736 & 0.6736 \\
  & KOI     & 0.0071 & 0.0093 & 0.1436 & 0.2523 \\
  & Shapley & 0.0037 & 0.0050 & 0.1802 & 0.2994 \\
\bottomrule
\end{tabular}
\end{adjustbox}
\end{table}

\section{Supplementary results from the numerical experiments}\label{app:resultssimulations}

\subsection{Experiment 3: 3-dimensional covariate} \label{app:experiment3}

The observed data consist of $(X, A, Y)$, where $X \coloneqq (X_1, X_2, X_3)$ is a covariate vector sampled from a Gaussian copula with covariance matrix $\Sigma$. The off-diagonal elements of $\Sigma$ are set to $\sigma \in \lbrace 0, 0.4, 0.8 \rbrace$. The treatment satisfies $A| X\sim \textnormal{Bernoulli}(\mathrm{expit}[-0.4 X_1 + 0.1 X_1 X_2])$, where $\mathrm{expit}(x) = 1 / \lbrace 1 + \exp(-x) \rbrace$. The outcome satisfies $Y|A,X\sim \mathcal{N}(X_1 X_2 + 2 X_2^2 - X_1 + A\tau, 1)$, where the CATE is defined as $\tau = \beta g(X_1)$, with $\beta \in \lbrace 0, 1, 5 \rbrace$. We consider two kinds of alternatives: smooth alternatives, where $g(X_1) = X_1$, and rough alternatives, where $g(X_1) = \sin(5 \pi X_1)$.

Figure \ref{figure2c} presents the results of the simulation study. The left column shows the empirical Type I error rates under the null hypothesis of no importance for $X_1$ (i.e., $\beta = 0$), demonstrating that our method effectively controls Type I error across various values of the correlation parameter $\sigma$. The middle column displays the empirical rejection probabilities under the smooth alternative ($g(X_1) = X_1$) for $\beta = 1$ and $\beta = 5$, while the right column presents the corresponding results under the rough alternative ($g(X_1) = \sin(5 \pi X_1)$). Although our method exhibits low power for small sample sizes, its power improves noticeably as the sample size increases.

An important distinction between this simulation setting and Experiments 1 and 2 in the main text is that here, $\tau$ is a function of $X_1$ only, whereas in Experiments 1 and 2, it depended on multiple covariates. As a result, under the setting $\beta = 0$ and $\sigma \in \lbrace 0.4, 0.8 \rbrace$, the variable importance measures KOI and Shapley values correspond to a null scenario. In contrast, in Experiments 1 and 2, the same parameter settings represented an alternative scenario because $\tau$ depended on additional covariates, so correlation between $X_1$ and those covariates induced nonzero variable importance for $X_1$ under KOI and Shapley measures.

\begin{figure}[H]
\begin{center}
 \includegraphics[width=\textwidth]{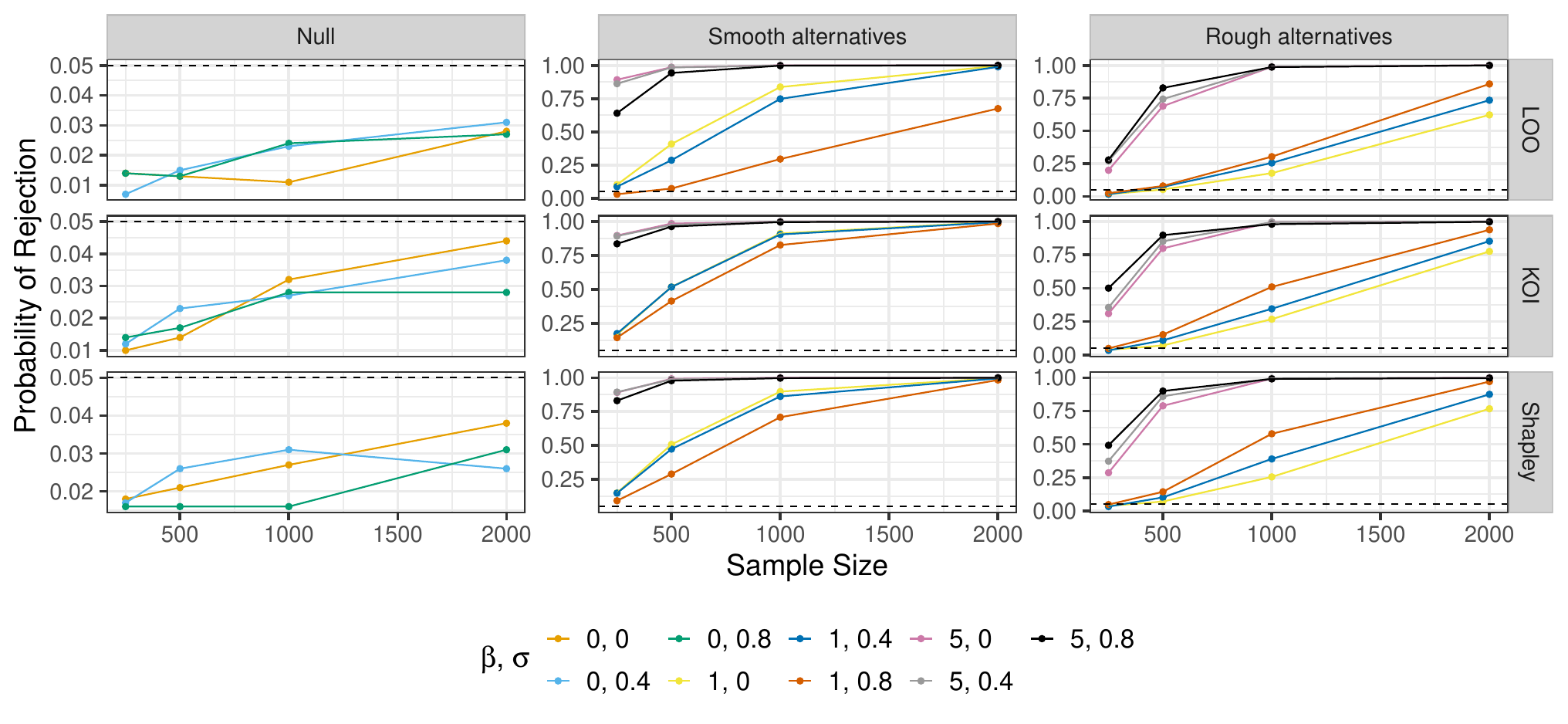}  
\end{center}
\caption{Simulation results from Experiment 3 comparing different variable importance measures (rows) across different data-generating processes (columns).} \label{figure2c}
\end{figure}

\subsection{Computational considerations for approximating Shapley values} \label{app:computationalconsiderations}

Due to their combinatorial nature, Shapley values are computationally expensive to calculate exactly, and quickly become infeasible as the number of covariates increases. We discuss practical considerations and provide recommendations for approximating Shapley value estimators in a computationally efficient manner.

To implement the test for zero importance based on Shapley values in our numerical experiments, we needed to compute two key expressions---the test statistic (\ref{eqn15}) and the corresponding bootstrap-based distribution (\ref{eqn16}). It was crucial to derive concise formulations involving as few matrix and vector multiplications as possible, and to execute these operations in an optimized order to avoid handling large matrices wherever feasible. Since our implementation was carried out in \texttt{R}, special care was taken to eliminate for-loops and instead express all computations using matrix and vector operations. Alternatively, these procedures could be implemented in faster languages such as \texttt{Python}, \texttt{Julia}, or \texttt{C++} for improved performance. Furthermore, we approximated the Shapley values via Monte Carlo sampling, by randomly drawing permutations of covariates. The error due to Monte Carlo sampling is of order $O(m^{-1/2})$, where $m$ is the number of samples. Consequently, the total approximation error is of order $O(m^{-1/2}) + o(n^{-1/2})$. Ideally, $m$ should be chosen larger than $n$ to ensure the sampling error remains negligible compared to the estimation error.

\end{document}